\newif\ifcr
\crtrue
\crfalse 

\documentclass[sigconf]{aamas}

\usepackage{balance} %

\usepackage{amsmath}
\usepackage{amsthm}
\usepackage{mathtools}
\usepackage{amsfonts}
\usepackage{paralist}
\usepackage{multirow}
\usepackage{dsfont}

\usepackage[textsize=tiny,textwidth=1.5cm,linecolor=green!70!black, backgroundcolor=green!10, bordercolor=black,disable]{todonotes} 

\usepackage{tikz}
\usetikzlibrary{decorations,arrows,arrows.meta,petri,topaths,backgrounds,shapes,positioning,fit,calc,decorations.pathreplacing,patterns,intersections}
\pgfdeclarelayer{bg}    %
\pgfdeclarelayer{fg}    %
\pgfsetlayers{bg,main,fg}  %

\tikzstyle{unode} = [draw, circle, inner sep=1.7pt, fill=black]
\tikzstyle{wnode} = [draw=gray, rectangle, inner sep=1.8pt, fill=gray]
\tikzstyle{edc} = [arrows = {-Stealth[length=6pt, inset=2pt]}, thick]
\tikzstyle{labelst} = [fill=white, inner sep = 1pt]
\tikzstyle{fc} = [blue!60!black, dotted, line width=1pt]
\tikzstyle{ec} = [red!60!black, line width=1pt]

\newtheorem{theorem}{Theorem}

\newtheorem{definition}{Definition}
\newtheorem{example}{Example}
\newtheorem{lemma}{Lemma}
\newtheorem{observation}{Observation}

\newtheorem{claim}{Claim}[theorem]
\newtheorem{remark}{Remark}

\newtheorem{appclaim}{Claim}[subsection] 
\newtheorem{appobservation}{Observation}[subsection] 
\newtheorem{applemma}{Lemma}[subsection] 

\newcommand{\mypara}[1]{
\smallskip
\noindent \textbf{#1}%
}

\usepackage{thm-restate}
\usepackage[title]{appendix}

\usepackage[vlined,ruled,linesnumbered]{algorithm2e}

\SetKw{KwAnd}{and}
\SetKw{KwOr}{or}
\SetKw{KwNo}{no}
\SetKw{KwYes}{yes}
\SetKw{KwT}{true}
\SetKw{KwF}{false}
\SetKw{KwST}{s.t.}
\SetKwInOut{Input}{Input}
\SetKwInOut{Output}{Output}
\SetKwProg{Def}{def}{:}{}
\SetKw{Break}{break}

\newcommand{\lineref}[1]{line~\ref{#1}}

\usepackage[sort&compress,nameinlink,capitalize]{cleveref}

\crefname{algorithm}{Algorithm}{Algorithms}
\crefname{lemma}{Lemma}{Lemmas}
\crefname{proposition}{Proposition}{Propositions}
\crefname{observation}{Observation}{Observations}
\crefname{figure}{Figure}{Figures}
\crefname{theorem}{Theorem}{Theorem}
\crefname{claim}{Claim}{Claims}
\crefname{table}{Table}{Tables}
\crefname{open}{Open question}{Open questions}
\crefname{line}{line}{lines}

\crefname{appclaim}{Claim}{Claims}
\crefname{applemma}{Lemma}{Lemmas}
\crefname{observation}{Observation}{Observations}

\usepackage{etoolbox}

\doi{WVYI7850}

\makeatletter
\gdef\@copyrightpermission{
  \begin{minipage}{0.2\columnwidth}
   \href{https://creativecommons.org/licenses/by/4.0/}{\includegraphics[width=0.90\textwidth]{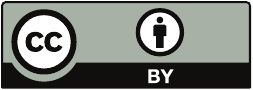}}
  \end{minipage}\hfill
  \begin{minipage}{0.8\columnwidth}
   \href{https://creativecommons.org/licenses/by/4.0/}{This work is licensed under a Creative Commons Attribution International 4.0 License.}
  \end{minipage}
  \vspace{5pt}
}
\makeatother

\setcopyright{ifaamas}
\acmConference[AAMAS '26]{Proc.\@ of the 25th International Conference
on Autonomous Agents and Multiagent Systems (AAMAS 2026)}{May 25 -- 29, 2026}
{Paphos, Cyprus}{C.~Amato, L.~Dennis, V.~Mascardi, J.~Thangarajah (eds.)}
\copyrightyear{2026}
\acmYear{2026}
\acmDOI{}
\acmPrice{}
\acmISBN{}

\title[Control in Hedonic Games]{Control in Hedonic Games}

\newcommand{\mytitle}{Control in Hedonic Games}
\newcommand{\appendixtitle}{Supplementary Material for the Paper ``\mytitle''}
\title{\mytitle}

\author{Jiehua Chen}
\affiliation{
  \institution{TU Wien}
  \city{Vienna}
  \country{Austria}}
\email{jiehua.chen@ac.tuwien.ac.at}

\author{Jakob Guttmann}
\affiliation{
  \institution{TU Wien}
  \city{Vienna}
  \country{Austria}}
\email{e11810289@student.tuwien.ac.at}

\author{Merisa Mustajba\v{s}i\'{c}}
\affiliation{
  \institution{TU Wien}
  \city{Vienna}
  \country{Austria}}
\email{e12450330@student.tuwien.ac.at}

\author{Sofia Simola}
\affiliation{
  \institution{TU Wien}
  \city{Vienna}
  \country{Austria}}
\email{ssimola@ac.tuwien.ac.at}

\begin{abstract}
  We initiate the study of control in hedonic games, where an external actor influences coalition formation by adding or deleting agents.
  We consider three basic control goals 
  \begin{inparaenum}[(1)]
    \item enforcing that an agent is \emph{not alone} (\myemph{\NA});
    \item enforcing that a \emph{pair} of agents is in the same coalition (\myemph{\PA});
    \item enforcing that all agents are in the \emph{same grand} coalition (\myemph{\GR}),
  \end{inparaenum}
  combined with two control actions: adding agents~(\emph{\AddAg}) or deleting agents~(\emph{\DelAg}). %
  We analyze these problems for friend-oriented and additive preferences under individual rationality, individual stability, Nash stability, and core stability.
  We provide a complete computational complexity classification for control in hedonic games.
\end{abstract}


\newcommand{\probname}[1]{{\normalfont\textsc{#1}}}

\newcommand{\decprob}[3]{
  \begin{center}
    \begin{compactitem}[d]
      \item[\probname{#1}]
      \item[\textbf{Input:}]  #2\\[0.2ex]
      \item[\textbf{Question:}]  #3
    \end{compactitem}
  \end{center}
}

\newcommand{\probdef}[3]{\decprob{#1}{#2}{#3}}

\newcommand{\pa}{PA} %
\newcommand{\irt}{individually rational} %
\newcommand{\ist}{individually stable} %
\newcommand{\cst}{core stable} %
\newcommand{\nst}{Nash stable} %
\newcommand{\no}{No-instance} %
\newcommand{\yes}{Yes-instance} %

\definecolor{darkyellow}{RGB}{204, 153, 0}

\newcommand{\goodG}{\ensuremath{\mathcal{F}}}

\newcommand{\utilG}{\ensuremath{\mathcal{G}}}

\newcommand{\additive}{additively separable}

\definecolor{ForestGreen}{rgb}{0.13,0.55,0.13}
\definecolor{BrickRed}{rgb}{0.8,0.25,0.33}
\definecolor{BurntOrange}{rgb}{0.8,0.33,0.0}
\definecolor{Plum}{rgb}{0.56,0.27,0.52}

\newcommand{\todoJ}[1]{\todo[linecolor=green!70!black, backgroundcolor=orange!30]{J: #1}}

\newcommand{\friends}{friend-oriented}

\newcommand{\DAG}{\ensuremath{\mathsf{DAG}}} %
\newcommand{\SYM}{\ensuremath{\mathsf{SYM}}} %
\newcommand{\DTab}{{\ensuremath{^{\mathsf{DAG}}}}}
\newcommand{\STab}{{\ensuremath{^{\mathsf{SYM}}}}}

\def\A{\ensuremath{\mathcal{A}}}
\def\G{\ensuremath{\mathcal{G}}}
\def\S{\ensuremath{\mathsf{S}}}

\def\M{\ensuremath{\mathsf{{M}}}}
\def\AddAg{\ensuremath{\mathsf{AddAg}}}
\def\DelAg{\ensuremath{\mathsf{DelAg}}}
\def\NA{{\color{red!60!black}\ensuremath{\mathsf{NA}}}}
\def\PA{{\color{blue!70!black}\ensuremath{\mathsf{PA}}}}
\def\GR{{\color{green!50!black}\ensuremath{\mathsf{GR}}}}
\def\FA{\textsc{FriHG}}
\def\AD{\textsc{AddHG}}
\def\IR{\ensuremath{\mathsf{IR}}}
\def\IS{\ensuremath{\mathsf{IS}}}
\def\CS{\ensuremath{\mathsf{CS}}}
\def\NS{\ensuremath{\mathsf{NS}}}

\newcommand{\FE}{\FA}
\newcommand{\ADD}{\AD}

\def\NPc{{\ensuremath{\mathsf{NPc}}}}
\def\NPcour{{\color{red!80!black}\ensuremath{\boldsymbol{\mathsf{NPc}}}}}
\def\NP{{\ensuremath{\mathsf{NP}}}}
\def\PPP{{\ensuremath{\mathsf{P}}}}

\def\coNP{{\ensuremath{\mathsf{coNP}}}}

\def\coNPcour{{\color{red!80!black}\ensuremath{\boldsymbol{\mathsf{coNPc}}}}}
\def\Wtwo{\ensuremath{\mathsf{W[2]}}}
\def\Wtwoh{{\color{purple}\ensuremath{\mathsf{W[2]h}}}}
\def\Wtwohour{{\color{purple}\ensuremath{\boldsymbol{\mathsf{W[2]h}}}}}
\def\XP{\ensuremath{\mathsf{XP}}}
\def\XPour{\ensuremath{\boldsymbol{\mathsf{XP}}}}
\def\sigmac{\ensuremath{\Sigma^P_2{c}}}
\def\sigmatwop{\ensuremath{\Sigma^P_2}}
\def\PP{{\color{green!60!black}\ensuremath{\mathsf{P}}}}
\def\PPour{{\color{green!60!black}\ensuremath{\boldsymbol{\mathsf{P}}}}}
\def\IMM{{\color{orange!80!black}\ensuremath{\mathsf{imm}}}}
\def\IMMh{{\color{orange!80!black}\ensuremath{\mathsf{imm}^{\star}}}}
\def\IMMour{{\color{orange!80!black}\ensuremath{\boldsymbol{\mathsf{imm}}}}}
\def\IMMhour{{\color{orange!80!black}\ensuremath{\boldsymbol{\mathsf{imm}^{\star}}}}}
\def\NEVER{\ensuremath{\boldsymbol{\mathsf{never}}}}

\newcommand{\immune}{immune}

\newcommand{\never}{never}

\newcommand{\agents}{\ensuremath{\mathcal{V}}}
\newcommand{\agentsU}{\ensuremath{\mathcal{U}}}
\newcommand{\agentsW}{\ensuremath{\mathcal{W}}}
\newcommand{\ag}[1]{\ensuremath{{#1}}}
\newcommand{\util}[1]{\ensuremath{\mu}_{#1}}

\newcommand{\weight}{\ensuremath{\omega}}

\newcommand{\budget}{\ensuremath{k}}

\newcommand{\agx}{\ag{x}}

\newcommand{\agi}{\ag{i}}
\newcommand{\agy}{\ag{y}}

\newcommand{\sss}{\ensuremath{{\agentsW'}}}

\newcommand{\coal}{\ensuremath{A}}
\newcommand{\bcoal}{\ensuremath{B}}

\newcommand{\myemph}[1]{{\color{green!40!black}\emph{#1}}}

\newcommand{\fas}{\ensuremath{\mathsf{f}}}

\newcommand{\RETCf}{\probname{Restricted Exact Cover by 3-Sets}}
\newcommand{\RETC}{\probname{RX3C}}
\newcommand{\scp}{\probname{Set Cover}}
\newcommand{\cliqprob}{\probname{Clique}}
\newcommand{\twoDSN}{\probname{2-DSN}}

\newcommand{\np}{\ensuremath{\mathsf{NP}}}

\newcommand{\els}{\ensuremath{\mathcal{E}}}
\newcommand{\sets}{\ensuremath{\mathcal{S}}}

\newcommand{\set}[1]{\ensuremath{S_{#1}}}
\newcommand{\excov}{\ensuremath{\mathcal{K}}}
\newcommand{\nn}{\ensuremath{\hat{n}}}

\newcommand{\sset}[1]{\ensuremath{s_{#1}}}
\newcommand{\elag}[1]{\ensuremath{u_{#1}}}
\newcommand{\axag}[1]{\ensuremath{y_{#1}}}
\newcommand{\dmag}[2]{\ensuremath{d_{#1}^{#2}}}

\newcommand{\utilP}[1]{\util{#1}(\Pi(#1))}

\newcommand{\controlprob}[4]{\probname{#4-#1-#2-#3}}

\newenvironment{claimproof}[1]{\begin{proof}\renewcommand{\qed}{\hfill (end of the proof of~\cref{#1})~$\diamond$}}{
\end{proof}}

\newcommand{\appendixenvchanges}[1]{
  \let\origclaim\claim
  \let\endorigclaim\endclaim

  \let\origobservation\observation
  \let\endorigobservation\endobservation

  \let\origlemma\lemma
  \let\endoriglemma\endlemma
  
  \renewenvironment{claim}{\begin{appclaim}}{\end{appclaim}}
  \renewenvironment{observation}{\begin{appobservation}}{\end{appobservation}}
  \renewenvironment{lemma}{\begin{applemma}}{\end{applemma}}
    #1
  \let\claim\origclaim
  \let\endclaim\endorigclaim

    \let\observation\origobservation
  \let\endobservation\endorigobservation

      \let\lemma\origlemma
  \let\endlemma\endoriglemma
}

\newcommand{\appsymb}{$\star$}

\usepackage{etoolbox} %

\newcommand{\appendixproofwithstatement}[3]{%
  \gappto{\appendixtext}{
    \subsection{Proof of \cref{#1}}\label{proof:#1}
    #2 \appendixenvchanges{
    \begin{proof}
    #3\end{proof} }
  }
}

\newcommand{\appendixproofwithstatementsketch}[4]{%
\begin{proof}[Proof Sketch]
#3
\end{proof}
  \gappto{\appendixtext}{
    \subsection{Proof of \cref{#1}}\label{proof:#1}
    #2
\appendixenvchanges{    \begin{proof}
    #4\end{proof} }
  }
}

\newcommand{\appendixsection}[1]{%
  \gappto{\appendixtext}{
    \section{Additional material for Section~\ref{#1}}
    \label{appsec:#1}
  }
}

\newcommand{\appendixfigure}[4]{%
 #2
  \gappto{\appendixtext}{
    \subsection{#3}\label{appfigure:#1}
\appendixenvchanges{      #4 }
    }
}

\begin{document}

\pagestyle{fancy}
\fancyhead{}

\maketitle

\section{Introduction}\label{sec:intro}
\todo{TODONOTES ARE ON}
Hedonic games~\cite{Dreze80_Hedonic} are coalition formation games where agents form coalitions based on their preferences over which coalition to join.
The task is to partition agents into disjoint coalitions satisfying certain stability criteria.
Typical stability criteria include \emph{individual rationality} (no agent prefers being alone to his current coalition),
\emph{individual stability} (no agent prefers to join an existing coalition that would accept them),
\emph{Nash stability} (no agent can improve by unilaterally moving to another existing coalition), and \emph{core stability} (no group of agents can all improve by forming a new coalition together).
Since their introduction, hedonic games have become an important research topic in algorithmic game theory~\cite{CEW2011CooperativeGames,BER2016} and computational social choice~\cite{AZ2016HG-Bookchapter}.

Most research has focused on analyzing fixed hedonic games--proposing new stability concepts, investigating existential questions (does a stable partition exist?), and computational questions (verifying stability, finding stable partitions when they exist).
Far less attention has been paid to how external influence--such as adding or deleting agents--can shape outcomes in hedonic games.
Such external influence is formalized as \emph{control} in voting~\cite{Bartholdi92_Control}, where an external actor manipulates elections by adding or deleting voters or candidates.
Control naturally arises in coalition formation settings where external actors shape outcomes: department chairs forming research collaborations, conference organizers balancing working groups, or managers assigning employees to project teams.

In this paper, we introduce and perform a systematic study of control in hedonic games, where an external party influences outcomes by adding or deleting agents.
We formalize control problems combining two actions: adding agents~(\emph{\AddAg}) or deleting agents~(\emph{\DelAg}),
with three goals: ensuring a specific agent is not alone~(\NA), ensuring a specific pair is together~(\PA), or forming a grand coalition~(\GR), all while reaching a stable partition.
These goals capture common control scenarios:
\begin{compactitem}[--]
\item \NA: A company forms project teams. A new employee joins who has not built relationships yet. Without intervention, they might remain isolated and unassigned. The manager can hire another employee with complementary skills, making them willing to form a team together or more attractive to existing teams.
\item \PA: The department chair wants two senior researchers to collaborate. A third, polarizing colleague, cannot work with one of them, which prevents a stable joint assignment. Reassigning this colleague to another project enables a mutually acceptable collaboration between the two senior researchers.
\item \GR: Multiple small research labs work on similar problems but compete rather than collaborate. A funding agency wants to form a single large consortium for a major grant. The agency can offer funding for postdocs who would only join if everyone collaborates, forcing consolidation into a grand coalition.
\end{compactitem}
We analyze these problems for friend-oriented~(\FA) and additive~(\AD) preferences under four stability concepts: individual rationality~(\IR), individual stability~(\IS), Nash stability~(\NS), and core stability~(\CS).

\mypara{Our contributions.}
Besides introducing the control model, we provide a comprehensive complexity picture of control in hedonic games. 
We summarize the key findings below:
\begin{compactitem}[--]
  \item We provide a few polynomial-time algorithms for achieving \NA\ and \PA\ by adding agents in the friend-oriented preference setting.
  The algorithms exploit specific graph structures: For \IR\ (and \IS), we reduce to finding minimum-weight paths and cycles in arc-weighted directed graphs; For \CS, we reduce to finding minimum-weight subgraphs in Steiner networks.
  These positive results show that targeted control goals (ensuring individuals are not isolated or specific pairs collaborate) are computationally tractable via agent addition.
  \item We discover that for most stability concepts, achieving \NA\ or \PA\ by deleting agents is impossible--the problems are \immune; see \cref{def:immune+never}.
  This asymmetry between addition and deletion is somewhat surprising: You can add agents to force desired coalitions, but removing agents rarely helps.
  \item For all intractable (\NP-hard, \coNP-hard, or \sigmatwop-hard) control problems, the base problem--verifying whether the goal already holds without any control--is already intractable.
  This provides a natural barrier against control: If determining whether control is necessary is already hard, executing control attacks is not easier.
  \item For the control goal~\GR, the problem is mostly polynomial-time solvable for constant number~$k$ of control actions (i.e., in \XP\ wrt.~$k$). 
  Interestingly, the complexity picture inverts between preference types: \GR\ is more resistant to control than \NA\ and \PA\ in \FA, but less resistant in \AD.
  This suggests grand coalition formation has different structural properties than targeted pair or individual goals.
\end{compactitem}
Our main results regarding the computational complexity are summarized in Table~\ref{tab:merged-complexity-overview}.

{

\crefname{theorem}{T\hspace*{-0.07cm}}{T}
\crefname{proposition}{P\hspace*{-0.07cm}}{P}
\crefname{observation}{O\hspace*{-0.07cm}}{O}

\newcommand{\addIRNAAGcite}{[\cref{thm:AD-IR-NA}]}
\newcommand{\addIRPAAGcite}{[T\ref{thm:AD-IR-PA},T\ref{thm:SYMMAD-IR-PA}]}
\newcommand{\addIRGRAGcite}{[T\ref{thm:DAG-AD-IRISNS-GR-ADD},T\ref{thm:SYM-AD-IRISNS-GR-ADD}]}
\newcommand{\addIRNADGcite}{[\cref{prop:AD-IR-IMMUNE-DELAG}]}
\newcommand{\addIRPADGcite}{[\cref{prop:AD-IR-IMMUNE-DELAG}]}
\newcommand{\addIRGRDGcite}{[T\ref{thm:DAG-AD-IRISNS-GR-ADD},T\ref{thm:SYM-AD-IRISNS-GR-ADD}]}

\newcommand{\addISNSNAAGcite}{\cite{Sung10_Additive},[\cref{obs:originalhard-follows}]} %
\newcommand{\addISNSPAAGcite}{\cite{Sung10_Additive},[\cref{obs:originalhard-follows}]} %
\newcommand{\addISNSGRAGcite}{[T\ref{thm:DAG-AD-IRISNS-GR-ADD},T\ref{thm:SYM-AD-IRISNS-GR-ADD}]}

\newcommand{\DAGSYMaddISNSNAAGcite}{\DAGaddISNSNAcite} %
\newcommand{\DAGSYMaddISNSPAAGcite}{\DAGaddISNSPAcite} %
\newcommand{\DAGSYMaddISNSGRAGcite}{XXX}

\newcommand{\addISNSNADGcite}{\cite{Sung10_Additive},[\cref{obs:originalhard-follows}]}
\newcommand{\addISNSPADGcite}{\cite{Sung10_Additive},[\cref{obs:originalhard-follows}]}
\newcommand{\addISNSGRDGcite}{[X\cref{thm:DAG-AD-IRISNS-GR-AD}]} %

\newcommand{\DAGSYMaddISNSNADGcite}{\DAGaddISNSPAcite}
\newcommand{\DAGSYMaddISNSPADGcite}{\DAGaddISNSPAcite}
\newcommand{\DAGSYMaddIRISNSGRDGcite}{[T\ref{thm:DAG-AD-IRISNS-GR-ADD},T\ref{thm:SYM-AD-IRISNS-GR-ADD}]} %

\newcommand{\addCSNAAGcite}{\cite{Woeginger13_CoreAdditiveHG,peters2017precise},[\cref{obs:originalhard-follows}]}
\newcommand{\addCSPAAGcite}{\cite{Woeginger13_CoreAdditiveHG,peters2017precise},[\cref{obs:originalhard-follows}]}
\newcommand{\addCSGRAGcite}{[\cref{thm:AD-CS-GS}]} %

\newcommand{\DAGSYMaddCSNAAGcite}{\SYMMaddCSNAAGcite}
\newcommand{\DAGSYMaddCSPAAGcite}{\SYMMaddCSPAAGcite}
\newcommand{\DAGYMaddCSGRAGcite}{[\cref{thm:AD-CS-GS}]} %
\newcommand{\addCSNADGcite}{\cite{Woeginger13_CoreAdditiveHG,peters2017precise},[\cref{obs:originalhard-follows}]}
\newcommand{\addCSPADGcite}{\cite{Woeginger13_CoreAdditiveHG,peters2017precise},[\cref{obs:originalhard-follows}]}
\newcommand{\addCSGRDGcite}{[\cref{thm:AD-CS-GS}]} %

\newcommand{\DAGaddCSNAAGcite}{\DAGaddIRCSNAAGcite}
\newcommand{\DAGaddCSPAAGcite}{\DAGaddIRCSPAAGcite}
\newcommand{\DAGaddCSGRAGcite}{[\cref{thm:DAG-AD-IRISNS-GR-ADD}]}

\newcommand{\DAGaddCSNADGcite}{\DAGaddIRCSNADGcite}
\newcommand{\DAGaddCSPADGcite}{\DAGaddIRCSPADGcite}
\newcommand{\DAGaddCSGRDGcite}{[\cref{thm:DAG-AD-IRISNS-GR-ADD}]}

\newcommand{\friendsIRISCSNAAGcite}{[T\ref{thm:FE-IR-NA-ADD},T\ref{thm:FE-CS-PA-ADD}]} %
\newcommand{\friendsIRISCSPAAGcite}{[T\ref{thm:FE-IR-NA-ADD},T\ref{thm:FE-CS-PA-ADD}]} %
\newcommand{\friendsIRISCSGRAGcite}{\friendsIRGRAGcite} %

\newcommand{\friendsIRISCSNADGcite}{\friendsIRNADGcite} %
\newcommand{\friendsIRISCSPADGcite}{\friendsIRPADGcite} %
\newcommand{\friendsIRISCSGRDGcite}{\friendsIRGRDGcite} %

\newcommand{\friendsNSNAAGcite}{\cite{Brandt24_AI_NashStable},[\cref{obs:originalhard-follows}]}
\newcommand{\friendsNSPAAGcite}{\cite{Brandt24_AI_NashStable},[\cref{obs:originalhard-follows}]}
\newcommand{\friendsNSGRAGcite}{[\cref{thm:FE-IR-GR-ADD}]}

\newcommand{\friendsNSNADGcite}{\cite{Brandt24_AI_NashStable},[\cref{obs:originalhard-follows}]}
\newcommand{\friendsNSPADGcite}{\cite{Brandt24_AI_NashStable},[\cref{obs:originalhard-follows}]}
\newcommand{\friendsNSGRDGcite}{[\cref{prop:FE-IRISNS-GR-DEL}]}

\newcommand{\SYMMfriendsNSNAAGcite}{[\cref{prop:SYM-FE-IRISNS-NA-ADD}]}
\newcommand{\SYMMfriendsNSPAAGcite}{[\cref{prop:SYM-FE-IRISNS-NA-ADD}]}
\newcommand{\SYMMfriendsNSGRAGcite}{[\cref{thm:FE-IR-GR-ADD}]}

\newcommand{\SYMMfriendsNSNADGcite}{[\cref{prop:IMMUNE-DELAG}]}
\newcommand{\SYMMfriendsNSPADGcite}{[\cref{prop:IMMUNE-DELAG}]}
\newcommand{\SYMMfriendsNSGRDGcite}{[\cref{prop:FE-IRISNS-GR-DEL}]}

\newcommand{\DAGfriendscite}{[\cref{obs:DAG-FE-NO}]}

\newcommand{\friendsIRNAAGcite}{[\cref{thm:FE-IR-NA-ADD}]}
\newcommand{\friendsIRPAAGcite}{[\cref{thm:FE-IR-NA-ADD}]}
\newcommand{\friendsIRGRAGcite}{[\cref{thm:FE-IR-GR-ADD}]}

\newcommand{\friendsIRNADGcite}{[\cref{prop:IMMUNE-DELAG}]}
\newcommand{\friendsIRPADGcite}{[\cref{prop:IMMUNE-DELAG}]}
\newcommand{\friendsIRGRDGcite}{[\cref{prop:FE-IRISNS-GR-DEL}]}

\newcommand{\friendsISNAAGcite}{[\cref{thm:FE-IR-NA-ADD}]}
\newcommand{\friendsISPAAGcite}{[\cref{thm:FE-IR-NA-ADD}]}
\newcommand{\friendsISGRAGcite}{[\cref{thm:FE-IR-GR-ADD}]}

\newcommand{\friendsISNADGcite}{[\cref{prop:IMMUNE-DELAG}]}
\newcommand{\friendsISPADGcite}{[\cref{prop:IMMUNE-DELAG}]}
\newcommand{\friendsISGRDGcite}{[\cref{prop:FE-IRISNS-GR-DEL}]}

\newcommand{\friendsCSNAAGcite}{[\cref{thm:FE-CS-PA-ADD}]}
\newcommand{\friendsCSPAAGcite}{[\cref{thm:FE-CS-PA-ADD}]}
\newcommand{\friendsCSGRAGcite}{[\cref{thm:FE-IR-GR-ADD}]}

\newcommand{\friendsCSNADGcite}{[\cref{prop:IMMUNE-DELAG}]}
\newcommand{\friendsCSPADGcite}{[\cref{prop:IMMUNE-DELAG}]}
\newcommand{\friendsCSGRDGcite}{[\cref{prop:FE-IRISNS-GR-DEL}]}

\newcommand{\our}[1]{\textbf{#1}}

\newcommand{\DAGaddIRCSNAAGcite}{[\cref{prop:SYM-AD-IR-NA-ADD}]}
\newcommand{\DAGaddIRCSPAAGcite}{[\cref{thm:AD-IR-PA}]}
\newcommand{\DAGaddIRCSGRAGcite}{[\cref{thm:DAGSYM-AD-IRISNS-GR-ADD}]}

\newcommand{\DAGaddIRCSNADGcite}{[\cref{prop:AD-IR-IMMUNE-DELAG}]}
\newcommand{\DAGaddIRCSPADGcite}{[\cref{prop:AD-IR-IMMUNE-DELAG}]}
\newcommand{\DAGaddIRCSGRDGcite}{[\cref{thm:DAGSYM-AD-IRISNS-GR-ADD}]}

\newcommand{\DAGaddISNSNAcite}{[T\ref{thm:DAGSYM-AD-ISNS-NA-ADD},T\ref{thm:SYM-AD-ISNS-NA-ADD}]}
\newcommand{\DAGaddISNSPAcite}{[T\ref{thm:DAGSYM-AD-ISNS-NA-ADD},T\ref{thm:SYM-AD-ISNS-NA-ADD}]}
\newcommand{\DAGaddISNSGRcite}{[\cref{thm:DAGSYM-AD-IRISNS-GR-ADD}]}

\newcommand{\SYMMfriendsIRNAAGcite}{[\cref{thm:FE-IR-NA-ADD}]}
\newcommand{\SYMMfriendsIRPAAGcite}{[\cref{thm:FE-IR-NA-ADD}]}
\newcommand{\SYMMfriendsIRGRAGcite}{[\cref{thm:FE-IR-GR-ADD}]}

\newcommand{\SYMMfriendsIRNADGcite}{[\cref{prop:IMMUNE-DELAG}]}
\newcommand{\SYMMfriendsIRPADGcite}{[\cref{prop:IMMUNE-DELAG}]}
\newcommand{\SYMMfriendsIRGRDGcite}{[\cref{prop:FE-IRISNS-GR-DEL}]}

\newcommand{\SYMMfriendsISNAAGcite}{[\cref{thm:FE-IR-NA-ADD}]}
\newcommand{\SYMMfriendsISPAAGcite}{[\cref{thm:FE-IR-NA-ADD}]}
\newcommand{\SYMMfriendsISGRAGcite}{[\cref{thm:FE-IR-GR-ADD}]}

\newcommand{\SYMMfriendsISNADGcite}{[\cref{prop:IMMUNE-DELAG}]}
\newcommand{\SYMMfriendsISPADGcite}{[\cref{prop:IMMUNE-DELAG}]}
\newcommand{\SYMMfriendsISGRDGcite}{[\cref{prop:FE-IRISNS-GR-DEL}]}

\newcommand{\SYMMfriendsCSNAAGcite}{[\cref{thm:FE-CS-PA-ADD}]}
\newcommand{\SYMMfriendsCSPAAGcite}{[\cref{thm:FE-CS-PA-ADD}]}
\newcommand{\SYMMfriendsCSGRAGcite}{[\cref{thm:FE-IR-GR-ADD}]}

\newcommand{\SYMMfriendsCSNADGcite}{[\cref{prop:IMMUNE-DELAG}]}
\newcommand{\SYMMfriendsCSPADGcite}{[\cref{prop:IMMUNE-DELAG}]}
\newcommand{\SYMMfriendsCSGRDGcite}{[\cref{prop:FE-IRISNS-GR-DEL}]}

\newcommand{\SYMMaddIRNAAGcite}{[\cref{prop:SYM-AD-IR-NA-ADD}]}
\newcommand{\SYMMaddIRPAAGcite}{[T\ref{thm:AD-IR-PA},T\ref{thm:SYMMAD-IR-PA}]}
\newcommand{\SYMMaddIRGRAGcite}{[\cref{thm:DAGSYM-AD-IRISNS-GR-ADD}]}

\newcommand{\SYMMaddIRNADGcite}{[\cref{prop:AD-IR-IMMUNE-DELAG}]}
\newcommand{\SYMMaddIRPADGcite}{[\cref{prop:AD-IR-IMMUNE-DELAG}]}
\newcommand{\SYMMaddIRGRDGcite}{[\cref{thm:DAGSYM-AD-IRISNS-GR-ADD}]}

\newcommand{\SYMMaddISNAAGcite}{[\cref{thm:DAGSYM-AD-ISNS-NA-ADD}]}
\newcommand{\SYMMaddISPAAGcite}{[\cref{thm:DAGSYM-AD-ISNS-NA-ADD}]}
\newcommand{\SYMMaddISGRAGcite}{[\cref{thm:FE-IR-GR-ADD}]}

\newcommand{\SYMMaddISNADGcite}{[\cref{thm:DAGSYM-AD-ISNS-NA-ADD}]}
\newcommand{\SYMMaddISPADGcite}{[\cref{thm:DAGSYM-AD-ISNS-NA-ADD}]}
\newcommand{\SYMMaddISGRDGcite}{[\cref{thm:DAGSYM-AD-IRISNS-GR-ADD}]}

\newcommand{\SYMMaddNSNAAGcite}{[\cref{thm:DAGSYM-AD-ISNS-NA-ADD}]}
\newcommand{\SYMMaddNSPAAGcite}{[\cref{thm:DAGSYM-AD-ISNS-NA-ADD}]}
\newcommand{\SYMMaddNSGRAGcite}{[\cref{thm:DAGSYM-AD-IRISNS-GR-ADD}]}

\newcommand{\SYMMaddNSNADGcite}{[\cref{thm:DAGSYM-AD-ISNS-NA-ADD}]}
\newcommand{\SYMMaddNSPADGcite}{[TODO]}
\newcommand{\SYMMaddNSGRDGcite}{[\cref{thm:DAGSYM-AD-IRISNS-GR-ADD}]}

\newcommand{\SYMMaddCSNAAGcite}{\cite{peters2017precise}} %
\newcommand{\SYMMaddCSPAAGcite}{\cite{peters2017precise}} %
\newcommand{\SYMMaddCSGRAGcite}{[\cref{thm:AD-CS-GS}]}

\newcommand{\SYMMaddCSNADGcite}{\cite{peters2017precise}} %
\newcommand{\SYMMaddCSPADGcite}{\cite{peters2017precise}} %
\newcommand{\SYMMaddCSGRDGcite}{[\cref{thm:AD-CS-GS}]}

\newcommand{\SYMMaddISNSNAAGcite}{[\cref{thm:DAGSYM-AD-ISNS-NA-ADD}]}
\newcommand{\SYMMaddISNSPAAGcite}{[\cref{thm:DAGSYM-AD-ISNS-NA-ADD}]}
\newcommand{\SYMMaddISNSGRAGcite}{[\cref{thm:FE-IR-GR-ADD}]}

\newcommand{\SYMMaddISNSNADGcite}{[\cref{thm:DAGSYM-AD-ISNS-NA-ADD}]}
\newcommand{\SYMMaddISNSPADGcite}{[\cref{thm:DAGSYM-AD-ISNS-NA-ADD}]}
\newcommand{\SYMMaddISNSGRDGcite}{[\cref{thm:DAGSYM-AD-IRISNS-GR-ADD}]}

\newcommand{\budgetzerosym}{$^\clubsuit$}

\begin{table*}
  \centering
  \small
  \caption{
    Complexity results of control in hedonic games,
    by either adding (\AddAg) or deleting agents (\DelAg),
    with the goals of ensuring that either a given agent is \textbf{n}ot \textbf{a}lone~(\NA),
    or a \textbf{pa}ir of agents is in the same coalition~(\PA),
    or all agents are in the grand coalition~(\GR).
    We study two compact preference representations--additive preferences (\AD) or \friends\ preferences in the friends-and-enemies model (\FA)--and four stability concepts--individual rationality (\IR),
    individual stability~(\IS),
    Nash stability~(\NS),
    and core stability~(\CS).
    All hardness results hold even if no control action is allowed (i.e., $\budget=0$).
    Problems with tag~``\DTab'' (resp.\ ``\STab'') means that the corresponding hardness results hold even if the preference graph is a DAG (resp.\ symmetric). 
    Entries labeled ``\PP'' denote polynomial-time solvability,
    ``\NPc'' NP-completeness,
    and ``\sigmac'' $\Sigma^{P}_2$-completeness.
    ``\IMM'' means \immune\ and ``\IMMh'' means that it is \immune\ while deciding yes-instances with $k=0$ remains \NP-hard. %
    \Wtwoh\ and \XP\ are with respect to the budget $\budget$.
    ``\NEVER'' means that the given instance is always a no instance.
    Results in \textbf{bold} are our contributions.
  }
  \label{tab:merged-complexity-overview}
  {
    \begin{tabular}{@{}c@{} l @{\;\;}
      l@{\,}l @{\;\;}
      l@{\,}l @{\;\;}
      l@{\,}l @{\;\;}
      l@{\,}l @{\;\;}
      l@{\,}l @{\;\;}
      l@{\,}l @{}}
    \toprule
& HG  & \multicolumn{2}{c}{\NA-\AddAg} & \multicolumn{2}{c}{\NA-\DelAg} & \multicolumn{2}{c}{\PA-\AddAg} & \multicolumn{2}{c}{\PA-\DelAg} & \multicolumn{2}{c}{\GR-\AddAg} & \multicolumn{2}{c}{\GR-\DelAg} \\
    \midrule
      & \FA-\{\IR,\IS,\CS\}\STab & \PPour & \friendsIRISCSNAAGcite & \IMMour & \friendsIRISCSNADGcite & \PPour & \friendsIRISCSPAAGcite & \IMMour & \friendsIRISCSPADGcite & \Wtwohour, \XPour & \friendsIRISCSGRAGcite & \PPour & \friendsIRISCSGRDGcite \\\cline{2-14}\\[-2ex]
   & \FA-\NS & \NPc & \friendsNSNAAGcite & \NPc & \friendsNSNADGcite & \NPc & \friendsNSPAAGcite & \NPc & \friendsNSPADGcite & \Wtwohour, \XPour & \friendsNSGRAGcite & \PPour & \friendsNSGRDGcite \\
    & \quad {\scriptsize \SYM} & \PPour & \SYMMfriendsNSNAAGcite & \IMMour & \SYMMfriendsNSNADGcite & \PPour & \SYMMfriendsNSPAAGcite & \IMMour & \SYMMfriendsNSPADGcite & \Wtwohour, \XPour & \SYMMfriendsNSGRAGcite & \PPour & \SYMMfriendsNSGRDGcite \\\cline{2-14}\\[-2ex]
      &{\small \FA-\{\IR,\IS,\NS,\CS\}}\textsuperscript{\scriptsize \DAG}  & \NEVER & \DAGfriendscite & \NEVER & \DAGfriendscite & \NEVER & \DAGfriendscite & \NEVER & \DAGfriendscite & \NEVER & \DAGfriendscite & \NEVER & \DAGfriendscite \\ %
      \midrule
      & \AD-\IR & \NPcour & \addIRNAAGcite  &\IMMhour  & \addIRNADGcite & \NPcour & \addIRPAAGcite & \IMMhour & \addIRPADGcite & \Wtwohour, \XPour &  \addIRGRAGcite & \Wtwohour, \XPour  &  \addIRGRDGcite \\
  & \quad {\scriptsize \DAG/\SYM} & \PPour & \SYMMaddIRNAAGcite & \IMMour & \SYMMaddIRNADGcite & \NPcour & \SYMMaddIRPAAGcite & \IMMhour & \SYMMaddIRPADGcite & \Wtwohour, \XPour & \DAGSYMaddIRISNSGRDGcite & \Wtwohour, \XPour & \DAGSYMaddIRISNSGRDGcite \\
   & \AD-\{\IS,\NS\} & \NPc & \addISNSNAAGcite & \NPc & \addISNSNADGcite & \NPc & \addISNSPAAGcite & \NPc & \addISNSPADGcite & \Wtwohour, \XPour & \addISNSGRAGcite & \Wtwohour, \XPour & \addISNSGRAGcite \\
  &  \quad {\scriptsize \DAG/\SYM} & \NPcour & \DAGSYMaddISNSNAAGcite & \NPcour & \DAGSYMaddISNSNADGcite & \NPcour & \DAGSYMaddISNSPAAGcite & \NPcour & \DAGSYMaddISNSPADGcite & \Wtwohour, \XPour & \DAGSYMaddIRISNSGRDGcite & \Wtwohour, \XPour & \DAGSYMaddIRISNSGRDGcite \\
   & \AD-\CS \STab & \sigmac & \addCSNAAGcite & \sigmac & \addCSNADGcite & \sigmac & \addCSPAAGcite & \sigmac & \addCSPADGcite & \coNPcour & \SYMMaddCSGRDGcite & \coNPcour & \SYMMaddCSGRDGcite \\
      & \quad {\scriptsize \DAG} & \PPour & \DAGaddCSNAAGcite & \IMMour & \DAGaddCSNADGcite & \NPcour & \DAGaddCSPAAGcite &\IMMhour & \DAGaddCSPADGcite & \Wtwohour, \XPour & \DAGaddCSGRAGcite & \Wtwohour, \XPour  &\DAGaddCSGRDGcite \\
      \bottomrule
    \end{tabular}
  }
\end{table*}
}

\mypara{Related work.}
Hedonic games were first conceptualized by \citet{Dreze80_Hedonic} and formally reintroduced by \citet{Banerjee01_Core} and \citet{Bogomolnaia02_NS_Symmetric}.
They independently defined hedonic coalition formation games and analyzed fundamental stability concepts such as the core and Nash stability for additive preferences. 
\citet{dimitrov2006simple} introduced the friends-and-enemies model.
\citet{Woeginger_CoreSurvey} surveyed different preference models for core stability.
\citet{Brandt24_AI_NashStable} demonstrated that most stability-related problems remain intractable even under fairly restrictive preference assumptions.
For a comprehensive overview of complexity results regarding verifying whether a partition is stable and determining the existence of stable partitions, we refer to a recent survey by~\citet{CHS2025FPTCOMSOC}.

To our knowledge, no prior work has studied control in hedonic games.
\citet{Bartholdi92_Control} introduced electoral control (adding or deleting voters or candidates to change the election winner) into voting theory.
\citet{Hemaspaandra07_Preclude} explored both constructive and destructive control in elections under various voting rules, showing that many forms of control can be computationally difficult (a desirable property for election security), and providing a template for defining control actions in a precise algorithmic way; also see the book chapter by \citet{fal-rot:b:handbook-comsoc-control-and-bribery} for more references on voting control.

\citet{Boehmer21_StableMarriage} first studied control in matching markets, systematically exploring external control in Stable Marriage by defining a range of manipulative actions including adding or deleting agents.
\citet{Chen25_StableMatchControl} recently provided a comprehensive complexity overview of adding or removing agents in both Stable Marriage and Stable Roommates settings to achieve certain outcomes (such as guaranteeing the existence of a stable matching or ensuring a particular pair is matched).
For control problems in other domains, we refer to a recent survey by \citet{che-kac-nue-rot-sch-see:c:control-in-computational-social-choice}.

\mypara{Outline of the paper.}
In \cref{sec:prelim}, we define hedonic games, the stability concepts we consider, and our control problems.
In \cref{sec:structural}, we provide some structural results for hedonic games and our control problems.
In Sections \ref{sec:friends} and \ref{sec:add-setting},
we investigate the complexity for the friend-oriented preference and additive preference settings, respectively.
For both sections, we first consider the control goal of \NA, then \PA, and finally \GR.
We conclude with a discussion on potential areas for future research in \cref{sec:conclusion}. 
\ifcr
\else
  Proofs and statements marked with (\appsymb) are deferred to the Appendix. 
\fi

\todoJ{UPDATE ARXIV ENTRY IN BIB, when finishing}

\section{Preliminaries}\label{sec:prelim}
\appendixsection{sec:prelim}
Given an integer~$t$, let \myemph{$[t]$} denote the set~$\{1,2,\dots,t\}$.
Given a directed graph~$G$ and a vertex~$v \in V(G)$, let \myemph{$N^+_G(v)$} and \myemph{$N^-_G(v)$} denote the out- and in-neighborhood of~$v$, respectively. 
Given a directed graph~$G = (V, A)$ and a subset~$V'\subseteq V$ of vertices,
a subgraph \myemph{induced} by $V'$, written as \myemph{$G[V']$}, is a subgraph $(V', A')$ of $G$ where
$A'=\{a' \in A \mid a' \subseteq V'\}$.

Let $\agents$ be a finite set of~$n$ agents.
A \myemph{coalition} is a non-empty subset of~$\agents$.
We call the entire agent set~$\agents$ the \myemph{grand coalition}.
The input of a hedonic game is a tuple~\myemph{$(\agents, (\succeq_{\ag{i}})_{\ag i \in \agents})$},
where $\agents$ is the agent set, 
and each agent~$\ag{i} \in \agents$ has a preference order~$\succeq_i$ over all coalitions that contain~$\ag i$.
Each preference order is a weak order (i.e., complete, reflexive, and transitive).
For two coalitions~$S$ and $T$ containing~$\ag{i}$, we say that agent~$\ag{i}$ \myemph{weakly prefers}~$S$ to~$T$ if $S \succeq_{\ag{i}} T$;
$\ag{i}$ (\myemph{strictly}) \myemph{prefers} $S$ to~$T$ (written as \myemph{$S \succ_{\ag{i}} T$}) if  $S \succeq_{\ag{i}} T$, but $T \not{\succeq}_{\agi} S$;
$\ag{i}$ is \myemph{indifferent between} $S$ and $T$ (written as \myemph{$S \sim_{\ag{i}} T$})  if $S \succeq_{\ag{i}} T$ and $T \succeq_{\ag{i}} S$.

A \myemph{partition}~$\Pi$ of~$\agents$ is a division of~$\agents$ into disjoint coalitions, i.e., the coalitions in~$\Pi$ are pairwise disjoint and $\bigcup_{\coal\in \Pi} \coal = \agents$.
Given a partition~$\Pi$ of $\agents$ and an agent~$\ag i\in \agents$, let~\myemph{$\Pi(\ag i)$} denote the coalition which contains~$\ag i$.
We also call the partition where every agent is in the grand coalition the \myemph{grand coalition partition}.

\mypara{Compact preference representations of hedonic games.}
In this paper, we focus on two compact presentation models that can be encoded polynomially in the number of agents. 
The first setting is called hedonic games with \myemph{\additive\ preferences~(\AD)}.
In \AD, every agent only needs to express a cardinal utility to every other agent. 
The second setting is simple restriction of the first setting and is called hedonic games with \myemph{friend-oriented preferences~(\FA)}.
In \FA, every agent regards every other agent either as a \myemph{friend} or an \myemph{enemy} such that he prefers coalitions with more friends to those with less friends.

\begin{definition}[\AD~\cite{Banerjee01_Core}]
  Let $\agents$ be a set of agents.
  The input of \AD\ is a tuple~$(\agents, (\util{\ag{i}})_{\ag{i}\in \agents})$, where
  every agent~$\ag{i}\in \agents$ has a \myemph{cardinal utility function}~$\util{i} \colon \agents \to \mathds{R}$ such that for each two coalitions~$S$ and $T$ containing~$\ag{i}$,
  agent~$\ag{i}$ weakly prefers $S$ to $T$ if $\sum_{j\in S}\util{\ag{i}}(\ag{j}) \geq \sum_{j\in T}\util{\ag{i}}(\ag{j})$.
  We assume $\util{\ag i}(\ag i) = 0$.
  The utility functions can also be compactly represented by a \myemph{preference graph}.
  It is a tuple~$(\utilG, \weight)$, where $\utilG$ is a directed graph on the agent set and~$\weight$ is an arc-weighting function.
  For each two agents~$\ag{i}$ and $\ag{j}$ there is an arc~$(\ag{i}, \ag{j})$ if and only if the utility of~$\ag{i}$ to~$\ag{j}$ is non-zero.
  The weight of this arc is equal to the non-zero utility~$\weight(\ag{i}, \ag{j}) = \util{\ag{i}}(\ag{j})$.
\end{definition}

\begin{definition}[\FA~\cite{dimitrov2006simple}]
  The input of \FA\ is a directed graph~$\goodG$ where every vertex corresponds to an agent such that every agent~$\ag{i}$ considers another agent~$\ag{j}$ as a \myemph{friend} if and only if there is an arc from~$\ag{i}$ to~$\ag{j}$; otherwise~$\ag{i}$ considers~$\ag{j}$ as an \myemph{enemy}.
  Graph~$\goodG$ is also called a \myemph{friendship} graph. 
  For each agent~$\ag{i} \in \agents$ and each two coalitions~$S$ and $T$ containing~$\ag{i}$,
  agent $i$ \myemph{weakly prefers}~$S$ to~$T$ if 
  \begin{compactenum}[(i)]
    \item either $|N^+_{\goodG}(\ag{i})\cap S| > |N^+_{\goodG}(\ag{i})\cap T|$, or
    \item $|N^+_{\goodG}(\ag{i})\cap S| = |N^+_{\goodG}(\ag{i})\cap T|$ and
    $|S \setminus N^+_{\goodG}(\ag{i})| \leq | T \setminus N^+_{\goodG}(\ag{i})|$.
  \end{compactenum}
  
\end{definition}

\begin{remark}\label{remark:FEN->Add} %
  Note that \FA\ is a simple restriction of \AD: Set $\util{x}(y) = n$ if $(x,y)$ is an arc in~$\goodG$, and $\util{x}(y)=-1$ otherwise. 
\end{remark}

  We say that an instance of \ADD\ (resp.\ \FE) is a \myemph{\DAG} if the preference graph (resp.\ friendship graph) is acyclic.
  Correspondingly, we say that it has \myemph{$\fas$~feedback arcs} if the graph can be turned acyclic by deleting at most~$\fas$ arcs.
  Further, we say that \ADD\ is \myemph{symmetric} if for every pair of agents $\ag i, \ag j$ it holds that $\util i (\ag j) = \util j (\ag i)$.
  Similarly, we say that an instance of \FE\ is \myemph{symmetric} if the friendship graph~\goodG\ is symmetric.
  In these cases we can assume that the preference graph (resp.\ friendship graph) is undirected.

\mypara{Relevant stability concepts.}
In this paper, we study four relevant stability concepts. 
\begin{definition}
  Let $\Pi$ be a partition of~$\agents$. 
  A coalition~$\bcoal$ is \myemph{blocking} a partition~$\Pi$ if every agent~$\ag i \in \bcoal$ prefers~$\bcoal$ to $\Pi(\ag i)$.
  An agent $\ag{i}$ and a coalition~$\bcoal$ form a \myemph{blocking tuple} if
  $\ag i$ prefers $\bcoal \cup \{\ag i\}$ to $\Pi(\ag{i})$ and each agent~$\ag j\in \bcoal$ weakly prefers $\bcoal\cup \{\ag i\}$ to~$\Pi(\ag{j})$.	
  \begin{compactitem}[--]
    \item $\Pi$ is \myemph{\irt\ (\IR)} if no agent~$\ag{i}$ prefers $\{\ag{i}\}$ to $\Pi(\ag{i})$. 
    If $\ag{i}$ prefers $\{\ag{i}\}$ to $\Pi(\ag{i})$, we say he wishes to \myemph{deviate} from~$\Pi$.
    \item $\Pi$ is \myemph{\nst\ (\NS)} if
    no agent~$\ag{i}$ and coalition~$\bcoal\in \Pi\cup \{\emptyset\}$ exist such that $\ag{i}$ strictly prefers~$\bcoal\cup \{\ag{i}\}$
    to his coalition~$\Pi(\ag{i})$.
    \item $\Pi$ is \myemph{\ist\ (\IS)} if no agent~$\ag i$ and coalition $\bcoal\in  \Pi \cup \{\emptyset\}$ can form a blocking tuple.
    \item $\Pi$ is \myemph{\cst\ (\CS)} if \emph{no} coalition is blocking~$\Pi$.
  \end{compactitem}
\end{definition}

By definition, the stability concepts satisfy the following:
\begin{observation}\label{obs:stability-relation}
  An \NS\ partition is \IS. An \IS\ partition is \IR.
  A \CS\ partition is \IR. 
\end{observation}

\mypara{Our control problems and their complexity upper bounds.}
In this paper, we study three different control goals
\begin{inparaenum}[(1)]
  \item enforcing that an agent is \emph{not alone} (\myemph{\NA});
  \item enforcing that a \emph{pair} of agents is in the same coalition (\myemph{\PA});
  \item enforcing that all agents are in the \emph{same grand} coalition (\myemph{\GR}).
\end{inparaenum}
Moreover, we study two possible control actions we can use to obtain the control goals:
adding agents~(\myemph{\AddAg}) or deleting agents~(\myemph{\DelAg}).

\begin{example}
  Let $U=\{x=u_1, u_2, u_3\}$ and $W=\{w_1,w_2\}$, where~$U$ is a set of consisting of three original agents while $W$
  consisting of two additional agents, respectively. 
  The preference graph of the agents~$U\cup W$ is depicted below.
  Throughout, we use blue dotted line to indicate friendship relation, and  red solid line enemy relation. 

  \def \xx {1.5}
  \def \xy {.65}
  {
    \centering
    \begin{tikzpicture}[>=stealth',shorten <= 1pt, shorten >= 1pt]
      \foreach \i / \j / \n in {0/0/u1,-1/1/u2,-1/-1/u3} {
        \node[unode] at (\i*\xx, \j*\xy) (\n) {};
      } 
      \foreach \i / \j / \n in {1/1/w1, 1/-1/w2} {
        \node[wnode] at (\i*\xx, \j*\xy) (\n) {};
      }
      \foreach \n / \p / \l / \nn in {u1/above/-1/{u_1}, u2/left/1/{u_2}, u3/left/1/{u_3}, w1/right/1/{w_1}, w2/right/1/{w_2}} {
        \node[\p = \l pt of \n] {$\nn$}; 
      }
      \foreach \s / \t / \x / \b / \et in {u1/u2/-1/0/ec,
        u3/u1/-2/0/ec, u3/u2/2/0/fc, u1/w1/1/15/fc, w1/u1/{-1}/15/ec, u1/w2/1/15/fc, w2/u1/{-1}/15/ec, w2/u3/1/0/fc} {
        \draw[->, \et] (\s) edge[bend left=\b] node[midway, labelst] {\small $\x$} (\t);
      }
      
    \end{tikzpicture}
    
  }

  Our special agent is~$x=u_1$.
  In the original instance, consisting of only~$U$, agent~$u_1$ must be alone in an \IR\ partition since he dislikes~$u_2$ but $u_3$ dislikes him.
  One can verify that it is impossible to make $u_1$ not alone by deleting agents since otherwise the new \IR\ partition augmented with the deleted agents in singletons would yield an \IR\ partition for the original instance.

  We can add agent~$w_2$ (not $w_1$) to the original instance to obtain an \IR\ partition, where every agent is in the same grand coalition.
\end{example}

Now we are ready to formally define our control problems (as decision problems)~HG-$\S$-\G-\A, %
where $\S\in \{\IR, \IS, \NS, \CS\}$ denotes one of the four stability concepts, $\G\in \{\NA, \PA, \GR\}$ one of the three control goals, and $\A\in \{\AddAg, \DelAg\}$ one of the two control actions, respectively.
The problem of control by adding agents is defined as follows:

\decprob{HG-$\S$-\NA-\AddAg~(\normalfont{\textit{resp.}} HG-$\S$-\PA-\AddAg)}
{A hedonic game instance $(\agents = \agentsU \cup \agentsW, (\succeq_i)_{\ag i \in \agents})$,
  a selected agent~$\agx \in \agentsU$ (resp.\ selected agent pair~$\{\agx, \agy\} \subseteq \agentsU$), and
  a non-negative value~$\budget \in \mathds{N}\cup \{0\}$.}
{Is there a subset $\agentsW' \subseteq \agentsW$ of size~$|\agentsW'| \leq \budget$
  such that $\agentsU\cup \agentsW'$ admits an $\S$-partition $\Pi$ where $\Pi(\agx) \neq \{\agx\}$ (resp.\ $\Pi(\agx) = \Pi(\agy)$)?}

For the control goal~\GR, the input only consists of the hedonic game instance and a budget:
\decprob{HG-$\S$-\GR-\AddAg}
{A hedonic game instance $(\agents = \agentsU \cup \agentsW, (\succeq_i)_{\ag i \in \agents})$ and
  a non-negative value~$\budget \in \mathds{N}\cup \{0\}$.}
{Is there a subset $\agentsW' \subseteq \agentsW$ of size~$|\agentsW'| \leq \budget$
  such that the partition consisting of the grand coalition~$\agentsU\cup \agentsW'$ is~$\S$?}

In the above, we call the set~$\agentsU$ the \myemph{original} agents and the set~$\agentsW$ the \myemph{additional} agents.

Similarly, we define the problems for the setting where we can delete agents.

\decprob{HG-$\S$-\NA-\DelAg~(\normalfont{\textit{resp.}} HG-$\S$-\PA-\DelAg)}
{A hedonic game instance $(\agentsU, (\succeq_i)_{\ag i \in \agentsU})$,
  a selected agent~$\agx \in \agentsU$ (resp.\ selected agent pair~$\{\agx, \agy\} \subseteq \agentsU$), and
  a non-negative value~$\budget \in \mathds{N}\cup \{0\}$.}
{Is there a subset $\agentsU' \subseteq \agentsU$ of size~$|\agentsU'| \leq \budget$
  such that $\agentsU\setminus \agentsU'$ admits an $\S$-partition $\Pi$ where $\Pi(\agx) \neq \{\agx\}$ (resp.\ $\Pi(\agx) = \Pi(\agy)$)?}

\decprob{HG-$\S$-\GR-\DelAg}
{A hedonic game instance $(\agentsU, (\succeq_i)_{\ag i \in \agentsU})$ and
  a non-negative value~$\budget \in \mathds{N}\cup \{0\}$.}
{Is there a subset $\agentsU' \subseteq \agentsU$ of size~$|\agentsU'| \leq \budget$
  such that the partition consisting of the grand coalition~$\agentsU\setminus \agentsU'$ is~$\S$?}

If the hedonic game is restricted to be \AD\ or \FA, we replace the prefix~HG with \AD\ or \FA.

To classify cases where control is impossible, we define the following.
\begin{definition}[\immune\ and \never]\label{def:immune+never}
  We say a hedonic game control problem is \myemph{\immune} if
  for every \no\ with $k=0$,
  it remains a \no\ even if we set $k=\infty$.
  We say that the problem is \myemph{\never} if all instances containing at least two agents are \no s. 
\end{definition}

Clearly, an instance with the control goal \NA\ with one agent overall is always a \no, and an instance with the control goal \PA\ contains at least two agents by definition.
On the other hand, an instance with the control goal \GR\ that contains one agent is trivially a \yes.

It is known that for both preference settings, one can check in polynomial time whether a given partition is \IR, \IS, or \NS.
However, the verification problem for \CS\ is \coNP-complete~\cite{CheCsaRoySim2023Verif}.
This immediately yields the following complexity upper bounds for our control problems. 
\begin{observation}\label{obs:complexityupperbounds}
  For preference model~$\M\in \{\AD$, $\FA\}$, control goal~$\G\in \{\NA, \PA, \GR\}$, and control action~$\A\in \{\AddAg,$ $\DelAg\}$,
  the first three sets of problems~\controlprob{\IR}{\G}{\A}{\M}, \controlprob{\IS}{\G}{\A}{\M}, and \controlprob{\NS}{\G}{\A}{\M} are in \NP,
  while the problems~\controlprob{\CS}{\NA}{\A}{\M} and~\controlprob{\CS}{\PA}{\A}{\M} are in \sigmatwop,
  and the problem~\controlprob{\CS}{\GR}{\A}{\M} is in \coNP.
\end{observation}

\mypara{$\boldsymbol{\Wtwo}$ and $\boldsymbol{\XP}$.}
In this paper, we prove that several problems are \Wtwo-hard\ (resp.\ in \XP) with respect to a parameter~$p$.
\Wtwo-hardness is shown via a \myemph{parameterized reduction} from a \Wtwo-hard problem~$Q$ (parameter~$q$) running in time $f(q)\cdot |I_Q|^{O(1)}$ and producing an instance with parameter $p\le g(q)$, for computable $f,g$.
Thus, it is unlikely that the problems are fixed-parameter tractable (FPT), i.e., solvable in time $f'(p)\cdot |I|^{O(1)}$ for computable $f'$, unless $\mathsf{FPT}=\Wtwo$.
Membership in \XP\ wrt.\ $p$ means solvability in time $|I|^{h(p)}$ for some computable $h$ (equivalently, polynomial time for every constant $p$).
If a problem remains \NP-hard for some fixed constant value of $p$, then it is not in \XP\ wrt.\ $p$, unless $\PPP=\NP$.
See~\cite{Nie06,CyFoKoLoMaPiPiSa2015} for details.

\section{Structural Observations}\label{sec:structural}

In this section, we collect some useful structural properties that may be of independent interest.

\mypara{Relations among the stability concepts.}
The first two observations describe two cases when a more stringent stability concept is equivalent to \IR.
This is useful in searching for both algorithms and hardness results. 

The first observation follows directly from definition. 
\begin{restatable}{observation}{obsGrandEQ}
  \label{obs:gra:eq}
  For every $\S \in \{\NS, \IS\}$, the grand coalition partition is $\S$ if and only if it is \IR.
\end{restatable}

Next, we observe that %
if the preference graph is a \DAG, the two stability concepts \IR\ and \CS\ coincide.

\begin{restatable}{observation}{obsDAGEQ}
  \label{obs:dag:eq}
  For every instance of \AD\ with an acyclic preference graph, a partition is \CS\ if and only if it is \IR.
\end{restatable}

\begin{proof}
  By \cref{obs:stability-relation}, it suffices to show that \IR\ implies \CS.
  Suppose, for the sake of contradiction, that $\Pi$ is \IR\ but not \CS.
  Let~$\bcoal$ be a blocking coalition with at least two agents.
  Since the preference graph is acyclic, one agent $t \in \bcoal$ is a sink in the preference graph induced by~$\bcoal$.
  Such a sink $t$ has zero utility towards~$\bcoal$, but has utility at least zero towards $\Pi(t)$, since $\Pi$ is \IR.
  Thus, agent $t$ has no incentive to deviate to~$\bcoal$, a contradiction.
\end{proof}

The next lemma shows that under the friend-oriented model and for the control goal~\PA, the existence questions for \IR\ and \IS\ are essentially the same.

\begin{restatable}[\appsymb]{lemma}{obsFriIRIS} 
  Let $I = (\agents, \goodG)$ be a \FE-instance, and $\agx$ and $\agy$ two agents in~$\agents$.
  From each \IR\ partition~$\Pi$ with $\Pi(\agx) = \Pi(\agy)$,
  one can construct in polynomial time an \IS\ partition~$\Pi'$ with $\Pi'(\agx) = \Pi'(\agy)$. 
\label{obs:fri:ir:is}
\end{restatable}

\begin{proof}[Proof sketch.]
  The idea is to start with an \IR\ partition where $\agx$ and~$\agy$ are in the same coalition, say~$C$, and merge it with all strongly connected components in~$\goodG$.
  We then repeatedly add agents to~$C$ whenever they can reach some agent in $C$.
  All remaining agents will be in their own singleton coalitions.
\end{proof}
\newcommand{\megacoal}{\ensuremath{M}}
\appendixproofwithstatement{obs:fri:ir:is}{\obsFriIRIS*}
{
Let $I,\agx,\agy,\Pi$ be as defined. %
The desired \IS\ partition $\Pi'$ will consist of a big coalition~$\megacoal$ with possibly some singleton coalitions.
\begin{compactenum}[Step 1:]
  \item Add the agents of all size-at-least-two coalitions in~$\Pi$ to~$\megacoal$.
  Formally, $\megacoal \coloneqq \bigcup_{\coal \in \Pi, |\coal| \geq 2} \coal$. %
  Since $\Pi(\agx) = \Pi(\agy)$, it holds that $\agx, \agy \in \megacoal$.
  \item For each strongly connected component~$C$ of size at least two in~$\goodG$,
  add the agents in~$C$ to~$\megacoal$. %
  \item Add all agents to $\megacoal$ that can reach some agent in~$\megacoal$.
\end{compactenum}
For every agent not in $\megacoal$, set his coalition as the singleton coalition under $\Pi'$.

Clearly by Step 1 it holds that $\Pi'(\agx) = \Pi'(\agy) = \megacoal$.

It remains to show that $\Pi'$ is \IS.
Let us first show that $\Pi'$ is \IR, i.e., no agent can prefer deviating to singleton coalition.
Clearly the agents in the singleton coalitions cannot deviate.
Now consider an agent $\ag i$ in $\megacoal$.
If $\ag i$ was added to $\megacoal$ in Step 1, then $|\Pi(\ag i)| \geq 2$.
This means that $\ag i$ must obtain at least one friend in $\Pi(\ag i)$, because otherwise~$\ag i$ would prefer $\{\ag i\}$ to $\Pi(\ag i)$, and hence $\Pi$ would not be \IR.
Since $\ag i$ obtains at least one friend in $\Pi(\ag i)$ and $\Pi(\ag i) \subseteq \megacoal$, he also obtains a friend under $\Pi'(\ag i) = \megacoal$, and thus prefers $\megacoal$ to $\{\ag i\}$.
If instead $\ag i$ was added to $\megacoal$ in Step 2, then $\ag i$ is part of a strongly connected component $C$ of size at least $2$. It must hold that $\ag i$ considers someone in $C$ a friend and $C \subseteq \megacoal$.
Thus, $\ag i$ obtains a friend in $\Pi'(\ag i) = \megacoal$, and thus prefers $\megacoal$ to $\{\ag i\}$. 

Next we must show that no agent, coalition pair $(\ag i, \coal)$, where $\ag i \in \agents, \coal \in \Pi'$ forms a blocking tuple.
We have three cases to consider:
\begin{compactitem}
\item[Case 1:] Agent~$\ag i$ is in a singleton coalition under $\Pi'$ and $\coal$ is a singleton coalition.
Let $\coal = \{\ag j\}$.
Since $\ag i$ prefers $\{\ag i, \ag j\}$ to $\{\ag i\}$, agent~$\ag j$ must be a friend of~$\ag i$.
Since $\ag j$ weakly prefers  $\{\ag i, \ag j\}$ to $\{\ag j\}$, agent $\ag i$ must be a friend of $\ag j$.
But then $\ag i$ and $\ag j$ are mutually reachable and they must be in the same strongly connected component of $\goodG$.
However, we did not add them to $\megacoal$ in Step 2, a contradiction.
\item[Case 2:] Agent~$\ag i$ is in a singleton coalition under $\Pi'$ and $\coal = \megacoal$.
Then $\ag i$ prefers $\{\ag i\} \cup \megacoal$ to $\ag i$ and thus must have a friend in $\megacoal$.
But then $\ag i$ can reach an agent in $\megacoal$, a contradiction to not adding $\ag i$ to $\megacoal$ in Step 3.
\item[Case 3:] Agent~$\ag i \in \megacoal$ and $\coal$ is a singleton coalition.
Let $\coal = \{\ag j\}$.
Since $\ag j$ weakly prefers  $\{\ag i, \ag j\}$ to $\{\ag j\}$, agent $\ag i$ must be a friend of $\ag j$.
But then $\ag j$ can reach an agent in $\megacoal$, a contradiction to not adding $\ag j$ to $\megacoal$ in Step 3.
\end{compactitem}
Since all cases lead to contradiction, it must be that $\Pi'$ is \IS, as required.
}

\mypara{Influence of the control goals on the complexity.}
The next result implies that enforcing the control goal~\NA\ or \PA\ will not help in lowering the complexity if the underlying problem of determining the existence of a stable partition is hard.
\begin{restatable}[\appsymb]{observation}{obsOriginalhardFollows}%
  For every $\M\in \{\AD$, $\FA\}$ and~$\S \in \{\IR$, $\NS, \IS, \CS\}$
  if it is \sigmatwop-hard (resp.\ \NP-hard)  to determine the existence of an $\S$-partition under preference representation~$\M$,
  then so is \controlprob{\S}{\G}{\A}{\M} for every control goal~$\G\in \{\NA, \PA\}$, control action $\A\in \{\AddAg, \DelAg\}$, even when $\budget = 0$.
  The same implication holds even when the preferences are symmetric.
\label{obs:originalhard-follows}
\end{restatable}
\appendixproofwithstatement{obs:originalhard-follows}{\obsOriginalhardFollows*}{
  Take any hardness reduction which shows that the existence of an $\S$-partition for $\M$ is \sigmatwop-hard (resp.\ \NP-hard).
  Add two additional agents $\{x,y\}$ to the instance of $\M$ who have positive utility towards each other and negative to all other agents (resp. in the friends-setting, agents $x,y$ are mutual friends and have no other friends).
  Now the reduction also holds for \controlprob{\S}{\G}{\A}{\M} for every control goal~$\G\in \{\NA, \PA\}$ by requiring that $\agx$ is not alone (resp. $\{\agx, y\}$ are paired).
}
\begin{proof}[Proof sketch.]
Add two additional agents $\{x,y\}$ who have positive utility towards each other and negative to all other agents. They will not affect the stability of the rest of the structure.
\end{proof}

The next result allows us to focus on the control goal~\PA\ when searching for efficient algorithms for the goal~\NA.
\begin{restatable}[\appsymb]{observation}{obsPAtoNA}
  For every $\M \in \{\AD, \FA\}$, $\S \in \{\IR, \NS$, $\IS, \CS\}$,
  and $\A\in \{\AddAg, \DelAg\}$, if \controlprob{\S}{\PA}{\A}{\M} is polynomial-time solvable, then so is \controlprob{\S}{\NA}{\A}{\M}.\label{obs:pa->na}
\end{restatable}
\appendixproofwithstatement{obs:pa->na}{\obsPAtoNA*}{
  Let $(\agents = \agentsU \cup \agentsW, (\succeq_i)_{i \in \agents}, \{x\},k)$ be an instance of \controlprob{\S}{\NA}{\A}{\M}.
  If \controlprob{\S}{\PA}{\A}{\M} is polynomial-time solvable, we can try to pair the special agent $\agx$ with every other agent $a \in \agentsU \cup \agentsW$.
  We iteratively solve \controlprob{\S}{\PA}{\A}{\M} with $(\agents = \agentsU \cup \agentsW, (\succeq_i)_{i \in \agents}, \{x, a\},k)$, which is possible in polynomial-time.
  If $a \in \agentsW$, we additionally decrease the budget $\budget$ by one and move $a$ from $\agentsW$ to $\agentsU$.
  
  If for some pair $\{\agx, a^*\}$, \controlprob{\S}{\PA}{\A}{\M} is a \yes, then \controlprob{\S}{\NA}{\A}{\M} is also a \yes\ by taking the partition returned for the paired version.
  If, on the other hand, for every pair $\{\agx, a\}$, \controlprob{\S}{\PA}{\A}{\M} is a \no, then no \S\ partition exists in which $\agx$ is paired with another agent. 
  Hence, $\agx$ must be lone in every \S\ partition.
}
\begin{proof}[Proof sketch]
  We guess the other agent~$y \in (\agentsU \cup \agentsW$) with whom $x$ should be paired with. 
  If there is no $y$ such that \controlprob{\S}{\PA}{\A}{\M} is a \yes, $x$ will be always alone.
\end{proof}

Next, we state a well-known tight relation between the coalitions in a \CS\ partition and the strongly connected components of the friendship graph, which is very useful for the \PA\ goal and \CS\ stability.
\begin{restatable}[\cite{dimitrov2006simple,Woeginger_CoreSurvey},\appsymb]{observation}{obsFriendsPASCC}\label{obs:csstogether}
For every \FA-instance and two agents~$\ag{x}$ and $\ag{y}$, there is a \cst\ partition $\Pi$ such that $x$ and $y$ belong to the same coalition if and only if they belong to the same strongly connected component.
\end{restatable} 
\appendixproofwithstatement{obs:csstogether}{\obsFriendsPASCC*}{
  We show the statement for the sake of completion.
  For the ``if'' part, assume that $\ag{x}$ and $\ag{y}$ are in the same strongly connected component (SCC).
  Since, by \citet{dimitrov2006simple}, there exists a \CS\ partition, where each coalition corresponds to an SCC in the graph,
  this \CS\ partition must include a coalition with $\ag{x}$ and $\ag{y}$ being together.

  For ``only if'' part, assume that $\Pi$ is a \CS\ partition with $\Pi(\ag{x})=\Pi(\ag{y})$.
  Suppose, for the sake of contradiction, that $\ag{x}$ and $\ag{y}$ do not belong to the same SCC.
  This implies that the friendship graph~$\goodG[\Pi(\ag{x})]$ induced by $\Pi(\ag{x})$ has more than one SCC.
  Let $C$ be a sink component in the topological ordering of the SCCs in $\goodG[\Pi(\ag{x})]$.
  Then, every agent in~$C$ has the same number of friends in~$C$ as in~$\Pi(\ag{x})$ and at least one enemy less in~$C$ than in~$\Pi(\ag{x})$,
  meaning that $C$ is blocking, a contradiction.
}

\mypara{Immune cases.} We close this section by summarizing cases when control is immune.
\begin{restatable}{proposition}{propIMMUNEDELAG}
  For each stability concept $\S \in \{\IR, \IS, \CS\}$ and each control goal $\G \in \{\NA,\PA\}$, \controlprob{\S}{\G}{\DelAg}{\FE} is immune.
Under symmetric preferences \controlprob{\NS}{\NA}{\DelAg}{\FE} and \controlprob{\NS}{\PA}{\DelAg}{\FE} are immune.
  \label{prop:IMMUNE-DELAG}
\end{restatable}

\begin{proof}
First observe that if it is possible to delete some agents to have an \IR-partition such that $x$ is not alone (resp.\ $x$ and $y$ are in the same coalition), then adding back the deleted agents each as a singleton coalition would be an \IR-partition with the same goal as well.
By \cref{obs:fri:ir:is}, this result extends to \IS\ as well.

By \cref{obs:csstogether}, if it is possible to delete some agents to have a \CS-partition such that $x$ is not alone (resp.\ $x$ and~$y$ are in the same coalition), then $x$ must be part of a connected component of size at least two (resp.\ $x$ and~$y$ are part of the same strongly connected component) in the friendship graph of the remaining agents.
If the agents are connected, they are also connected after adding back the agents.
By \cref{obs:csstogether}, there is then a \CS\ partition of the original agent set such that $x$ is not alone (resp.\ $x$ and $y$ are together).

If it is possible to delete some agents to have a \NS\ partition such that $x$ is not alone (resp.\ $x$ and~$y$ are in the same coalition), then $x$ must have a friend (resp. $x$ and $y$ both have a friend).
Let us now add back all the removed agents.
Consider the partition in which all agents with degree at least one are placed in the same big coalition, while all remaining agents form singleton coalitions.
This partition is \NS.
Since preferences are symmetric, no agent in the big coalition is friends with the singletons, and since every agent in the big coalition has a friend within that coalition, he prefers it to being alone.
Since $x$ has a friend (resp.\ $x$ and $y$ have a friend) he (resp.\ they) must be in the big coalition.
\end{proof}

DAGs in \FE\ are very restrictive.
The \IR\ partition is unique and no two agents are together in it.

\begin{restatable}[\appsymb]{observation}{obsDAGFENO}
  For \FE-instances with acyclic friendship graph, the only partition that may be stable for $S \in \{\IR, \IS, \NS, \CS\}$ consists of every agent being in a singleton coalition.
  \label{obs:DAG-FE-NO}
\end{restatable}
\appendixproofwithstatement{obs:DAG-FE-NO}{\obsDAGFENO*}{
    For each coalition of size at least two, there is always an agent who is a sink in the induced subgraph.
    This sink agent will prefer to be alone.
    Hence, the only \IR\ partition is to have every agent alone.
}

\begin{restatable}{proposition}{propADIRIMMUNEDELAG}
  For control goal~$\G \in \{\NA,\PA\}$, \controlprob{\IR}{\G}{\DelAg}{\AD} is immune.
  \label{prop:AD-IR-IMMUNE-DELAG}
\end{restatable}

\begin{proof}
  If it would be possible to delete some agents to have an \IR-partition such that $x$ is not alone (resp.\ $x$ and $y$ are in the same coalition),
  then adding back the deleted agents each as a singleton coalition would be an \IR-partition with the same goal as well.
\end{proof}

\section{Friend-Oriented Preferences}
\label{sec:friends}

In this section, we discuss the results relating to \FE.
For the control goals \NA\ and \PA, we discover that \IR, \IS, and \CS\ are all polynomial-time solvable for the control action \AddAg: These stability concepts require certain friendship structures and finding a minimum number of agents to add to realize those can be done in polynomial time.
In \cref{prop:IMMUNE-DELAG} we saw that these stability concepts were also immune to the control action \DelAg.

For \NS, on the other hand, \controlprob{\NS}{\G}{\A}{\FE} is \NP-hard for every $\G \in \{ \NA, \PA\}, \A \in \{\AddAg, \DelAg\}$ even when $\budget = 0$ by a result from \cite{Brandt24_AI_NashStable}.
We however discover that symmetric preferences make the problems tractable.

Finally, we look into enforcing that the grand coalition partition is stable (\GR) and see that regardless of the stability concept, obtaining grand coalition through agent deletion is computationally easier than through agent addition.

We start the section by presenting an algorithm that solves \controlprob{\IR}{\PA}{\AddAg}{\FE}.

\begin{restatable}[\appsymb]{theorem}{thmFEIRNAADD}
For each stability concept $\S \in \{\IR, \IS\}$ and each control goal $\G \in \{ \NA, \PA\}$, \controlprob{\S}{\G}{\AddAg}{\FE} is polynomial-time solvable.
\label{thm:FE-IR-NA-ADD}
\end{restatable}

\newcommand{\wpath}{\ensuremath{\weight^P}}
\newcommand{\wcycle}{\ensuremath{\weight^C}}
\newcommand{\apath}{\ensuremath{P}}
\newcommand{\acycle}{\ensuremath{C}}

\appendixproofwithstatementsketch{thm:FE-IR-NA-ADD}{\thmFEIRNAADD*}{

\begin{figure}
  \begin{tikzpicture}[>=stealth',shorten <= 1pt, shorten >= 1pt, yscale=0.8]
  \def \xx {1.3}
  \def \xy {1}
  
  \begin{scope}
  \foreach \x / \y / \n in {
  	3.7/4/x,
  	4/3/i,
  	4.4/4/y,
  	3.5/2/j,
        4.6/2/dummy} {
    \node[unode] at (\x*\xx, \y*\xy) (\n) {};
  }

  \foreach \i / \p / \a / \r / \n in {%
    x/above left/-5pt/2pt/x,
    i/below right/-5pt/2pt/\hat{y},
    y/above right/-5pt/2pt/y,
    j/below left/0/0/\hat{x}}%
 {
     \node[\p = \a and \r of \i,fill=white,inner sep=0.1pt] {$\n$};
   }

  \begin{pgfonlayer}{bg}
    \foreach \s / \t / \aa  / \typ in {
      x/i/20/fc,y/i/20/fc,i/j/20/fc,j/dummy/-60/fc, dummy/j/-60/fc}
    {
      \draw[->, \typ] (\s) edge[bend right = \aa]  (\t);
    }

  \end{pgfonlayer}
  \end{scope}
  
  \begin{scope}[xshift=2.4cm]
      \foreach \x / \y / \n in {
  	3.7/4/x,
  	3.7/2/i,
  	4.7/4/y,
  	4.7/2/j} {
    \node[unode] at (\x*\xx, \y*\xy) (\n) {};
  }

  \foreach \i / \p / \a / \r / \n in {%
    x/above left/-5pt/2pt/x,
    i/below left/0/0/\hat{x},
     y/above right/-5pt/2pt/y,
    j/below right/0/0/\hat{y}}%
 {
     \node[\p = \a and \r of \i,fill=white,inner sep=0.1pt] {$\n$};
   }

  \begin{pgfonlayer}{bg}
    \foreach \s / \t / \aa  / \typ in {
    x/i/20/fc,y/j/-20/fc,i/j/-60/fc,j/i/-60/fc}
  {
    \draw[->, \typ] (\s) edge[bend right = \aa]  (\t);
  }

  \end{pgfonlayer}
  \end{scope}
  
  \begin{scope}[xshift=4.8cm]
  \foreach \x / \y / \n in {
  	4/4/x,
  	5/3.2/i, 6/3.2/dummy} {
    \node[unode] at (\x*\xx, \y*\xy) (\n) {};
  }

  \foreach \i / \p / \a / \r / \n in {%
     x/above left/-5pt/2pt/x,
    i/above left/0/-5pt/\hat{x}}%
 {
     \node[\p = \a and \r of \i,fill=white,inner sep=0.1pt] {$\n$};
   }

  \begin{pgfonlayer}{bg}
    \foreach \s / \t / \aa  / \typ in {
    x/i/20/fc, i/dummy/-60/fc, dummy/i/-60/fc}
  {
    \draw[->, \typ] (\s) edge[bend right = \aa]  (\t);
  }

  \end{pgfonlayer}
  
  \begin{scope}[yshift=-1.2cm]
  
   \foreach \x / \y / \n in {
     4/4/y,
     5/3.2/j, 6/3.2/dummy} {
     \node[unode] at (\x*\xx, \y*\xy) (\n) {};
  }

  \foreach \i / \p / \a / \r / \n in {%
     y/above left/-5pt/2pt/y,
     j/above left/0/-5pt/\hat{y}}%
   {
     \node[\p = \a and \r of \i,fill=white,inner sep=0.1pt] {$\n$};
   }

  \begin{pgfonlayer}{bg}
    \foreach \s / \t / \aa  / \typ in {
    y/j/20/fc, j/dummy/-60/fc, dummy/j/-60/fc}
  {
    \draw[->, \typ] (\s) edge[bend right = \aa]  (\t);
  }
 
  \end{pgfonlayer}
  
  \end{scope}
  \end{scope}
\end{tikzpicture}
\caption{The three possible substructures for an \IR\ coalition containing $\agx$ and $\agy$. The arcs indicate paths of varied lengths.
    Structure on left corresponds to \lineref{check1}, middle to \lineref{check2}, and right to \lineref{check3}.}\label{fig:FE-IR-NA-ADD}
\end{figure}
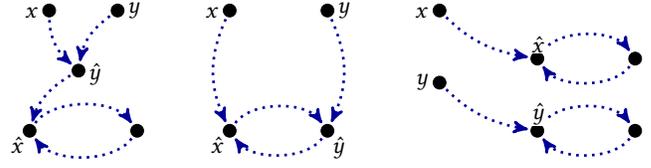

\begin{algorithm}
\caption{Algorithm for \controlprob{\IR}{\pa}{\AddAg}{\FE}. \label{alg:IRPAADDFE}}
\KwIn{%
  $(\agents = \agentsU \cup \agentsW, \goodG = (\agentsU \cup \agentsW, \mathcal{A}), \agx, \agy, k)$}
\lForEach{$(i, j) \in \mathcal{A}$}{$\weight((\ag i, \ag j)) \gets \begin{cases}1, & \text{ if } \ag i \in \agentsW \\ 0, & \text{ otherwise}. \end{cases}$}

\ForEach{$\ag i, \ag j \in \agents$}{$\wpath(\ag i, \ag j) \gets$ weight of a min-weight path from $\ag i$ to $\ag j$.}

\lForEach{$\ag i \in \agents$}{$\wcycle(\ag i) \gets$ weight of a min-weight non-trivial cycle containing $\ag i$.}

\ForEach{$\hat\agx, \hat\agy \in \agents$}{ 
  \If{$\wpath(\agx, \hat\agy) + \wpath(\agy, \hat\agy) + \wpath(\hat\agy, \hat\agx) + \wcycle(\hat\agx) \leq \budget$\label{check1}}{\Return \KwYes}
  
  \lIf{$\wpath(\agx, \hat\agx) + \wpath(\agy, \hat\agy) + \wpath(\hat\agx, \hat\agy) + \wpath(\hat\agy, \hat\agx) \leq \budget$ and $\hat\agx \neq\hat\agy$\label{check2}}{\Return \KwYes}

  \If{ $\wpath(\agx, \hat\agx) + \wpath(\agy, \hat\agy) + \wcycle(\hat\agx) + \wcycle(\hat\agy) \leq \budget$ \label{check3}}{\Return \KwYes}
}
\Return \KwNo

\end{algorithm}

The idea is that a coalition containing both $\agx$ and $\agy$, which is part of an \IR\ partition, must contain one of three possible structures, which we illustrate in \cref{fig:FE-IR-NA-ADD}.
Any agent who is part of an \IR\ partition must be part of a ``$\rho$-shaped'' subgraph~\cite[Chapter 14]{Galbraith2012MPKC}, that is, a graph that consists of a directed cycle of length at least two and a path (possibly of length zero) that reaches a vertex in it.
This is a consequence of every agent obtaining a friend in an $\IR$ coalition of size at least two: If we start from an agent and follow a friendship path, we must eventually encounter an agent we  have already seen.
Since $\agx$ and $\agy$ are both a part of an $\IR$ coalition, they both must be a part of a ``$\rho$-shaped'' subgraph.
There are three ways in which they can intersect, as shown in \cref{fig:FE-IR-NA-ADD}.
We show that these structures can be found in polynomial time by combining with finding minimum-weight paths.
The algorithm is in \cref{alg:IRPAADDFE}.

In the algorithm we construct an arc-weight function $\weight$ for~$\goodG$ based on whether an arc starts with an agent in $\agentsW$ or not.
We compute the all-pair minimum-weight paths on $(\goodG, \weight)$ in polynomial time using Floyd-Warshall~\cite{Floyd_Warshall} or some other algorithm. 
While running the algorithm, we also store, for each agent, an extra entry that records the minimum-weight path from the agent to himself of length at least one, i.e., the minimum-weight non-trivial cycle containing the agent.
Then, we use those paths and cycles to search for the structures shown in \cref{fig:FE-IR-NA-ADD}.

By \cref{obs:fri:ir:is}, this algorithm can also be used to solve \controlprob{\IS}{\PA}{\AddAg}{\FE}.  
By \cref{obs:pa->na}, we then can solve the remaining  \controlprob{\IR}{\NA}{\AddAg}{\FE} and \controlprob{\IS}{\NA}{\AddAg}{\FE}. 
However, the latter two problems can be solved more efficiently by searching for a single ``$\rho$-shaped'' subgraph containing $\agx$.
}{ 
Given two agents $\ag i, \ag j \in \agents$, we use $\wpath(\ag i, \ag j)$ to denote the weight of the minimum-weight path from $\ag i$ to $\ag j$ and $\wcycle(\ag i)$ to denote the weight of the minimum-weight non-trivial cycle containing $\ag i$.
Similarly, let $\apath(\ag i, \ag j)$ denote the set of agents on a minimum-weight path from $\ag i$ to $\ag j$ and $\acycle(\ag i)$ denote the set of agents on a minimum-weight cycle containing $\ag i$, as discovered by the minimum-weight paths algorithm.

Let us first show that if the \cref{alg:IRPAADDFE} returns \KwYes, then $I$ is indeed a \yes for \FA-\IR-\PA-\AddAg.
Since the algorithm returns yes, there are $\ag i \in \agents, \ag j \in \agents$ such that one of the checks was satisfied.

First assume \lineref{check1} returns \KwYes.
Consider the coalition $\coal$ that consists of the agents on $\apath(\agx, \hat\agy) + \apath(\agy, \hat\agy) + \apath(\hat\agy, \hat\agx) + \acycle(\hat\agx)$.
Observe that every agent in this coalition is contained in a path that does not end on him, or in a cycle.
Thus, every agent has an out-arc in~$\goodG[\coal]$ and thus obtains a friend.
Hence no agent wishes to deviate from~$\coal$.

Next observe that $|\coal \cap \agentsW| \leq \budget$: 
Every agent $\ag w \in \coal \cap \agentsW$ is the non-last agent on a path or is contained in a cycle, and hence the respective path or cycle contains an arc starting from~$\ag w$. This arc contributes one to the sum on \lineref{check1}.

Clearly $\agx, \agy \in \coal$.
We can form an \IR\ partition $\Pi$ such that $\coal \in \Pi$ by assigning all the other agents in singleton coalitions.

Through similar reasoning, we obtain an \IR\ partition if \lineref{check2} or \lineref{check3} returns \KwYes.\\

Now assume that $I$ is a \yes\ of \controlprob{\IR}{\PA}{\AddAg}{\FE}.
Let $\sss \subseteq \agentsW$ such that $|\sss| \leq \budget$ and 
let $\Pi$ be an \IR\ partition of $\agents \cup \sss$ such that $\Pi(\agx) = \Pi(\agy)$ and let $\coal \coloneqq \Pi(\agx) = \Pi(\agy)$.
Clearly $|\coal| \geq 2$.
Thus, every agent in $\coal$ must obtain a friend.

Construct a sequence $(\agx, \agx_1, \dots, \agx_{t})$ by starting from $\agx$, picking an arbitrary friend $\agx_1$ of $\agx$ from $\coal$, picking an arbitrary friend $\agx_2$ of $\agx_1$ and so on until. We stop when we reach an agent $\agx_t$ who has a friend $\{\agx, \agx_1, \dots, \agx_{t - 1}\}$.
Let us call this friend $\hat{\agx}$.
Since $\coal$ is finite, we must reach such an agent eventually.

Now construct a sequence $(\agy, \agy_1, \dots, \agy_{s})$ similarly, except instead of stopping when we reach an agent~$\agy_{s}$ with a friend in~$(\agy, \agy_1, \dots, \agy_{s-1})$, we may also stop earlier if we encounter an agent~$\agy_s$ who has a friend among $\{\agx, \agx_1, \dots, \agx_{t}\}$. 
Similarly to the previous case, we use $\hat{\agy}$ to refer to the friend of~$\agy_s$.
This agent~$\hat{\agy}$ is either $\{\agx, \agx_1, \dots, \agx_{t}\}$ or $\{\agy, \agy_1, \dots, \agy_{s - 1}\}$.
Note that it is possible that $\agy \in \{\agx, \agx_1, \dots, \agx_{t}\}$, in which case the sequence is empty and $\hat{\agy} = \agy$.
Moreover, if $\agy$ has a friend among $(\agx, \agx_1, \dots, \agx_{t})$, then the sequence is just $(\agy)$.

By construction, all the agents in these two sequences are pairwise disjoint, and they are all in $\coal$. 

We proceed in three possible cases regarding the identity of $\hat{\agy}$.
\begin{compactitem}
  \item[Case 1] Agent~$\hat{\agy}$ is before $\hat{\agx}$ on the sequence $\agx, \agx_1, \dots, \agx_{t}$ or $\hat{\agy} = \hat{\agx}$.
First consider the set of agents in $\agx, \agx_1, \dots, \agx_{t}$ who precede $\hat{\agy}$.
Let us call this set $X_1$. 
We observe that $\wpath(\agx, \hat{\agy}) \leq |X_1 \cap \sss|$: By construction, the agents in $X_1 \cup \{\hat{\agy}\}$ form a path from $\agx$ to $\hat{\agy}$ and the weight of this path is the number of arcs whose first element is in $\agentsW$.
Note that if $X_1 = \emptyset$, then $\hat{\agy} = \agx$ and both sizes of the inequality are zero.

Similarly, consider the set of agents $Y$ that are in the sequence  $(\agy, \agy_1, \dots, \agy_{s})$.
We obtain that $\wpath(\agy, \hat{\agy}) \leq |Y \cap \sss|$, because by construction the agents in $Y \cup \{\hat{\agy}\}$ form a path from $\agy$ to~$\hat{\agy}$ and the weight of this path is the number of arcs whose first element is in $\agentsW$.
Note that if $Y = \emptyset$, then $\hat{\agy} = \agy$ and both sizes of the inequality are zero.

Now consider the set of agents on the sequence $(\agx, \agx_1, \dots, \agx_{t})$ who are after $\hat{\agy}$ (including $\hat{\agy}$) but before $\hat{\agx}$ (excluding $\hat{\agx}$).
Let us call this set $X_2$. 
We observe that $\wpath(\hat{\agy}, \hat{\agx}) \leq |X_2 \cap \sss|$, because by construction the agents in $X_2 \cup \{\hat{\agx}\}$ form a path from $\hat{\agy}$ to $\hat{\agx}$ and the weight of this path is the number of arcs whose first element is in $\agentsW$.

Finally, consider the set of agents $X_3$ on the sequence $(\agx, \agx_1,$ $\dots, \agx_{t})$ who are after $\hat{\agx}$ (including $\hat{\agx}$).
We observe that $\wcycle(\hat{\agx})$ $\leq |X_3 \cap \sss|$:
Since~$\hat{\agx}$ is, by construction, an agent such that~$\agx_t$ considers him a friend, graph $\goodG[X_3]$ contains a friendship cycle containing $\hat{\agx}$.
By construction, the weight of this path is the number of arcs on it whose first element is in $\agentsW$.
This is equivalent to the number of agents in the cycle who are in $\agentsW$.

Recall that $X_1, X_2, X_3$, and $Y$ are all pairwise disjoint sets that are subsets of $\coal$ and hence subsets of $\agents \cup \sss$.
We obtain that 
\begin{align*}\budget &\geq |\sss| \geq |\coal \cap \sss| \\
&\geq |X_1 \cap \sss| + |Y \cap \sss| + |X_2 \cap \sss| + |X_3 \cap \sss| \\
&\geq \wpath(\agx, \hat{\agy}) + \wpath(\agy, \hat{\agy}) +  \wpath(\hat{\agy}, \hat{\agx}) + \wcycle(\hat{\agx}).
\end{align*}
We obtain that \lineref{check1} returns \KwYes.

\item[Case 2] Agent~$\hat{\agy}$ is after $\hat{\agx}$ on the sequence $\agx, \agx_1, \dots, \agx_{t}$.
This implies that $\hat{\agx} \neq \hat{\agy}$, so the second part of \lineref{check2} is satisfied.
First consider the set of agents in $\agx, \agx_1, \dots, \agx_{t}$ who precede $\hat{\agx}$.
Let us call this set $X_1$. 
By reasoning analogous to the previous case we have that $\wpath(\agx, \hat{\agx}) \leq |X_1 \cap \sss|$.

Similarly, consider the set of agents $Y$ that are in the sequence  $(\agy, \agy_1, \dots, \agy_{s})$.
Again  $\wpath(\agy, \hat{\agy}) \leq |Y \cap \sss|$.

Now consider the set of agents on the sequence $(\agx, \agx_1, \dots, \agx_{t})$ who are after $\hat{\agx}$ (including $\hat{\agx}$) but before $\hat{\agy}$ (excluding $\hat{\agy}$).
Let us call this set $X_2$. 
We observe that $\wpath(\hat{\agx}, \hat{\agy}) \leq |X_2 \cap \sss|$.

Finally, consider the set of agents $X_3$ on the sequence $(\agx, \agx_1,$ $\dots, \agx_{t})$ who are after $\hat{\agy}$ (including $\hat{\agy}$).
We observe again that $\wpath(\hat{\agy}, \hat{\agx}) \leq |X_3 \cap \sss|$.

Recall that $X_1, X_2, X_3$, and $Y$ are all pairwise disjoint subsets of $\coal$, and hence subsets of $\agents \cup \sss$.
It follows that 
\begin{align*}\budget &
\geq |X_1 \cap \sss| + |Y \cap \sss| + |X_2 \cap \sss| + |X_3 \cap \sss| \\
&\geq \wpath(\agx, \hat{\agx}) + \wpath(\agy, \hat{\agy}) +  \wpath(\hat{\agx}, \hat{\agy}) + \wpath(\hat{\agy}, \hat{\agx}).
\end{align*}
We obtain that \lineref{check2} returns \KwYes.

\item[Case 3] Agent~$\hat{\agy}$ is not on the sequence $\agx, \agx_1, \dots, $ $\agx_{t}$, i.e., $\hat{\agy}$ is in $\{\agy, \agy_1, \dots, \agy_s\}$.
First consider the set of agents in $\agx, \agx_1, \dots, \agx_{t}$ who precede $\hat{\agx}$.
Let us call this set $X_1$. 
Again we have that $\wpath(\agx, \hat{\agx}) \leq |X_1 \cap \sss|$.

Now consider the set of agents $X_2$ on the sequence $(\agx, \agx_1, \dots, $ $\agx_{t})$ who are after $\hat{\agx}$ (including $\hat{\agx}$).
We observe that $\wcycle(\hat{\agx}) \leq |X_2 \cap \sss|$:
Since $\hat{\agx}$ is, by construction, an agent such that $\agx_t$ considers him a friend, graph $\goodG[X_2]$ contains a friendship cycle containing $\hat{\agx}$.
By construction, the weight of this cycle is the number of arcs--and hence agents--in it whose first element is in $\agentsW$.

Now consider the set of agents in $\agy, \agy_1, \dots, \agy_{s}$ who precede~$\hat{\agx}$.
Let us call this set $Y_1$. 
Again we have that $\wpath(\agy, \hat{\agy}) \leq |Y_1 \cap \sss|$.

Next consider the set of agents $Y_2$ on the sequence $(\agy, \agy_1, \dots, $ $\agy_{s})$ who are after $\hat{\agy}$ (including $\hat{\agy}$).
Analogously to the set $X_2$, we observe that $\wcycle(\hat{\agy}) \leq |Y_2 \cap \sss|$.

Recall that $X_1, X_2, Y_1$, and $Y_2$ are all pairwise disjoint sets that are subsets of $\coal$ and hence subsets of $\agents \cup \sss$.
Hence  
\begin{align*}\budget &
\geq |X_1 \cap \sss| + |Y_1 \cap \sss| + |X_2 \cap \sss| + |Y_2 \cap \sss| \\
&\geq \wpath(\agx, \hat{\agx}) + \wpath(\agy, \hat{\agy}) +  \wcycle(\hat{\agx}) + \wcycle(\hat{\agy}).
\end{align*}
We obtain that \lineref{check3} returns \KwYes.
\end{compactitem}

Thus, the algorithm correctly solves \controlprob{\IR}{\PA}{\AddAg}{\FE}.
By \cref{obs:pa->na}, \controlprob{\IR}{\NA}{\AddAg}{\FE} is then also polynomial-time solvable.
However, it is straightforward to prove that we can obtain a more efficient algorithm by looking for $\ag i \in \agents \setminus \{\agx\}$ such that $\wpath(\agx, \ag i) + \wcycle(\ag i) \leq \budget$.
Finally, \cref{obs:fri:ir:is} implies that the algorithm can also be used to solve \controlprob{\IS}{\PA}{\AddAg}{\FE} and \controlprob{\IS}{\NA}{\AddAg}{\FE}.
}

\citet{Brandt24_AI_NashStable} show that determining the existence of an \NS\ partition is \NP-hard for \FE, and thus \controlprob{\NS}{\G}{\A}{\FE} is \NP-hard for every $\A \in \{\AddAg, \DelAg\}, \G \in \{\NA, \PA\}$ by \cref{obs:originalhard-follows}.
Moreover, if an instance of \FE\ is a DAG, then by \cref{obs:DAG-FE-NO} we cannot hope to have any pair of agents in the same coalition in a stable partition.
However, we obtain a non-trivial--although simple--algorithm when the preferences are symmetric for adding agents, whereas deleting agents is immune (see \cref{prop:IMMUNE-DELAG}).

\begin{restatable}[\appsymb]{proposition}{propSymFeIRISNSNaAdd}
  For symmetric preferences, \controlprob{\NS}{\NA}{\AddAg}{\FE} and \controlprob{\NS}{\PA}{\AddAg}{\FE} are polynomial-time solvable.
  \label{prop:SYM-FE-IRISNS-NA-ADD}
\end{restatable}

\begin{proof}[Proof sketch]
  For symmetric friendship relations, the following partition is \NS: Put all agents who have at least one friend together, while all agents without any friends form singleton coalitions.
  Due to this, the problem of ensuring $x$ and $y$ are together reduces to adding at most two agents to the original graph and checking whether afterwards $x$ and $y$ will each have a friend.
\end{proof}
\appendixproofwithstatement{prop:SYM-FE-IRISNS-NA-ADD}{\propSymFeIRISNSNaAdd*}
{
We first note that when the preferences are symmetric, the partition that contains all the agents with degree at least one in one coalition $\megacoal$, and the remaining agents in singleton coalitions, is \NS.
Since preferences are symmetric, every agent who has friends has all of his friends in $\megacoal$, and thus cannot prefer any other coalition to $\megacoal$.
An agent who has no friends cannot prefer any other coalition to the singleton coalition containing him.
Moreover, there is no \NS\ partition where an agent who has no friends is in a coalition other than the coalition containing only him.

Thus, to enforce that $\agx$ and $\agy$ are together in an \NS\ partition, it is sufficient and necessary to make them have a friend.
Hence, we can determine \controlprob{\NS}{\PA}{\AddAg}{\FE} as follows:
\begin{compactenum}
\item[Case 1] $\budget = 0$. Return yes if and only if both $\agx$ and $\agy$ have a friend in $\agentsU$.
\item[Case 2] $\budget = 1$. If both $\agx$ and $\agy$ have a friend in $\agentsU$, then return yes. If only $\agx$ (resp.\ only $\agy$) does not have a friend in $\agentsU$, then return yes if $\agentsW$ contains an agent who is a friend of $\agx$ (resp.~$\agy$).
If neither $\agx$ nor $\agy$ has a fried in $\agentsU$, then return yes if $\agentsW$ contains an agent who is a friend of both $\agx$ and $\agy$.
If the above does not return yes, return no.
\item[Case 3] $\budget \geq 2$. If both $\agx$ and $\agy$ have a friend in $\agentsU$, then return yes. If only $\agx$ (resp.\ only~$\agy$) does not have a friend in $\agentsU$, then return yes if $\agentsW$ contains an agent who is a friend of $\agx$ (resp.~$\agy$).
If neither $\agx$ nor $\agy$ has a fried in $\agentsU$, then return yes if $\agentsW$ contains an agent who is a friend of $\agx$ and an agent who is a friend of $\agy$.
If the above does not return yes, return no.
\end{compactenum}

For \controlprob{\NS}{\NA}{\AddAg}{\FE} it is sufficient to check if $\agx$ has a friend among $\agentsU$ when $\budget = 0$, or among $\agentsU \cup \agentsW$ when $\budget \geq 1$.
}

For core stability, we can solve the control goal~\PA\ efficiently by reducing to finding a minimum-weight subgraph where $\agx$ and $\agy$ are mutually reachable.
This is connected to the \twoDSN\ problem, which admits a polynomial-time algorithm~\cite{FeldmanRuhl_SCSS,li1992point}:

\probdef{Directed Steiner Network (\twoDSN)}
{A directed graph $G=(V, A)$ with an arc-weighting function~$\weight\colon A\to \mathds{R}$,
  two pairs $(s_1, t_1), (s_2, t_2)$, and $\delta \in \mathds{R}$.}
{Is there a subgraph~$H=(V',A')$ of $G$
  which contains a path from~$s_1$ to $t_1$ and from $s_2$ to $t_2$, respectively, such that $\sum_{a \in A'}\weight(a)$ $ \leq~\delta$?}

\begin{restatable}[\appsymb]{theorem}{thmFECSPAADD}
\controlprob{\CS}{\NA}{\AddAg}{\FE} and \controlprob{\CS}{\PA}{\AddAg}{\FE} are polynomial-time solvable.
\label{thm:FE-CS-PA-ADD}
\end{restatable}

\newcommand{\agu}[1]{\ensuremath{u_{#1}}}
\newcommand{\agwI}[1]{\ensuremath{w^1_{#1}}}
\newcommand{\agwO}[1]{\ensuremath{w^2_{#1}}}
\newcommand{\agux}{\agu x}
\newcommand{\aguy}{\agu y}

\appendixproofwithstatementsketch{thm:FE-CS-PA-ADD}{\thmFECSPAADD*}{

By \cref{obs:csstogether}, two agents $\agx$ and $\agy$ can be in the same coalition in a \CS\ partition if and only if they are in the same strongly connected component of the friendship graph, i.e., if there is a path from $\agx$ to $\agy$ and vice versa.
This is closely related to the polynomial-time solvable \twoDSN-problem~\cite{FeldmanRuhl_SCSS,li1992point}.

The main difference is that \twoDSN\ has arc-weights, whereas we have vertex weights.
To address this, we replace every agent in $\agentsW$ with two agents, one taking care of the in-arcs and the other the out-arcs.
Let $((\agents = \agentsU \cup \agentsW, \goodG), \agx, \agy, \budget)$ be an instance of \controlprob{\CS}{\PA}{\AddAg}{\FE}.
We will transform this instance into an instance $I = (G, \weight, (s_1, t_1), (s_2, t_2))$ of \twoDSN\ as follows:
\begin{compactitem}[--]
\item For every $\ag i \in \agentsU$, we add to~$G$ a vertex~$\agu i$.
\item For every $\ag i \in \agentsW$, we add to~$G$ two vertices, $\agwI{i}$ and $\agwO{i}$.
\end{compactitem}
The idea is that $\agwI{i}$ inherits the in-arcs of $\ag i \in \agentsW$ and $\agwO{i}$ the out-arcs.
Formally, we construct the following arcs with weights and add them to~$G$. 
\begin{compactitem}[--]
\item For every arc $(\ag i, \ag j) \in A(\goodG)$, we add to~$G$ an arc $(u, v)$ with weight $0$, where
\begin{inparaitem}[$\bullet$]
\item if $\ag i \in \agentsU$, then $u \coloneqq \agu i$, otherwise $u \coloneqq \agwO{i}$ and
\item if $\ag j \in \agentsU$, then $v \coloneqq \agu j$, otherwise $v \coloneqq \agwI{j}$.
\end{inparaitem}
\item For every $\ag i \in \agentsW$, we construct the arc $(\agwI{i}, \agwO{i})$ with weight~$1$.
\end{compactitem}
We set $(s_1, t_1) \coloneqq (\agux, \aguy)$ and $(s_2, t_2) \coloneqq (\aguy, \agux)$. Moreover, we set $\delta \coloneqq \budget$.
The correctness is deferred to the appendix.}{
We continue the proof by showing the correctness of our reduction.
\begin{claim}\label{prop:pacsaddfe1}
If  $((\agents = \agentsU \cup \agentsW, \goodG), \agx, \agy, \budget)$ is a \yes\ of \controlprob{\CS}{\PA}{\AddAg}{\FE}, then $I$ is a \yes\ of \twoDSN.
\end{claim}
\begin{claimproof}{prop:pacsaddfe1}
Let $\sss \subseteq \agentsW$ such that $|\sss| \leq \budget$ and $\agentsU \cup \sss$ admits a \cst\ partition $\Pi$ such that $\Pi(\agx) = \Pi(\agy)$.
We construct a subgraph $H$ of $G$ as follows: let $V' \coloneqq \{\agu i \mid \ag i \in \agentsU\} \cup \{\agwI{i}, \agwO{i} \mid \ag i \in \sss\}$ and let $H \coloneqq G[V']$.

Let us first show that $\agux$ is reachable from $\aguy$ and vice versa.
By \cref{obs:csstogether}, we must have that $\agx$ and $\agy$ belong to the same strongly connected component of $\goodG[\agentsU \cup \sss]$.
In other words, agent~$\agx$ must be reachable from $\agy$ and vice versa on $\goodG[\agentsU \cup \sss]$. 
Let $\agx, \ag {i_1}, \dots, \ag{i_t}, \agy$ be a path from $\agx$ to $\agy$ on $\goodG[\agentsU \cup \sss]$. 
We construct a path from $\agux$ to $\aguy$ on $H$. For every $\ell \in [t]$:
\begin{compactitem}[--]
\item  If $\ag{i_\ell} \in \agentsU$, add the vertex $\agu{i_\ell}$ to the path. Since $\ag{i_\ell} \in \agentsU$, we have that $\agu{i_\ell} \in V'$.
 \item If $\ag{i_\ell} \in \agentsW$, add the vertices $\agwI{i_\ell},\agwO{i_\ell}$ to the path. Since $\ag{i_\ell} \in \sss$, we have that $\agwI{i_\ell},\agwO{i_\ell} \in V'$.
\end{compactitem}
Consider every successive pair $(u, v)$ of vertices on the constructed path.

If both $u$ and $v$ correspond to vertices in $\agentsU$, then clearly there is an arc from $u$ to $v$ by construction, because $u$ and $v$ inherit all the in- and out-arcs of the corresponding vertices.

If $u = \agwI{i_{\ell}}$ for some $\ell \in [t], \ag {i_{\ell}} \in \agentsW$, then by construction, it must be that $v = \agwO{i_\ell}$, because $\agwI{i_{\ell}}$ has no other out-neighbors. 
By construction, arc $(u, v)$ clearly exists on $H$.
Similarly, if $v = \agwO{i_\ell}$ for some $\ell \in [t], \ag {i_{\ell}} \in \agentsW$, then it must be that $u = \agwI{i_\ell}$, because $\agwO{i_{\ell}}$ has no other in-neighbors. 

Now consider the case when $u = \agwO{i_{\ell}}$ for some $\ell \in [t], \ag {i_{\ell}} \in \agentsW$. 
If $v$ corresponds to agent $\hat{\ag{v}} \in \agentsU$, then by construction $H$ has an arc from $u$ to $v$ if and only if there is an arc from $\ag {i_{\ell}}$ to $\hat{\ag{v}} $ in $\goodG$.
Since $\hat{\ag{v}}$ and  $\ag {i_{\ell}}$ are consecutive agents on the path from $\agx$ to $\agy$, such an arc must exist.
Similarly, if $v$ instead corresponds to an agent $ \ag {i_{\ell + 1}} \in \agentsW$, then by construction $v = \agwI{i_{\ell + 1}}$ and $H$ has an arc from $u$ to $v$ if and only if there is an arc from $\ag {i_{\ell}}$ to $\ag {i_{\ell + 1}}$ on $\goodG$.
Since they are on the path from $\agx$ to $\agy$, such an arc must exist.

The case where $v = \agwI{i_{\ell}}$  for some $\ell \in [t], \ag {i_{\ell}} \in \agentsW$ is analogous.

We have now shown that $H$ contains a path from $\agux$ to $\aguy$. We can analogously show that $H$ contains a path from $\aguy$ to $\agux$.

It remains to show that the total weight of the arcs in $H$ does not exceed $\delta = \budget$.
Observe that the only arcs with a positive weight in $H$ are $\{(\agwI{i}, \agwO{i}) \mid i \in \sss\}$ and each of those has weight one.
Clearly $|\{(\agwI{i}, \agwO{i}) \mid i \in \sss\}| = |\sss| \leq \budget = \delta$, as required. 
\end{claimproof}

\begin{claim}\label{prop:pacsaddfe2}
	If   $I$ is a \yes\ of \twoDSN, then $((\agents = \agentsU \cup \agentsW, \goodG), \agx, \agy, \budget)$ is a \yes\ of \controlprob{\CS}{\PA}{\AddAg}{\FE}.
\end{claim}
\begin{claimproof}{prop:pacsaddfe2}
Let $H=(V',A')$ be a subgraph of $G$ that contains a path from $\agux$ to $\aguy$ and vice versa and $\sum_{a \in A'}\weight(a) \leq \delta$.
We construct a subset $\sss \subseteq \agentsW$ as follows:
For every $a \in A'$, if $a = (\agwI{i}, \agwO{i})$ for some $\ag i \in \agentsW$, then add $\ag i$ to $\sss$.
Observe that since the arcs that are $(\agwI{i}, \agwO{i})$ for some $\ag i \in \agentsW$ are the only arcs in $G$ with non-zero weight, and their weight is 1, we must have that $|\sss| \leq \delta = \budget$, as required.

Let us first show that $\goodG[\agentsU \cup \sss]$ admits a \cst\ partition where $\agx$ and $\agy$ are in the same coalition.
By \cref{obs:csstogether}, it is sufficient to show that $\goodG[\agentsU \cup \sss]$ admits a path from $\agx$ to $\agy$ and vice versa.
Let $\agux, v_1, \dots, v_t, \aguy$ be a path from $\agux$ to $\aguy$ in $H$.

If $v_1$ corresponds to an agent $\ag i \in \agentsU$, then by construction $\goodG[\agentsU \cup \sss]$ must contain the arc $(\agx, \ag i)$.
Since $\agx \in \agentsU$, we can continue the reasoning starting from $v_1$.

If $v_1 = \agwI{\ag i}$ for some $\ag i \in \agentsW$, then we must have that $v_2 = \agwO{\ag i}$, because $ \agwI{\ag i}$ only has an out-arc to $ \agwO{\ag i}$ and  $ \agwI{\ag i}$ is on a path to $\aguy$. Then $\ag i \in \sss$ and moreover $\goodG$ contains the arc $(\agx, \ag i)$.
Since $\ag i \in \agentsU \cup \sss$, we can continue the reasoning starting from $v_2$; recall that $\agwO{\ag i}$ inherits all the out-arcs of $\ag i$.

By repeating this reasoning, we obtain that there must be a path from $\agx$ to $\agy$ in $\agentsU \cup \sss$. Analogous reasoning shows us that there is a path from $\agy$ to $\agx$ as well.
\end{claimproof}

By \cref{prop:pacsaddfe1,prop:pacsaddfe2} we can construct every instance of \controlprob{\CS}{\PA}{\AddAg}{\FE} to an equivalent instance of \twoDSN.
As \twoDSN\ is solvable in time $O(mn + n^2 \log(n))$~\cite{FeldmanRuhl_SCSS,li1992point,natu1997point}, \controlprob{\CS}{\PA}{\AddAg}{\FE} is solvable in time $O(n^3)$, where $n$ is the number of agents.
\cref{obs:gra:eq} shows that  \controlprob{\CS}{\NA}{\AddAg}{\FE} is then also polynomial-time solvable, in particular solvable in time $O(n^4)$.
However, the algorithm obtained through \cref{obs:gra:eq} is suboptimal: We can obtain a $n^2$ running time by using Dijkstra's algorithm to find the minimum-weight cycle containing $\agx$.
}

\newcommand{\scs}{\ensuremath{h}}

We finally look into enforcing that grand coalition is stable.
For every $\S \in \{$\IR, \IS, \NS, \CS$\}$, we discover an algorithm for \controlprob{S}{\GR}{\AddAg}{\FE} that runs in time $|\agentsW|^{\budget} |\agents|^{O(1)}$.
However, we show that the problem is W[2]-hard wrt.~$\budget$, i.e., it is unlikely to admit an algorithm that runs in time  $f(\budget) \cdot |\agents|^{O(1)}$, where $f$ is some computable function.
The result is through a straightforward reduction from \scp, which is W[2]-hard wrt.\ the set cover size $\scs$~\cite{DF13}.

\begin{restatable}[\appsymb]{theorem}{thmFeIRGrADD}
 For every $\S \in \{$\IR, \IS, \NS, \CS$\}$, \controlprob{S}{\GR}{\AddAg}{\FE} is W[2]-hard wrt. $\budget$ even when the preference graph is symmetric and in \XP\ wrt.\ $\budget$.
  \label{thm:FE-IR-GR-ADD}
\end{restatable}

\appendixproofwithstatement{thm:FE-IR-GR-ADD}{\thmFeIRGrADD*}{
We first show that \controlprob{S}{\GR}{\AddAg}{\FE} is W[2]-hard wrt.~$\budget$ for every $\S \in \{$\IR, \IS, \NS, \CS$\}$.

As mentioned, we reduce from the following problem, which is W[2]-hard wrt. the solution size~$\scs$:
 \decprob{\scp}
  {A $\nn$-element set $\els = [\nn]$ and a collection~$\sets = \{\set 1, \dots, \set{m}\}$ of subsets of~$\els$, an integer $\scs$.}{Does~$\sets$ contain a \myemph{cover} of size at most \scs\ for~$\els$,  i.e., a subcollection~$\excov \subseteq \sets$ such that $|\excov| \leq \scs$ and $\cup_{\set{j}\in \excov}\set{j} = \els$?}

Let $I = (\els = [\nn], \sets,$ $\scs)$ be an instance of \scp.
Let us construct an instance of \FE\ with an agent set $\agents = \agentsU \cup \agentsW$ and friendship graph $\goodG$ as follows:
Let $\agentsU \coloneqq \{u_i \mid i \in [\nn]\}$ and $\agentsW \coloneqq \{s_j \mid \set j \in \sets\}$.
For every $i \in [\nn]$, agent $u_i$ is mutually friends with $s_j$ if and only if $i \in \set j$.
Moreover, the agents in $\{s_j \mid \set j \in \sets\}$ are all friends with each other.
The preferences are symmetric.
Moreover, we set $\budget \coloneqq \scs$ to complete the construction.

\begin{claim}If~$I$ is a \yes\ of \scp, then there is $\sss \subseteq \agentsW$ such that $|\sss| \leq \scs$ and the grand coalition partition of $\agentsU \cup \sss$ is \CS.\label{clm1:FE-IR-GR-ADD}
\end{claim}
\begin{claimproof}{clm1:FE-IR-GR-ADD}
Let $\excov$ be a cover of $\els$ such that $|\excov| \leq \scs$.
Let $\sss \coloneqq \{s_j \mid \set j \in \excov\}$.
Clearly $|\sss| = |\excov| \leq \budget$.
It remains to show that the grand coalition partition is a \CS\ partition of $\agentsU \cup \sss$.
Recall that by  \citet{dimitrov2006simple} it is sufficient to show that $\goodG[\agentsU \cup \sss]$ is strongly connected.
Since~$\excov$ is a cover of $[\nn]$, for every $i \in [\nn]$ there is a set $\set j \in \excov$ such that $i \in \set j$.
Thus $u_i$ has a mutual friend $\set j$ among $\agentsU \cup \sss$.
Since the $\goodG[\sss]$ is a clique, we must have that $\goodG[\agentsU \cup \sss]$ is strongly connected, as required.
\end{claimproof}

\begin{claim}If there is $\sss \subseteq \agentsW$ such that $|\sss| \leq \budget$ and the grand coalition partition of $\agentsU \cup \sss$ is \IR, then $I$ is a \yes\ of \scp.\label{clm2:FE-IR-GR-ADD}
\end{claim}
\begin{claimproof}{clm2:FE-IR-GR-ADD}
Let $\excov \coloneqq \{\set j \in \sets \mid s_j \in \sss\}$.
Clearly $|\excov| = |\sss| \leq \scs$.
It remains to show that $\excov$ is a cover of~$\els$.
Since grand coalition partition is \IR, every agent must obtain a friend, as otherwise he would prefer being alone.
Observe that for every $ i \in [\nn]$, agent~$u_i$ must have a friend among $\agentsU \cup \sss$.
By construction the only friends of $u_i$ are $s_j$ such that $\set j \in \sets, i \in \set j$.
Thus $\excov$ must contain a set $\set j$ such that $i \in \set j$.
As this holds for every $i \in \els,$ ~$\excov$ is a cover of $\els$.
\end{claimproof}

\cref{clm1:FE-IR-GR-ADD,clm2:FE-IR-GR-ADD} show that for every $\S \in \{$\IR, \IS, \NS, \CS$\}$, \controlprob{S}{\GR}{\AddAg}{\FE} is W[2]-hard wrt. $\budget$:
If $(\els, \sets, \scs)$ is a \yes, then by \cref{clm1:FE-IR-GR-ADD} $((\agents = \agentsU \cup \agentsW), \budget)$ is also a \yes; recall that by definition \CS\ implies \IR\ and by \cref{obs:gra:eq} the grand coalition partition is \IR\ if and only if it is also \IS\ and \NS.
Next, if $((\agents = \agentsU \cup \agentsW), \budget)$ is a \yes, then so is $((\els, \sets), \scs)$ by \cref{clm2:FE-IR-GR-ADD}.
We can use \cref{clm2:FE-IR-GR-ADD} because a partition that is \IS, \NS, or \CS\ must also be \IR.
This concludes showing that \controlprob{S}{\GR}{\AddAg}{\FE} is W[2]-hard wrt. $\budget$.\\

Next we show that  \controlprob{\S}{\GR}{\AddAg}{\FE} is in \XP\ wrt.\ $\budget$ for every $\S \in \{$\IR, \IS, \NS, \CS$\}$.
Observe that we can verify in polynomial time whether the grand coalition partition is \S: For \IR, \IS, and \NS\ this is clear from the definition, and by \cref{obs:csstogether} to verify \CS\ it is sufficient to check that the friendship graph of the grand coalition partition is strongly connected.
Thus we can try every subset of $\agentsW$ of cardinality at most $\budget$.
The number of such subsets is in $O(|\agentsW|^{\budget})$, and thus this simple brute-force algorithm is in \XP\ wrt.~$\budget$.
}

In contrast to \controlprob{\S}{\GR}{\AddAg}{\FE}, \controlprob{\S}{\GR}{\DelAg}{\FE} is solvable in polynomial time for every $\S \in \{\IR,\IS,\NS, \CS\}$.
For each stability concept, we can determine in polynomial time a maximum subset of agents such that the grand coalition partition of it is stable.
For \IR, \IS, and \NS, we obtain this set by recursively removing agents who have no friends, as they cannot be contained in an \IR\ partition.
For \CS, it is known that every coalition in a stable partition is a strongly connected component of the friendship graph.
Thus we only need to keep a largest strongly connected component.

\begin{restatable}[\appsymb]{proposition}{propSymFeIRISNSGrDel}
For every stability concept $\S \in \{\IR,\IS,\NS, \CS\}$, \controlprob{\S}{\GR}{\DelAg}{\FE} is polynomial-time solvable.
  \label{prop:FE-IRISNS-GR-DEL}
\end{restatable}

\appendixproofwithstatement{prop:FE-IRISNS-GR-DEL}{\propSymFeIRISNSGrDel*}{
Recall that by \cref{obs:gra:eq} \IR, \IS, and \NS\ are equivalent for grand coalition partition.
Thus  \controlprob{\IR}{\GR}{\DelAg}{\FE},  \controlprob{\IS}{\GR}{\DelAg}{\FE}, and  \controlprob{\NS}{\GR}{\DelAg}{\FE} are equivalent.

To solve \controlprob{\IR}{\GR}{\DelAg}{\FE}, we start with the agent set $\agentsU$ and recursively remove every agent who does not have any friends among the remaining agents.
We repeat this as long as every remaining agent has at least one friend.
If the number of removed agents is at most $\budget$, then we return yes, otherwise no.

For correctness, first observe that after the recursive removal, every agent has a friend among the remaining agents.
Thus the grand coalition partition is \IR.
Also observe that a grand coalition partition of a set of agents is \IR\ if and only if everyone has a friend in it.
No agent who is removed can be in a coalition where every agent has a friend.
Hence the algorithm removes the smallest number of agents such that the grand coalition partition of the remaining agents is \IR.

 Initially checking whether an agent has no friends can be done in time $O(|\agents|)$. When recursively deleting agents, we touch any edge only when we delete it. For each edge deletion, we can check whether someone who considers him a friend should be removed. Hence the number of operations during the deletion process is in $O(|A(\goodG)|)$ and the algorithm's time complexity is $O(|\agents| + |A(\goodG)|)$.\\
 
Next, let us describe an algorithm for \controlprob{\CS}{\GR}{\DelAg}{\FE}.
We compute the strongly connected components of $\goodG$, select a largest-cardinality strongly connected component, and remove all the other agents.
If the cardinality of the removed agents is at most $\budget$, we return yes, otherwise we return no.

For correctness, recall that by \cref{obs:csstogether} grand coalition partition is \CS\ if and only if it is connected.
Observe that after the removal, the remaining set of agents form a strongly connected component of $\goodG$.
Moreover, since removing agents cannot make a pair of agents connected, the agents in a largest strongly connected component is a largest subset of agents such that the grand coalition partition of them is \CS.

The algorithm runs in the same time as computing strongly connected components.
We can recognize strongly connected components in time $O(|\agents| + |A(\goodG)|)$.
}

\section{Additive Preferences}\label{sec:add-setting}
\appendixsection{sec:add-setting}

In this section, we consider the case with additive preferences.
It turns out that most of our control problems remain intractable even in very restricted cases such as when the budget is zero or the preference graph is a DAG or symmetric.
As in the previous section, we first consider \NA, then~\PA, and finally \GR.

First, we show that even the most basic stability requirement \IR\ becomes difficult to determine once we enforce some control goal.

\appendixfigure{fig:thm:AD-IR-NA}{}{Figure for \cref{thm:AD-IR-NA}}{
  Sketch of of the preference graph for \cref{thm:AD-IR-NA}; see \cref{fig:thm:AD-IR-NA}.
  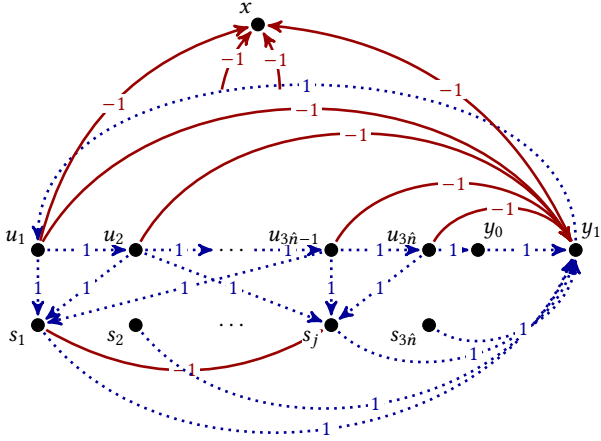
\begin{figure}[t!]
  \centering
  \begin{tikzpicture}[>=stealth',shorten <= 1pt, shorten >= 1pt]
    \path[use as bounding box,draw=none] (0.5,-4) rectangle (12,2);
    \def \xx {1.3}
  \def \xy {1}
  \foreach \x / \n in {1/1, 2/2, 4/j, 5/3n} {
    \node[unode] at (\x*\xx, -2.5*\xy) (s\n) {};
  }

  \node[unode] at (3.25*\xx, 1.5*\xy) (x) {};

  \foreach \x / \n in {1/1,2/2, 4/3n1, 5/3n} {
    \node[unode] at (\x*\xx, -1.5*\xy) (u\n) {};
  }

  \node[unode] at (5.5*\xx, -1.5*\xy) (y0) {};
  \node[unode] at (6.5*\xx, -1.5*\xy) (y1) {};

  \coordinate (hx1) at ($(x)!0.3!(u1)$);
  \coordinate (hx2) at ($(x)!0.3!(u2)$);
  \coordinate (hx3) at ($(x)!0.3!(u3n1)$);
  \coordinate (hx4) at ($(x)!0.3!(u3n)$);

  \draw[->, ec] (u1) to[bend left=30] node[midway, fill=white, inner sep=1pt] {\small $-1$} (x);
  \draw[->, ec] (hx2) to[bend left=15] node[midway, fill=white, inner sep=1pt] {\small $-1$} (x);
  \draw[->, ec] (hx3) to[bend right=15] node[midway, fill=white, inner sep=1pt] {\small $-1$} (x);
  \draw[->, ec] (y1) to[bend right=30] node[midway, fill=white, inner sep=1pt] {\small $-1$} (x);

  \foreach \i / \p / \a / \r / \n in {%
    x/above left/2pt/0pt/{x},
    u1/above left/0pt/2pt/{u_1},   u2/above left/0pt/2pt/{u_2},   u3n1/above left/0pt/2pt/{u_{3\nn-1}},  u3n/above left/0pt/2pt/{u_{3\nn}}, 
    y0/above right/2pt/0pt/{y_0}, y1/above right/2pt/0pt/{y_1},
    s1/below left/0pt/2pt/{s_1},   s2/below left/0pt/2pt/{s_2},   sj/below left/0pt/2pt/{s_j},  s3n/below left/0pt/2pt/{s_{3\nn}},
  } {
     \node[\p = \a and \r of \i,fill=white,inner sep=0.1pt] {$\n$};
   }

  \path (s2) -- node[midway, inner sep=1pt] (sdots) {$\dots$} (sj);

  \path (u2) -- node[midway, inner sep=1pt] (udots) {$\dots$} (u3n1);

  \begin{pgfonlayer}{bg}
    \foreach \s / \t / \aa / \w / \typ in {
    u1/u2/0/{1}/fc, u2/udots/0/{1}/fc, udots/u3n1/0/{1}/fc, u3n1/u3n/0/{1}/fc, u3n/y0/0/{1}/fc, y0/y1/0/{1}/fc,
    s1/sj/30/{-1}/ec,
    y1/u1/90/{1}/fc,
    u1/y1/-60/{-1}/ec, u2/y1/-60/{-1}/ec, u3n1/y1/-60/{-1}/ec, u3n/y1/-60/{-1}/ec,
    s1/y1/60/{1}/fc, s2/y1/60/{1}/fc, sj/y1/60/{1}/fc, s3n/y1/60/{1}/fc,
    u1/s1/0/{1}/fc, u2/s1/0/{1}/fc, u3n1/s1/0/{1}/fc, u2/sj/0/{1}/fc, u3n1/sj/0/{1}/fc, u3n/sj/0/{1}/fc%
    }
    {
      \draw[->, \typ] (\s) edge[bend right = \aa] node[fill=white, inner sep=1pt, midway] {\small $\w$} (\t);
    }
  \end{pgfonlayer}
\end{tikzpicture}
\caption{Illustration of the relevant part of the reduction for \cref{thm:AD-IR-NA}, assuming that $\set{1} = \{1,2,3\nn-1\}$ and $\set{j} = \{2,3\nn-1,3\nn\}$. The remaining negative utilities to $\agx$ and the negative utilities of the set-agents are omitted in the figure.}\label{fig:thm:AD-IR-NA}
\end{figure}
}

\begin{restatable}{theorem}{thmADIRNANPhard}
  For each control action~$\A\in \{\AddAg, \DelAg\}$, \controlprob{\IR}{\NA}{\A}{\AD} is \NP-complete; NP-hardness remains even if the budget is $\budget=0$ and the preference graph has only one feedback arc.
  \label{thm:AD-IR-NA}
\end{restatable}

\begin{proof}
  The \NP-containment follows directly from \cref{obs:complexityupperbounds}.
  Moreover, since $\budget = 0$, both control problems reduce to checking whether there is an \IR\ partition with agent~$x$ not alone in his coalition. 
  It remains to show the NP-hardness for the mentioned restriction.
  In the remainder of the proof, we will focus on this restricted variant and reduce from the following well-known \NP-complete problem~\cite{Gonzalez1985}.
  \decprob{\RETCf~(\RETC)}
  {A $3\nn$-element set $\els = [3\nn]$ and a collection~$\sets = \{\set 1, \dots, \set{3\nn}\}$ of $3$-element subsets of~$\els$ such that each element appears in \emph{exactly} three members of $\sets$.}{Does~$\sets$ contain an \myemph{exact cover} for~$\els$,  i.e., a subcollection~$\excov \subseteq \sets$ such that $|\excov|=\nn$ and $\cup_{\set{j}\in \excov}\set{j} = \els$?}

  Let $I = (\els = [3\nn], \sets)$ be an instance of \RETC. 
  We construct an \AD-instance whose preference graph has one feedback arc.
  The set of agents is $\agentsU = \{x, \axag{0}, \axag{1}\} \cup \{\elag{i} \mid i \in \els\} \cup \{\sset{j} \mid \set{j} \in \sets\}$, where
  \begin{compactitem}[--]
    \item $\agx$ is the special agent not to be alone,
    \item $\axag{0},\axag{1}$ are auxiliary-agents,
    \item each element~$i\in \els$ has one \myemph{element-agent~$\elag{i}$},
    \item each set~$\set{j}\in \sets$ has one \myemph{set-agent~$\sset{j}$}.
  \end{compactitem} 
  We construct the utilities as follows; the unmentioned utilities are~$0$:
  \begin{compactitem}[--]
    \item For each agent~$a \in \agentsU\setminus \{x\}$, set $\util{a}(x) = -1$.
    \item For each element~$i \in [3\nn-1]$, set $\util{\elag{i}}(\elag{i+1}) = 1$, for element-agent $\elag{3\nn}$, set $\util{\elag{3\nn}}(\axag{0}) = 1$.
    \item For auxiliary-agent $\axag{0}$, set $\util{\axag{0}}(\axag{1}) = 1$, for $\axag{1}$, set $\util{\axag{1}}(\elag{1}) = 1$.
    \item For each element~$i \in [3\nn]$, set $\util{\elag{i}}(\axag{1}) = -1$.
    \item For each set $\set{j} \in \sets$ and all elements $i \in \set{j}$, set $\util{\elag{i}}(\sset{j}) = 1$.
    \item For each set $\set{j} \in \sets$, set $\util{\sset{j}}(\axag{1}) = 1$.
    \item For all $\set{j}, \set{\ell} \in \sets$ with $j < \ell$ and $\set{j} \cap \set{\ell} \neq \emptyset$, set $\util{\sset{j}}(\sset{\ell}) = -1$.
  \end{compactitem}

  \ifcr
  \else
    An illustration of the preference graph of the \ADD-instance is depicted in \cref{fig:thm:AD-IR-NA}.
  \fi
  To complete the construction, we define $\agx$ as the agent not being alone.
  Before we show the correctness, we first observe that by deleting the arc $(\axag{1},\elag{1})$ we obtain a \DAG\ with following topological order: $\elag{1},\dots, \elag{3\nn}, \allowbreak \sset{1}, \dots, \sset{3\nn}, \allowbreak \axag{0}, \axag{1},\allowbreak x$.

  \begin{restatable}[\appsymb]{claim}{claimADIRNAForward}
   \label{clm:thm:AD-IR-NA-Forward}
   If $\excov$ is an exact cover for~$I$, then the following partition~$\Pi$
   with $\Pi(\agx) = \{\agx, \elag{1}, \elag{2}, \dots, \elag{3\nn}, \axag{0},\axag{1}\} \cup \{\sset{j} \mid \set{j} \in \excov\}$ and $\Pi(\sset{j}) = \{\sset{j}\}$ for all $\set{j} \notin \excov$ is \IR.
 \end{restatable}
 \appendixproofwithstatement{clm:thm:AD-IR-NA-Forward}{\claimADIRNAForward*}{
  To show the statement, we show that every agent has zero utility towards his coalition under $\Pi$.

  Let us first consider agent $\agx$, he has no preference to any of the agents, thus his utility towards his coalition $\Pi(x)$ is 0.

  Next we show that every element-agent~$\elag{i}, i \in [3\nn]$, has zero utility towards his coalition. 
  All have -1 utility towards $\agx$ and $\axag{1}$. 
  Every element-agent is together with $\elag{i+1}$ for which $\elag{i}$ has positive utility, except $\elag{3\nn}$ has positive utility towards $\axag{0}$.
  Since $\excov$ is an exact cover, there must be a set $\set{j} \in \excov$ such that $i \in S_j$.
  Thus, $\util{\elag{i}}(\Pi(\elag{i})) = \util{\elag{i}}(\{x, \axag{1}, \elag{i+1}, \sset{j}\}) = -1 -1 + 1 +1 = 0$ for $i \in [3\nn -1]$ and $\util{\elag{3\nn}}(\Pi(\elag{3\nn})) = \util{\elag{3\nn}}(\{x, \axag{1}, \axag{0}, \sset{j}\}) = -1 - 1 + 1 + 1 = 0$.

  Auxiliary-agents $\axag{0},\axag{1}$ both have negative utility towards~$\agx$, but~$\axag{0}$ has positive utility towards~$\axag{1}$ and $\axag{1}$ has positive utility towards $\elag{1}$. 
  Thus, both have zero utility towards their coalition.

  All set-elements $\sset{j}$ whose corresponding set $\set{j} \in \excov$ have negative utility towards $x$ and positive utility towards $\axag{1}$.
  Hence, set~$\excov$ is an exact cover, all sets in $\excov$ are pairwise disjoint, ensuring that no set-agent $\sset{j}$ has negative utility towards any other set-agent in $\Pi(\sset{j})$.
  It holds that $\util{\sset{j}}(\Pi(\sset{j})) = \util{\sset{j}}(\{\agx, \axag{1}\}) = -1 + 1 = 0$ and so $\Pi(\sset{j})$ is \IR\ for every set-agent from $\excov$.

  Clearly, every set-agent $\sset{j}$ with $\set{j} \notin \excov$ has zero utility towards his singleton coalition.
  This concludes the proof of $\Pi$ being \IR.
 }

 \begin{restatable}{claim}{claimADIRNABackward}
   \label{clm:thm:AD-IR-NA-Backward}
   If $\Pi$ is an \IR\ partition with $|\Pi(\agx)| \geq 2$, then the sets corresponding to the set-agents in $\Pi(\agx)$ form an exact cover of $I$.
 \end{restatable} 
\begin{claimproof}{clm:thm:AD-IR-NA-Backward}

  Let $C \coloneqq \Pi(x)$ be the coalition containing special agent~$\agx$ with $|C| \geq 2$.
  We want to show that the sets corresponding to set-agents contained in $C$ form an exact cover of $I$.
  We start by showing that auxiliary-agent $\axag{1}$ must be in $C$.

  If $\sset{j} \in C$ for some $\set j \in \sets$, then he needs to be in the same coalition with $\axag{1}$ for $C$ to be \IR, because he has negative utility towards $\agx$ and the only positive utility is towards $\axag{1}$.
  The same argument holds if $\axag{0} \in C$: He has negative utility towards $x$ and his only positive utility towards $\axag{1}$.

  If an element-agent $\elag{i} \in C$ for some $i \in [3\nn]$, then there is a set-agent $\sset{j}$ with $i \in \set{j}$ such that $\sset{j} \in C$, which implies $\axag{1} \in C$, or $\elag{i+1} \in C$; otherwise $\elag{i}$ has negative utility towards $C$.
  By repeating this argument, either a set-agent--and thus $\axag{1}$--must be in $C$, or $\elag{3\nn} \in C$.
  However, if $\elag{3\nn}$ is in $C$, then either a set-agent or $\axag{0}$ must be in C; in both cases, this implies $\axag{1}$ must also be in $C$.

  Thus $\axag{1} \in C$ and hence $\elag{1} \in C$, because $\axag{1}$ has negative utility towards $\agx$ and only positive utility towards $\elag{1}$.
  Agent $\elag{1}$ has -1 utility towards $\agx$ and $\axag{1}$ in $C$, so $\elag{1}$ needs at least one set-agent~$\sset{j}$ with $1 \in \set{j}$ to not deviate.
  However, there can be at most one $\sset{j} \in C$ with $1 \in \set{j}$:
  Suppose, towards a contradiction, that there are $\sset{j}, \sset{j'} \in C, j < j'$ with $1 \in \set{j}$ and $1 \in \set{j'}$. 
  Then $\sset{j}$ has negative utility towards $x$ and $\sset{j'}$.
  But he has only one positive out-arc to $\axag{1}$, so $C$ would not be \IR\ anymore, contradiction.
  
  Hence we have exactly one set-agent $\sset{j} \in C$ with $1 \in \set{j}$. Thus $\elag{2} \in C$ for $C$ to be \IR. 
Similarly to element~$1$, there must be exactly one $\sset{j} \in C$ with $2 \in \set{j}$. 
  Same arguments hold for $\elag{3},\dots, \elag{3\nn}$. 

  To summarize, we need for every element-agent $\elag{i}$ for $i \in [3\nn]$ at least one set-agent $\sset{j}$ with $i \in \set{j}$ which implies we have a set cover of all sets corresponding to the set-agents contained in $C$.
  As we have discussed, the sets must also be disjoint, so the set cover is also an exact-cover.
\end{claimproof}
 The correctness follows immediately from Claims \ref{clm:thm:AD-IR-NA-Forward} and \ref{clm:thm:AD-IR-NA-Backward}.
\end{proof}

The above result is tight since for DAGs or symmetric preferences, we can solve the problem in polynomial time.
In both cases we need to find only one agent $a \in \agentsU$, such that $\agx$ has non-negative utility towards agent $a$ or vice versa:

\begin{restatable}[\appsymb]{proposition}{propDagSymADIRNAAddAg}
  For DAGs, \controlprob{\IR}{\NA}{\AddAg}{\AD} and hence \controlprob{\CS}{\NA}{\AddAg}{\AD} is polynomial-time solvable.
  For symmetric preferences, \controlprob{\IR}{\NA}{\AddAg}{\AD} is polynomial-time solvable.
  \label{prop:SYM-AD-IR-NA-ADD}
\end{restatable}
\appendixproofwithstatement{prop:SYM-AD-IR-NA-ADD}{\propDagSymADIRNAAddAg*}{
  If the preference graph of the \controlprob{\IR}{\NA}{\AddAg}{\AD} instance is a \DAG, then the special agent $\agx \in \agentsU$ will always be alone in an \IR\ partition $\Pi$ if and only if $\agx$ has only negative in- or out-arcs to every other agent.
  
  Considering the case where $\agx$ has only negative in- or out-arcs.
  If $\agx$ is in a coalition $C$ with agent $a$ with $\util{x}(a) < 0$, then $\agx$ will always have negative utility towards $C$ and wishes to deviate from $C$.
  Now suppose that $\agx$ is in a coalition $C$ with agent $a$ with $\util{a}(x) < 0$.
  In order for $a$ not to deviate, there must exist another agent $i \in C$ with $\util{a}(i) \geq \util{a}(x)$.
  However, again either $\util{\agx}(i) < 0$ or $\util{i}(\agx)$ holds.
  In the first case, as we have already seen, $\agx$ wishes to deviate from $C$. 
  In the second case, agent $i$ must itself be compensated by another agent $j \in C$ such that $\util{i}(j) \geq \util{i}(\agx)$.

  Repeating this argument, we must eventually terminate with an agent $i^* \in \agents \setminus \{\agx\}$ who has negative utility towards $C$ or agent $\agx$ has negative utility towards $C$.
  Hence, $C$ cannot be \IR.

  In the other case where there exists an agent $a \in \agentsU \cup \agentsW$ with $\util{\agx}(a) \geq 0$ or $\util{a}(\agx) \geq 0$, then the partition $\Pi$ with $\Pi(\agx) = \Pi(a) = \{\agx,a\}$ and $\Pi(a') = \{a'\}$ for all $a' \in \agentsU \setminus \{x,a\}$ is \IR. 

  If the utilities of the \controlprob{\IR}{\NA}{\AddAg}{\AD} instance is symmetric, then we only need to find an agent $a \in \agents$ with $\util{x}(a) = \util{a}(x) \geq 0$. 

  Note that in both cases, when the instance is a \DAG\ or has symmetric preferences, we would need to add at most one additional agent from $\agentsW$ to find an \IR\ partition with $x$ not being alone.
}

Next, we consider the two more stringent stability concepts \IS\ and \NS. %
It is known that determining the existence of \NS\ and \IS\ partitions is \NP-hard on \AD~\cite{Sung10_Additive}.
Hence, by \cref{obs:gra:eq}, the control problems with goals \NA\ and \PA\ are \NP-hard as well.
We discover that the control problems remain \NP-hard even when the preference graph is a DAG or symmetric.

\cref{thm:DAGSYM-AD-ISNS-NA-ADD} to \cref{thm:SYMMAD-IR-PA} are all shown via reductions from \RETC, using approaches that are similar in structure but distinct in their technical details.

\begin{restatable}[\appsymb]{theorem}{thmDAGSYMADISNSNAADD}
  For each stability concept $\S \in \{\IS,\NS\}$ and each control action $\A \in \{\AddAg, \DelAg\}$, \controlprob{\S}{\NA}{\A}{\AD} and  \controlprob{\S}{\PA}{\A}{\AD} are NP-complete; NP-hardness remains even if the budget is $k=0$ and the preference graph is a DAG.
  \label{thm:DAGSYM-AD-ISNS-NA-ADD}
\end{restatable}
\appendixproofwithstatement{thm:DAGSYM-AD-ISNS-NA-ADD}{\thmDAGSYMADISNSNAADD*}
{
  The \NP-containment follows directly from \cref{obs:complexityupperbounds}.
  Moreover, since $\budget = 0$, both control problems reduce to checking whether there is an \IS\ or \NS\ partition where a given pair of agents is in the same coalition.  
  It remains to show the NP-hardness for the mentioned restriction.
  We reduce again from \RETC.

  Let $I = (\els = [3\nn], \sets)$ be an instance of \RETC\ and we construct an \AD-instance with an acyclic preference graph. 
  The set of agents is $\agentsU = \{x, y\} \cup \{\elag{i} \mid i \in \els\} \cup \{\sset{j} \mid \set{j} \in \sets\} \cup \{\dmag{\ell}{1}, \dmag{\ell}{2}, \dmag{\ell}{3} \mid \ell \in [2\nn]\}$, where
  \begin{compactitem}[--]
    \item $x$ is the special agent, 
    \item each element~$i\in \els$ has one \myemph{element-agent~$\elag{i}$},
    \item each set~$\set{j}\in \sets$ has one \myemph{set-agent~$\sset{j}$},
    \item and auxiliary-agent~$\dmag{\ell}{1}, \dmag{\ell}{2}, \dmag{\ell}{3}$, for $\ell \in [2\nn]$ occurring in triples.
  \end{compactitem}
  We construct the utilities as follows; the unmentioned utilities are~$0$:
  \begin{compactitem}[--]
    \item For each element~$i \in [3\nn]$, set $\util{y}(\elag{i}) = 1$.
    \item For each agent~$a\in \agentsU\setminus \{x,y\}$, set $\util{x}(a) = -1$.
    \item For each set~$\set{j} \in \sets$ and all elements~$i \in \set{j}$,
    set $\util{\sset j}(\elag{i}) = 1$.
    \item For each set~$\set j \in \sets$ and each index~$\ell \in [2\nn]$, set $\util{\sset j}(\dmag{\ell}{1}) = \util{\sset j}(\dmag{\ell}{2}) = \util{\sset j}(\dmag{\ell}{3}) = 1$.
    \item For all~$\ell \in [2\nn]$, set $\util{\dmag{\ell}{1}}(\dmag{\ell}{2}) = \util{\dmag{\ell}{1}}(\dmag{\ell}{3}) = \util{\dmag{\ell}{2}}(\dmag{\ell}{3}) = 1$.
  \end{compactitem}
  The preference graph of the instance is depicted in \cref{fig:thm:DAGSYM-AD-ISNS-NA-ADD}.
  Before we show the correctness, we first observe that the following sequence is a topological order of the preference graph: $y, x, \sset{1}, \sset{2}, \dots, \sset{3\nn}, \allowbreak 
  \elag{1}, \elag{2}, \dots, \elag{3\nn}, \allowbreak 
  \dmag{1}{1}, \dmag{1}{2}, \dmag{1}{3}, \allowbreak 
  \dmag{2}{1}, \dmag{2}{2}, \dmag{2}{3} \dots, \allowbreak
  \dmag{2\nn}{1}, \dmag{2\nn}{2}, \dmag{2\nn}{3}$.

  In the forward direction of the correctness proof, we show that if $I$ has an exact cover, we can construct a \NS\ partition with $|\Pi(x)| \geq 2$, which is by \cref{obs:stability-relation} also \IS.
  In the backward direction, we assume the partition $\Pi$ being \IS--an even weaker assumption than \NS--and show that this implies the existence of an exact cover for~$I$.

  \begin{claim} \label{clm:thm:DAG-AD-ISNS-NA-ADD-Forward}
   If $\excov$ is an exact cover for~$I$, then the following partition~$\Pi$ is \NS:
   $\Pi(\agx) = \{\agx,\agy\}$, $\Pi(\sset{j}) = \{\sset{j}, \elag{j_1}, \elag{j_2}, \elag{j_3}\}$ for $\set{j} = \{j_1,j_2,j_3\} \in \excov$ and $\Pi(\sset{j}) = \{\sset{j}, \dmag{\ell}{1}, \dmag{\ell}{2}, \dmag{\ell}{3}\}$ for $\set{j} \notin \excov$ for some $\ell \in  [2\nn]$.
 \end{claim} 
 \begin{claimproof}{clm:thm:DAG-AD-ISNS-NA-ADD-Forward}
  To show the statement, we need to show that every agent $\ag{i} \in \agentsU$ has no incentive to leave his current coalition, i.e., there exist no other coalition $C \in \Pi$ with $C \cup \{i\} \succ_i \Pi(i)$.
  
  Let us first consider agent~$x$. He has negative utility towards every agent except agent~$y$ and therefore wants to stay in his current coalition.

  Agent~$y$ has utility 1 towards his coalition $\{x,y\}$ and also towards all the other coalitions $C \in \Pi$, because they all contain only one set-agent $\sset{j}$. 
  Hence, agent~$y$ also wants to stay in his coalition with agent~$x$.

  For the set-agents $\sset{j}, \set j \in \sets$, they are either in a coalition with their three element-agents~$\elag{i}$ or with three auxiliary-agents~$\{\dmag{\ell}{1},\dmag{\ell}{2},\dmag{\ell}{3}\}$ for some $\ell \in [2\nn]$.
  In both cases, the set-agents have utility 3 towards their coalitions and therefore, leaving a coalition would not strictly improve the utility.

  The element-agents~$\elag{i}$, for $i \in [3\nn]$, do not have any preferences, so they also do not want to deviate.

  Auxiliary-agents~$\dmag{\ell}{1}$ for $\ell \in [2\nn]$ have utility 2 towards their coalitions and $\dmag{\ell}{2}$ for $\ell \in [2\nn]$ have utility 1 which is in both cases the maximum utility they can get. 
  The agents~$\dmag{\ell}{3}$ do not have any preferences, so they also do not want to deviate.
 \end{claimproof}

 \begin{claim}\label{clm:thm:DAG-AD-ISNS-NA-ADD-Backward}
   If~$\Pi$ is an \IS\ partition with $|\Pi(\agx)| \geq 2$, then $\agy\in \Pi(x)$ and all set-agents which are in a coalition with their three element-agents form an exact cover of $I$.
 \end{claim} 
 \begin{claimproof}{clm:thm:DAG-AD-ISNS-NA-ADD-Backward}
  We aim to show that all set-agents which are in a coalition with their three element-agents form an exact cover of $I$.

  By assumption, let $C$ be the coalition containing $x$. 
  The only coalition $C$ with $x$ not alone is $\{x,y\}$, because agent~$x$ has negative utility towards every other agent and would then prefer being alone.

  Hence, agent $y$ has utility 1 towards coalition $C$. As we know that $\Pi$ is \IS, no two set-agents  are in the same coalition.
  Otherwise, agent $y$ would prefer joining that coalition.

  We also know that for all~$\ell \in [2\nn]$, all three auxiliary-agents--$\dmag{\ell}{1}$, $\dmag{\ell}{2}$ and $\dmag{\ell}{3}$--must be in the same coalition.
  If $\dmag{\ell}{1}$ is not in the same coalition with $\{\dmag{\ell}{2}, \dmag{\ell}{3}\}$, $\dmag{\ell}{1}$ would have utility 0 towards his coalition and would therefore want to deviate to a coalition containing $\dmag{\ell}{2}$ or $\dmag{\ell}{3}$.
  If $\dmag{\ell}{2}$ is not in the same coalition with $\{\dmag{\ell}{1},\dmag{v}{3}\}$, he would want to deviate to the coalition containing $\dmag{\ell}{3}$.
  The same happens if $\dmag{\ell}{3}$ was separated: Then also $\dmag{\ell}{2}$ would want to join the coalition containing $\dmag{\ell}{3}$.

  Two auxiliary-agent triples $\{\dmag{\ell}{1}, \dmag{\ell}{2},\dmag{\ell}{3}\}$ and $\{\dmag{\ell'}{1}, \dmag{\ell'}{2},\dmag{\ell'}{3}\}$ cannot be in the same coalition. 
  Suppose, towards a contradiction, there exists a coalition $C'$ containing at least two different auxiliary-agents triples, i.e. $C' \subseteq \{\dmag{\ell}{1}, \dmag{\ell}{2}, \dmag{\ell}{3}, \dmag{\ell'}{1}, \dmag{\ell'}{2}, \dmag{\ell'}{3}\}$.
  Any set-agent $\sset{j}$ has utility 6 towards $C'$, therefore all set-agents $\sset{j}$ must have utility at least 6 towards their coalition in order for $\Pi$ to remain \IS.
  Hence, the set-agents are together either with (1) their three element-agents and an auxiliary-agent triple or (2) two auxiliary-agent triples.
  However, at most $\nn$ set-agents can be together with his three element-agents and an auxiliary-agent triple. 
  For the remaining $2\nn$ set-agents, only $\nn$ auxiliary-agents triples are available. 
  Therefore, there must exist one set-agent $\sset{j^*}$ without any element-agent nor auxiliary-agents because all set-agents are in different coalitions. 
  The set-agent $\sset{j^*}$ has utility 0 towards his coalition, wanting to deviate to a coalition with auxiliary-agents, a contradiction to $\Pi$ being \IS.

  It remains to show that every set-agents~$\sset{j}, \set j \in \sets$ are either together with his three element-agents $\elag{j_1},\elag{j_2},\elag{j_3}$ contained in $\set{j} = \{j_1,j_2,j_3\}$ or with exactly three auxiliary-agents $\{\dmag{\ell}{1}, \dmag{\ell}{2}, \dmag{\ell}{3}\}$ for some $\ell \in [2\nn]$.
  We have shown that an auxiliary-agent triple $\{\dmag{\ell}{1}, \dmag{\ell}{2} ,\dmag{\ell}{3}\}$ is together in a coalition and no two different auxiliary-agent triples $\{\dmag{\ell}{1}, \dmag{\ell}{2}, \dmag{\ell}{3}\}$ and $\{\dmag{\ell'}{1}, \dmag{\ell'}{2}, \dmag{\ell'}{3}\}$ are in the same coalition.
  Hence, every set-agent $\sset{j}$ must have utility at least three towards his coalition.
  But at most $2\nn$ set-agents can be together with an auxiliary-agent triple.
  Thus, $\nn$ set-agents must be together with their corresponding three element-agents, which implies that the sets corresponding to these set-agents form an exact-cover. 
 \end{claimproof}

  \begin{figure}[t!]
  \centering  \begin{tikzpicture}[>=stealth',shorten <= 1pt, shorten >= 1pt]
  \def \xx {1.3}
  \def \xy {1}
  \foreach \x / \n in {1/1, 2/2, 4/j, 5/3n} {
    \node[unode] at (\x*\xx, -1.5*\xy) (s\n) {};
  }

  \node[unode] at (3*\xx, 0*\xy) (x) {};
  \node[unode] at (3*\xx, 1*\xy) (y) {};

  \node[] (h1) {};
  \node[] (h2) {};
  \node[] (h3) {};
  \node[] (h4) {};
  \coordinate (h1) at ($(x) + (-2.0, -0.75)$);
  \coordinate (h2) at ($(x) + (-0.7, -0.75)$);
  \coordinate (h3) at ($(x) + (0.7, -0.75)$);
  \coordinate (h4) at ($(x) + (2.0, -0.75)$);
  
  \foreach \x / \n in {1/1,2/2, 4/i, 5/3n} {
    \node[unode] at (\x*\xx, -2.5*\xy) (u\n) {};
  }

  \foreach \x / \n in {1.2/1, 4.8/3} {
    \node[unode] at (\x*\xx, -3.5*\xy) (d\n1) {};
    \node[unode] at (\x*\xx, -4.25*\xy) (d\n2) {};
    \node[unode] at (\x*\xx, -5*\xy) (d\n3) {};
  }

  \foreach \i / \p / \a / \r / \n in {%
    x/above left/2pt/0pt/{x}, y/above left/2pt/0pt/{y},
    u1/below left/0pt/-5pt/{u_1},   u2/below left/0pt/-5pt/{u_2},   ui/below left/0pt/-5pt/{u_i},  u3n/below left/0pt/-5pt/{u_{3\nn}}, 
    s1/below left/0pt/-5pt/{s_1},   s2/below left/0pt/-5pt/{s_2},   sj/below left/0pt/-5pt/{s_j},  s3n/below left/0pt/-5pt/{s_{3\nn}},
    d11/above right/-4pt/7pt/{d_1^1}, d12/above right/1pt/7pt/{d_1^2}, d13/above right/1pt/7pt/{d_1^3}, 
    d31/above left/-4pt/7pt/{d_{2\nn}^1}, d32/above left/1pt/7pt/{d_{2\nn}^2}, d33/above left/1pt/7pt/{d_{2\nn}^3}, 
  } {
     \node[\p = \a and \r of \i,fill=white,inner sep=0.1pt] {$\n$};
   }

  \path (s2) -- node[midway, inner sep=1pt] (sdots) {$\dots$} (sj);

  \path (u2) -- node[midway, inner sep=1pt] (udots) {$\dots$} (ui);

  \path (d12) -- node[midway, inner sep=1pt] (ddots) {$\dots$} (d32);

  \begin{pgfonlayer}{bg}
    \foreach \s / \t / \aa / \w / \typ in {
    x/h1/0/{-1}/ec, x/h2/0/{-1}/ec, x/h3/0/{-1}/ec, x/h4/0/{-1}/ec,
    y/x/0/{1}/fc, y/s1/20/{1}/fc, y/s2/0/{1}/fc, y/sj/0/{1}/fc, y/s3n/-20/{1}/fc,
    s1/u1/0/{1}/fc, s1/u2/0/{1}/fc, s1/ui/0/{1}/fc, sj/u2/0/{1}/fc, sj/ui/0/{1}/fc, sj/u3n/0/{1}/fc,
    s1/d11/70/{1}/fc, s1/d12/65/{1}/fc,s1/d13/65/{1}/fc,
    s3n/d31/-70/{1}/fc, s3n/d32/-65/{1}/fc,s3n/d33/-65/{1}/fc,
    d11/d12/0/{1}/fc, d11/d13/-30/{1}/fc, d12/d13/0/{1}/fc,
    d31/d32/0/{1}/fc, d31/d33/30/{1}/fc, d32/d33/0/{1}/fc%
    }
    {
      \draw[->, \typ] (\s) edge[bend right = \aa] node[fill=white, inner sep=1pt, midway] {\footnotesize $\w$} (\t);
    }
  \end{pgfonlayer}
\end{tikzpicture}
\caption{Illustration of the relevant part of the reduction for \cref{thm:DAGSYM-AD-ISNS-NA-ADD}, assuming that $\set{1} = \{1,2,i\}$ and $\set{j} = \{2,i,3\nn\}$.}\label{fig:thm:DAGSYM-AD-ISNS-NA-ADD}
\end{figure}
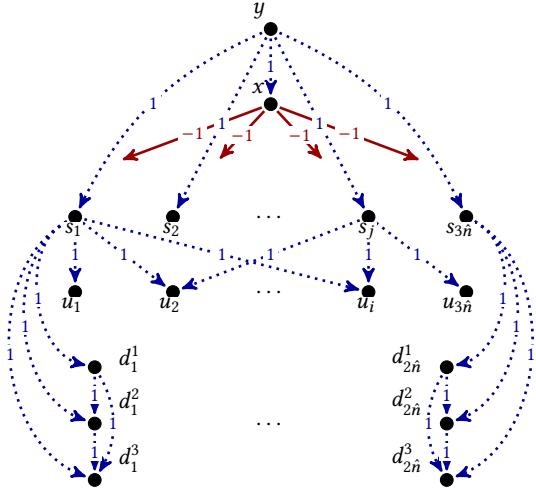

 The correctness follows immediately from \cref{clm:thm:DAG-AD-ISNS-NA-ADD-Forward,clm:thm:DAG-AD-ISNS-NA-ADD-Backward}. 
 Finally, we note that the same construction also works for \PA\ when we enforce that $x$ and $y$ must be in the same coalition.
}

Hardness for \IS\ and \NS\ extend to symmetric preferences.

\begin{restatable}{theorem}{thmSYMADISNSNAADD}
  For each stability concept $\S \in \{\IS,\NS\}$ and each control action $\A \in \{\AddAg, \DelAg\}$, \controlprob{\S}{\NA}{\A}{\AD} is NP-complete; NP-hardness remains even if the budget is $k=0$ and the preference graph is symmetric.
  \label{thm:SYM-AD-ISNS-NA-ADD}
\end{restatable}

\begin{proof}
  The \NP-containment follows directly from \cref{obs:complexityupperbounds}.
  As before, we will reduce from \RETC\ to show the NP-hardness for the mentioned control problem.

  Let $I = (\els = [3\nn], \sets)$ be an instance of \RETC. 
  Let us construct a hedonic game instance with additive and symmetric preferences as follows. The set of agent is $\agentsU = \{\agx, y\} \cup \{\elag{i} \mid i \in \els\} \cup \{\sset{j} \mid \set{j} \in \sets\} \cup \{\dmag{\ell}{} \mid \ell \in [2\nn]\}$, where

  \begin{compactitem}[--]
    \item $\agx$ is the special agent and $y$ an auxiliary-agent,
    \item each element~$i\in \els$ has one \myemph{element-agent~$\elag{i}$},
    \item each set~$\set{j}\in \sets$ has one \myemph{set-agent~$\sset{j}$},
    \item and auxiliary-agents~$\dmag{\ell}{}$ for $\ell \in [2\nn]$.
  \end{compactitem}

  We construct the symmetric utilities as follows; the unmentioned utilities are 0:
  \begin{compactitem}[--]
    \item Set $\util{\agx}(y) = \util{y}(\agx) = 1$.
    \item For each agent~$a\in \agentsU\setminus \{x,y\}$, set $\util{\agx}(a) = \util{a}(\agx) = -1$.
    \item For each~$\set{j} \in \sets$, set~$\util{\sset{j}}(y) = \util{y}(\sset{j}) = 1$.
    \item For each set~$\set{j} \in \sets$ and all elements~$i \in \set{j}$,
    set $\util{\sset j}(\elag{i}) = \util{\elag{i}}(\sset j) =~1$.
    \item For each set~$\set{j} \in \sets$ and every auxiliary-agent $\dmag{\ell}{}$ for $\ell \in [2\nn]$, set $\util{\sset{j}}(\dmag{\ell}{}) = \util{\dmag{\ell}{}}(\sset{j}) = 3$.
  \end{compactitem}

  \appendixfigure{fig:thm:SYM-AD-ISNS-NA-ADD}{}{Figure of \cref{thm:SYM-AD-ISNS-NA-ADD}}
{
  Sketch of the preference graph for \cref{thm:SYM-AD-ISNS-NA-ADD}, see \cref{fig:thm:SYM-AD-ISNS-NA-ADD}.

\begin{figure}[t!]
  \centering  \begin{tikzpicture}[>=stealth',shorten <= 1pt, shorten >= 1pt]
  \def \xx {1.3}
  \def \xy {1}
  \foreach \x / \n in {1/1, 2/2, 4/j, 5/3n} {
    \node[unode] at (\x*\xx, -1.5*\xy) (s\n) {};
  }

  \node[unode] at (3*\xx, 0*\xy) (y) {};
  \node[unode] at (3*\xx, 1*\xy) (x) {};

  \node[] (h1) {};
  \node[] (h2) {};
  \node[] (h3) {};
  \node[] (h4) {};
  \coordinate (h1) at ($(x) + (-2.0, -0.75)$);
  \coordinate (h2) at ($(x) + (-0.7, -0.75)$);
  \coordinate (h3) at ($(x) + (0.7, -0.75)$);
  \coordinate (h4) at ($(x) + (2.0, -0.75)$);
  
  \foreach \x / \n in {1/1,2/2, 4/i, 5/3n} {
    \node[unode] at (\x*\xx, -2.5*\xy) (u\n) {};
  }

  \foreach \x / \n in {1.2/1, 4.8/3} {
    \node[unode] at (\x*\xx, -3.5*\xy) (d\n) {};
  }

  \foreach \i / \p / \a / \r / \n in {%
    x/above left/2pt/0pt/{x}, y/above left/2pt/0pt/{y},
    u1/below left/0pt/-5pt/{u_1},   u2/below left/0pt/-5pt/{u_2},   ui/below left/0pt/-5pt/{u_i},  u3n/below left/0pt/-5pt/{u_{3\nn}}, 
    s1/below left/0pt/-5pt/{s_1},   s2/below left/0pt/-5pt/{s_2},   sj/below left/0pt/-5pt/{s_j},  s3n/below left/0pt/-5pt/{s_{3\nn}},
    d1/below right/2pt/0pt/{d_1}, d3/below left/2pt/0pt/{d_{2\nn}}%
  } {
     \node[\p = \a and \r of \i,fill=white,inner sep=0.1pt] {$\n$};
   }

  \path (s2) -- node[midway, inner sep=1pt] (sdots) {$\dots$} (sj);

  \path (u2) -- node[midway, inner sep=1pt] (udots) {$\dots$} (ui);

  \path (d1) -- node[midway, inner sep=1pt] (ddots) {$\dots$} (d3);

  \begin{pgfonlayer}{bg}
    \foreach \s / \t / \aa / \w / \typ in {
    x/h1/0/{-1}/ec, x/h2/0/{-1}/ec, x/h3/0/{-1}/ec, x/h4/0/{-1}/ec,
    y/x/0/{1}/fc, y/s1/20/{1}/fc, y/s2/0/{1}/fc, y/sj/0/{1}/fc, y/s3n/-20/{1}/fc,
    s1/d1/90/{3}/fc, s3n/d3/-90/{3}/fc,
    s2/d1/25/{3}/fc,
    sj/d3/25/{3}/fc,
    s1/u1/0/{1}/fc, s1/u2/0/{1}/fc, s1/ui/0/{1}/fc, sj/u2/0/{1}/fc, sj/ui/0/{1}/fc, sj/u3n/0/{1}/fc%
    }
    {
      \draw[-, \typ] (\s) edge[bend right = \aa] node[fill=white, inner sep=1pt, midway] {\footnotesize $\w$} (\t);
    }
  \end{pgfonlayer}

\end{tikzpicture}
\caption{Illustration of the relevant part of the reduction for \cref{thm:SYM-AD-ISNS-NA-ADD}, assuming that $\set{1} = \{1,2,i\}$ and $\set{j} = \{2,i,3\nn\}$. Note that $x$ has negative edges between all other agents (except $y$) and $\{\sset{1}, \dots, \sset{3\nn}\}, \{\dmag{1}{}, \dots, \dmag{2\nn}{}\}$ form a complete bipartite graph.}\label{fig:thm:SYM-AD-ISNS-NA-ADD}
\end{figure}
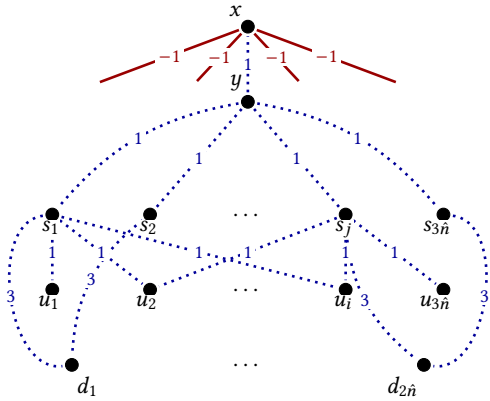
}

  \ifcr
\else
    A sketch of the preference graph can be found in the appendix, see \cref{fig:thm:SYM-AD-ISNS-NA-ADD}.
  \fi
  To complete the construction, let agent $\agx$ be the special agent not to be alone.

  In the forward direction of the correctness proof, we show that if~$I$ has an exact cover, we can construct a \NS\ partition with $|\Pi(x)| \geq 2$, which is, by \cref{obs:stability-relation}, also \IS.
  In the backward direction, we assume the partition $\Pi$ being \IS--an even weaker assumption than \NS--and show that this implies the existence of an exact cover for~$I$. 
  
  \begin{restatable}[\appsymb]{claim}{claimSYMADISNSNAADDForward}
   \label{clm:thm:SYM-AD-ISNS-NA-ADD-Forward}
   If $\excov$ is an exact cover for~$I$, then the following partition~$\Pi$
   with $\Pi(\agx) = \{\agx,\agy\}$, $\Pi(\sset{j}) = \{\sset{j}, \elag{j_1}, \elag{j_2}, \elag{j_3}\}$ for $\set{j} = \{j_1,j_2,j_3\} \in \excov$ and $\Pi(\sset{j}) = \{\sset{j}, \dmag{\ell}{}\}$ for $\set{j} \notin \excov$ and $\ell \in  [2\nn]$ is \NS.
 \end{restatable} 
 \appendixproofwithstatement{clm:thm:SYM-AD-ISNS-NA-ADD-Forward}{\claimSYMADISNSNAADDForward*}{
  To show the statement, we need to show that every agent $i \in \agentsU$ has no incentive to leave his current coalition, i.e., there exist no other coalition $C \in \Pi$ with $C \cup \{i\} \succ_i \Pi(i)$.

  Let us first consider special agent $\agx$. 
  It has utility 1 towards his coalition $\{\agx,y\}$ but to all others he has negative utility. 
  So agent~$\agx$ wants to stay in his coalition.

  Agent~$y$ has also utility 1 towards his coalition, so do all other coalitions containing exactly one set-agent $\sset{j}$.
  Therefore, agent~$y$ would not improve by joining another existing coalition.

  All set-agents~$\sset{j}$, for $\set{j} \in \sets$, have utility 3 towards their coalitions.
  Either they are together with three element-agents~$\elag{i}$ with $i \in \set{j}$ or with an auxiliary-agent $\dmag{\ell}{}$, for some $\ell \in [2\nn]$.
  Joining another coalition would not increase the utility and thus the set-agents $\sset{j}$ have no incentive to leave their current coalitions.

  The element-agents $\elag{i}$ have utility 1 towards their coalitions with a set-agent. 
  They also have at most utility 1 towards every other set-agent, but they all are in different coalitions, hence no element-agent would improve by switching to another coalition.

  The same argument holds for auxiliary-agents $\dmag{\ell}{}$, as they have utility 3 towards their coalitions. 
  No two set-agents are in the same coalition.
  Therefore, the auxiliary-agents also do not have an incentive to leave their current coalitions.
 }

 \begin{restatable}[\appsymb]{claim}{claimSYMADISNSNAADDBackward}
   \label{clm:thm:SYM-AD-ISNS-NA-ADD-Backward}
   If $\Pi$ is an \IS\ partition with $|\Pi(\agx)| \geq 2$, then $\agy\in \Pi(x)$ and all set-agents which are in a coalition with their three element-agents form an exact cover of $I$.
 \end{restatable} 
 \appendixproofwithstatement{clm:thm:SYM-AD-ISNS-NA-ADD-Backward}{\claimSYMADISNSNAADDBackward*}{
  We aim to show that all set-agents which are in a coalition with their three element-agents form an exact cover of $I$.

  By assumption, let $C: = \Pi(x)$  with $|C| \geq 2$ be a coalition containing $x$ in an \IS\ partition.
  All coalitions with $x$ not alone must be at least with agent $y$, because agent~$x$ has negative utility towards every other agent.
  There can be at most one other agent besides $y$ in $C$, otherwise $x$ would deviate to his singleton coalition.
  
  A set-agent with $\sset{j} \in C$ is not possible because $\sset{j}$ would have utility 0 towards $C$. 
  But $\sset{j}$ has only negative utility towards agent $x$, so he would prefer any different coalition containing an auxiliary-agent~$\dmag{\ell}{}$ or an element-agent $\elag{i}$ with $i \in \set{j}$.

  Element-agents $\elag{i}$ and auxiliary-agents $\dmag{\ell}{}$ already have negative utility for the coalition $\{x,y\} \subseteq C$, so they would prefer to be alone rather than join it.
  Thus, the only coalition $C$ containing $x$ is $\{\agx, y\}$.
  
  Every agent $a \in \agentsU \setminus \{x,y\}$ has negative utility only toward agent $x$.
  Thus, for every agent $j \in \agentsU$ and every $a \in \agentsU \setminus \{x\}$ holds that agent $j$ is indifferent between the coalitions $\Pi(j) \cup \{a\}$ and $\Pi(j)$.
  Therefore, we can argue in the following that no agent has an incentive to deviate to an existing coalition in $\Pi$.

  We know that $y$ does not want to deviate to another coalition, so no two set-agents $\sset{j}$ can be in the same coalition; otherwise $y$ would prefer to join them, contradicting $\Pi$ being \IS.

  We now want to show that no two auxiliary-agents can be in the same coalition in $\Pi$.
  Suppose, towards a contradiction, that for two auxiliary-agents $\dmag{\ell}{}, \dmag{\ell'}{}$, it holds that $C' := \Pi(\dmag{\ell}{}) = \Pi(\dmag{\ell'}{})$.
  So all set-agents $\sset{j}$ have utility 6 towards $C'$.
  Since $\Pi$ is \IS\, all set-agents must also have utility 6 towards their coalition $\Pi(\sset{j})$.
  We can say that at most $\nn$ coalitions with at least two auxiliary-agents can exist.
  Thus, at most $\nn$ set-agents can be within those coalitions and at least the remaining $2\nn$ set-agents must be together with their element-agents.
  But every set-agent has only three element-agents with utility greater than 0, so they can have at most utility 3 towards their coalition.
  They would then want to deviate to, e.g., $C'$ and therefore $\Pi$ would not be \IS\ anymore, contradiction.

  No two set-agents $\sset{j}$ can be in the same coalitions and no two auxiliary-agents $\dmag{\ell}{}$ can be in the same coalition.
  In order for $\Pi$ to be \IS, every $\dmag{\ell}{}$ must be with a set-agent $\sset{j}$ in a coalition; otherwise $\dmag{\ell}{}$ would deviate to a coalition containing a set-agent.
  Therefore, $2\nn$ set-agents are with exactly one auxiliary-agent in the same coalition and have utility 3 towards their coalitions.

  The remaining $\nn$ set-agents $\sset{j}$ without any auxiliary-agent must also have utility 3 towards their coalitions, this can only be achieved being together with all their element-agents $\elag{i}$ with $i \in \set{j}$.
  This means, the sets corresponding to these $\nn$ set-agents must form a set cover, and since every set contains exactly three elements, these $\nn$ sets cover $3\nn$ elements, which implies these sets form an exact cover.
 }

 The correctness follows immediately from Claims \ref{clm:thm:SYM-AD-ISNS-NA-ADD-Forward} and \ref{clm:thm:SYM-AD-ISNS-NA-ADD-Backward}.
\end{proof}

In contrast to \cref{prop:SYM-AD-IR-NA-ADD}, \PA\ remains \NP-hard on DAGs or symmetric preferences, even when no agent can be added or deleted.

\begin{restatable}[\appsymb]{theorem}{thmADIRPANPhard}
  For each stability concept~$\S \in \{\IR, \CS\}$ and each control action~$\A\in \{\AddAg, \DelAg\}$, \controlprob{\S}{\PA}{\A}{\AD} is \NP-complete; NP-hardness remains even if the budget is $\budget=0$ and the preference graph is acyclic with maximum vertex degree nine.
  \label{thm:AD-IR-PA}
\end{restatable}
\appendixproofwithstatement{thm:AD-IR-PA}{\thmADIRPANPhard*}
{
  By \cref{obs:dag:eq}, it suffices to consider the stability concept \IR.
  To show NP-hardness, we again reduce from \RETC. 
  Let $I=(\els=[3\nn], \sets)$ be an instance of \RETC.
  Let us construct a hedonic game instance with additive preferences as follows.
  The set of agents is $\agentsU=\{\axag{0}\}\cup \{\axag{i}, \elag{i}\mid i\in \els\}\cup \{\sset{j}\mid \set{j} \in \sets\}$, where 
  \begin{compactitem}[--]
    \item $\axag{0}$ is a special agent,
    \item each element~$i\in \els$ has one \myemph{element-agent~$\elag{i}$} and an \myemph{auxiliary-agent~$\axag{i}$}, 
    \item and each set~$\set{j}\in \sets$ has one \myemph{set-agent~$\sset{j}$}.
  \end{compactitem}
  We construct the utilities as follows; the unmentioned utilities are~0:
  \begin{compactitem}[--]
    \item Set $\util{\axag{0}}(\axag{1}) = -2, \util{\axag{0}}(\elag 1) = 1$, and $\util{\axag{0}}(\axag 2) = 1$.
    \item For all~$i \in [3\nn - 2]$, set $\util{\axag i}(\axag {i + 1}) = -2$, $\util{\axag i}(\elag {i + 1}) = \util{\axag i}(\axag {i + 2}) = 1$.
    \item Set $\util{\axag {3\nn - 1}}(\axag {3\nn}) = -1, \util{\axag {3\nn - 1}}(\elag {3\nn}) = 1$.
    \item For all~$i \in [3\nn]$, set $\util{\elag i}(\axag i) = -1$.
    \item For all~$\set j \in \sets$ and all~$i \in \set j$, set $\util{\elag i}(\sset j) = 1$.
    \item For all~$\set j, \set {\ell} \in \sets$ with $j < \ell$ and $\set j \cap \set \ell \neq \emptyset$,
    set $\util{\sset j}(\sset \ell) = -1$.
  \end{compactitem}
  The preference graph of the instance is depicted in \cref{fig:thm:AD-IR-PA}.
  To complete the construction, we define~$\{x,y\}=\{\axag{0}, \axag{1}\}$ as the pair of agents which we aim to be in the same coalition in an \IR\ partition.
  Before we show the correctness, we first observe that the following sequence is a topological ordering of the preference graph: $\axag{0}, \elag 1$, $\axag 1$, $\elag 2, \axag 2, \dots, \elag {3\nn}, \axag{3\nn},$ $ \sset 1, \dots, \sset {3\nn}$. 
  We observe the following regarding the maximum degree in the preference graph:

  \begin{observation}\label{obs:AD-IR-PA-Deg}
    Every vertex in the preference graph has degree at most nine.
  \end{observation}
  We analyze the degree of every vertex one by one.
  \begin{compactitem}[--]
    \item The agent $\axag{0}$ has two out-arcs and no in-arcs.
    \item Every auxiliary-agent~$\axag{i}$ has at most three in-arcs and three out-arcs.
    \item Every element-agent~$\elag{i}$ has at most one in-arc and four out-arcs.
    \item Every set-agent~$\sset{j}$ has three in-arcs from the element-agents corresponding to the elements contained in~$\set{j}$. 
    Moreover, since every element appears in exactly three sets, there are at most six sets that intersect any given set.
    Thus, a set-agent has maximum degree nine.
  \end{compactitem}

  \begin{claim}\label{clm:thm:pa_ir_add_add_maxdeg_dag_maxutil_kzero_1}
    If $\excov$ is an exact cover for~$I$, then the following partition~$\Pi$
   with $\Pi(\axag{0}) = \{\axag{0},\axag{1}, \dots, \axag{3\nn}, 
   \elag{1},\dots, \elag{3\nn}\}\cup \{\sset j \mid \set j \in \excov\}$, and 
   $\Pi(\sset{j})=\{\sset{j}\}$ for all $\set{j}\notin \excov$ is \IR.
  \end{claim}
  \begin{claimproof}{clm:thm:pa_ir_add_add_maxdeg_dag_maxutil_kzero_1}
    To show the statement, we show that every agent has zero utility towards his coalition under~$\Pi$.

   Let us first consider agent~$\axag{0}$ and all the auxiliary-agents.
   We observe that $\util{\axag{0}}(\Pi(\axag{0})) = \util{\axag{0}}(\{\axag 1, \axag 2, \elag 1\}) = -2 + 1 + 1 = 0$, as required.
   
   For every $i \in [3\nn - 2]$, it holds that $\util{\axag i}(\Pi(\axag i)) = \util{\axag i}(\{\axag {i + 1}, \axag {i + 2},$ $\elag{i + 1}\}) = -2 + 1 + 1 = 0$, as required.
   
   We have that $\util{\axag{ 3\nn - 1}}(\Pi(\axag{3\nn - 1})) = \util{\axag{3\nn-1}}(\{\axag{3\nn}, \elag{3\nn}\}) = -1 + 1 = 0$, as required.  
   Since $\axag{3\nn}$ has utility zero towards every other agent, his utility is zero as well.

   Next we show that every element-agent $\elag i$, for $i \in [3\nn]$, has zero utility towards his coalition as well.
   Since $\excov$ is an exact cover, there must be a set $\set j \in \excov$ such that $i \in \set j$. 
   Thus $\utilP{\elag i} = \util{\elag i}(\{\axag i, \sset j\}) = - 1 + 1 = 0$, as required.

   Finally, we show that every set-agent $\sset j$, with $\set j \in \sets$, has zero utility towards his coalition.
   This is clearly the case for each set-agent~$\sset{j}$ with $\set j \notin \excov$ since he is in his own singleton coalition.
   Since $\excov$ is an exact cover, no two sets in~$\excov$ intersect (as otherwise $|\excov| > \nn$).
   Hence, $\utilP{\sset j} = 0$ also holds for every set-agent~$\sset{j}$ with $\set{j}\in \excov$.
   This concludes the proof.
  \end{claimproof}

  \begin{claim}\label{clm:thm:pa_ir_add_add_maxdeg_dag_maxutil_kzero_2}
    If $\Pi$ is an \IR\ partition with $\Pi(\axag{0})=\Pi(\axag{1})$, then the sets corresponding to the set-agents contained in $\Pi(\axag{0})$ form an exact cover of~$I$.
  \end{claim}
  \begin{claimproof}{clm:thm:pa_ir_add_add_maxdeg_dag_maxutil_kzero_2}
    By assumption, let $C$ be the coalition containing~$\axag{0}$ and $\axag{1}$.
   We aim to show that the sets corresponding to the set-agents contained in $C$ form an exact cover of~$I$.
   
   We first show that all auxiliary- and element-agents are contained in~$C$.
   We achieve this by showing that for each~$i\in \{0\}\cup[3\nn-2]$, if $\axag{i}, \axag{i+1}\in C$,
   then $\axag{i+2}, \elag{i+1} \in C$ as well.
   Notice that for each~$i\in \{0\}\cup[3\nn-2]$, agent~$\axag{i}$ has $-2$ utility towards $\axag{i+1}$, $1$ utility towards $\axag{i+2}$ and $\elag{i+1}$, and zero utility towards all remaining agents. 
   Since $\axag{i}$ and $\axag{i+1}$ are both in~$C$,
   in order for $C$ to be \IR, we have both~$\axag{i+2}$ and $\elag{i+1}$ in~$C$.

   Since every element-agent~$\elag{i}$ has $-1$ utility towards his auxiliary-agent~$\axag{i}$,
   by previous paragraph and by the utilities of~$\elag{i}$, coalition~$C$ has to have at least one set-agent~$\sset{j}$ such that~$i\in \set{j}$ to make sure overall utility of~$\elag{i}$ become non-negative.
   This implies that the sets corresponding to the set-agents in~$C$ form a set cover.
   Let $\excov$ be the set $\excov = \{\set{j}\mid \sset{j}\in C\}$.
   To show that $\excov$ is indeed an exact cover it suffices to show that no two sets in~$\excov$ intersect.   
   Suppose, towards a contradiction, that $\excov$ contains $\set{j}$ and $\set{j'}$ with $\set{j}\cap \set{j'}\neq \emptyset$.
   Without loss of generality, let $j< j'$.
   Then, $\sset{j}$ has $-1$ utility towards~$\sset{j'}$.
   Since no set-agent has positive utility towards any other agent, coalition~$C$ yields negative utility for~$\sset{j}$ and is not \IR\ for him, a contradiction.
  \end{claimproof}
   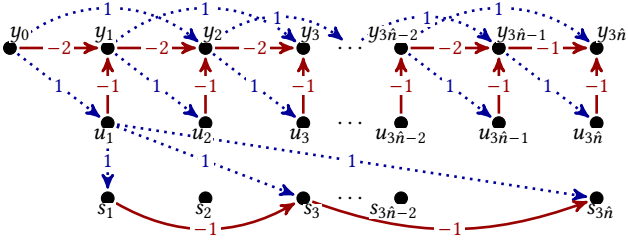
\begin{figure}[t!]
  \begin{tikzpicture}[>=stealth',shorten <= 1pt, shorten >= 1pt]
  \def \xx {1.3}
  \def \xy {1}
  \foreach \x / \n in {1/1, 2/2, 3/3, 4/3n2, 6/3n} {
    \node[unode] at (\x*\xx, -2*\xy) (s\n) {};
  }

  \node[unode] at (0*\xx, 0*\xy) (y0) {};
  
  \foreach \x / \n in {1/1,2/2,3/3,4/3n2, 5/3n1, 6/3n} {
    \node[unode] at (\x*\xx, 0*\xy) (y\n) {};
    \node[unode] at (\x*\xx, -1*\xy) (u\n) {};
  }

  \foreach \i / \p / \a / \r / \n in {%
    y0/above left/0/-10pt/{y_0}, y1/above left/0pt/-5pt/{y_1}, y2/above left/0pt/-9pt/{y_2}, y3/above left/0pt/-9pt/{y_3}, y3n2/above left/0pt/-9pt/{y_{3\nn-2}}, y3n1/above right/0pt/-3pt/{y_{3\nn-1}}, y3n/above right/0pt/-2pt/{y_{3\nn}},
    u1/below left/0pt/-5pt/{u_1},   u2/below left/0pt/-5pt/{u_2},   u3/below left/0pt/-5pt/{u_3},   u3n2/below left/0pt/-12pt/{u_{3\nn-2}},  u3n1/below left/0pt/-14pt/{u_{3\nn-1}},  u3n/below left/0pt/-5pt/{u_{3\nn}}, 
    s1/below left/0pt/-5pt/{s_1},s2/below left/0pt/-5pt/{s_2},s3/below left/0pt/-9pt/{s_3},s3n2/below right/0pt/-14pt/{s_{3\nn-2}},s3n/below left/0pt/-9pt/{s_{3\nn}}%
  } {
     \node[\p = \a and \r of \i,fill=white,inner sep=0.1pt] {$\n$};
   }
  
  \foreach \y in {y, u, s} {
    \path (\y 3) -- node[midway,inner sep=1.7pt] (\y m) {$\dots$} (\y 3n2);
  }

  \begin{pgfonlayer}{bg}
    \foreach \s / \t / \aa / \w / \typ in {
    y0/y1/0/{-2}/ec, y1/y2/0/{-2}/ec,y2/y3/0/{-2}/ec,y3n2/y3n1/0/{-2}/ec,y3n1/y3n/0/{-1}/ec,
    y0/y2/{-40}/{1}/fc,  y1/y3/{-40}/{1}/fc,  y3n2/y3n/{-40}/{1}/fc,
    y0/u1/0/{1}/fc, y1/u2/0/{1}/fc,y2/u3/0/{1}/fc,y3n2/u3n1/0/{1}/fc,y3n1/u3n/0/{1}/fc,
    u1/y1/0/{-1}/ec, u2/y2/0/{-1}/ec,u3/y3/0/{-1}/ec,u3n2/y3n2/0/{-1}/ec,u3n1/y3n1/0/{-1}/ec,u3n/y3n/0/{-1}/ec, %
    u1/s1/0/{1}/fc, u1/s3/0/1/fc, u1/s3n/0/1/fc,
    s1/s3/30/{-1}/ec, s3/s3n/20/{-1}/ec%
  }
  {
    \draw[->, \typ] (\s) edge[bend right = \aa] node[fill=white, inner sep=1pt, midway] {\small $\w$} (\t);
  }

  \draw[->, fc] (y2) edge[bend right=-40]  node[fill=white, inner sep=1pt, midway] {\small $1$} (ym);
  \draw[->, fc] (ym) edge[bend right=-40]  node[fill=white, inner sep=1pt, midway] {\small $1$}  (y3n1);
  \end{pgfonlayer}
\end{tikzpicture}
  \caption{Illustration of the relevant part of the reduction for \cref{thm:AD-IR-PA}, assuming that element~$1$ appears in sets~$\set{1}, \set{3}$, and $\set{3\nn}$.}\label{fig:thm:AD-IR-PA}
\end{figure}

 By \cref{obs:stability-relation},
 the correctness of the construction follows immediately from \cref{clm:thm:pa_ir_add_add_maxdeg_dag_maxutil_kzero_1,clm:thm:pa_ir_add_add_maxdeg_dag_maxutil_kzero_2}.
}

The hardness remains for symmetric preferences.

\begin{restatable}[\appsymb]{theorem}{thmADIRPANPhardSYMM}
  For each stability concept $\S \in \{\IR,\IS,\NS\}$ and each control action~$\A\in \{\AddAg, \DelAg\}$, \controlprob{\S}{\PA}{\A}{\AD} is \NP-hard even if the budget is $\budget=0$ and the preference graph is symmetric.
  \label{thm:SYMMAD-IR-PA}
\end{restatable}
\appendixproofwithstatement{thm:SYMMAD-IR-PA}{\thmADIRPANPhardSYMM*}
{
  We again provide a reduction from \RETC\ problem. 

  Let $I=(\els=[3\nn], \sets)$ be an instance of \RETC\ and we construct a hedonic game instance with symmetric preferences as follows.
  The set of agents is $\agentsU= \{x\} \cup \{\elag{i} \mid i \in \els\} \cup \{\sset{j} \mid \set{j} \in \sets\} \cup \{\dmag{\ell}{} \mid \ell \in [2\nn]\}$, where 
  \begin{compactitem}[--]
    \item $x$ is an auxiliary-agent
    \item each element~$i\in \els$ has one \myemph{element-agent~$\elag{i}$},
    \item each set~$\set{j}\in \sets$ has one \myemph{set-agent~$\sset{j}$},
    \item and auxiliary-agents $\dmag{\ell}{}$ for $\ell \in [2\nn]$.
  \end{compactitem}
  We construct the symmetric utilities as follows; the unmentioned utilities are 0:
  \begin{compactitem}[--]
    \item For all~$i \in [3\nn]$, set $\util{x}(\elag{i}) = \util{\elag{i}}(x) = -3$.
    \item For all~$\set j \in \sets$, set $\util{x}(\sset{j}) = \util{\sset{j}}(x) = 9$.
    \item For all~$i \in [3\nn - 1]$, set $\util{\elag{i}}(\elag{i+1}) = \util{\elag{i+1}}(\elag{i}) = 1$ and $\util{\elag{1}}(\elag{3\nn}) = \util{\elag{3\nn}}(\elag{1}) = 1$.
    \item For all~$\set j \in \sets$ and all $i \in \set j$, set $\util{\elag{i}}(\sset{j}) = \util{\sset{j}}(\elag{i}) = 1$.
    \item For all~$\set{j}, \set{\ell} \in \sets$ with $\set{j} \cap \set{\ell} \neq \emptyset$, set $\util{\sset j}(\sset \ell) = \util{\sset \ell}(\sset j) = -13$.
    \item For all~$\set{j} \in \sets$ and all $\ell \in [2\nn]$, set $\util{\sset{j}}(\dmag{j}{}) = \util{\dmag{j}{}}(\sset{j}) = 1$. 
    \item For all~$\ell \in [2\nn]$, set $\util{\agx}(\dmag{\ell}{}) = \util{\dmag{\ell}{}}(\agx) = -9$
    \item For all~$\ell \in [2\nn]$ and $i \in [3\nn]$, set $\util{\dmag{\ell}{}}(\elag{i}) = \util{\elag{i}}(\dmag{\ell}{}) = -3$. 
  \end{compactitem}
  The preference graph of the instance is depicted in \cref{fig:thm:SYMMAD-IR-PA}.
  To complete the construction, we define $\{x,\elag{1}\}$ as the pair of agents which we aim to be in the same coalition in an $\{\IR, \IS, \NS\}$ partition.

  In the forward direction of the correctness proof, we show that if $I$ has an exact cover, we can construct a \NS\ partition with $|\Pi(x)| = \Pi(\elag{1})$, which is by \cref{obs:stability-relation} also \IS and \IR.
  In the backward direction, we assume the partition $\Pi$ being \IR--an even weaker assumption than \NS--and show that this implies the existence of an exact cover for~$I$.

  \begin{claim}\label{clm:thm:symmad_ir_pa_forward}
    If $\excov$ is an exact cover for~$I$, then the following partition~$\Pi$
    with $\Pi(x) = \{x, \elag{1}, \dots, \elag{3\nn} \}\cup \{\sset j \mid \set j \in \excov\}$, and 
    for each $\set{j}\notin \excov$, $\Pi(\sset{j}) = \{\sset{j},\dmag{\ell}{}\}$ for one $\ell \in [2\nn]$, is \NS.
  \end{claim}
  \begin{claimproof}{clm:thm:symmad_ir_pa_forward}
    To show the statement, we need to show that every agent $i \in \agentsU$ has no incentive to leave his current coalition, i.e., there exist no other coalition $C \in \Pi$ with $C \cup \{i\} \succ_i \Pi(i)$.

  Let us first consider agent~$x$. 
  Being together with all element-agents $\elag{i}$ leads to utility $3\nn (-3) = -9 \nn$.
  Since $\excov$ is an exact cover, there are $\nn$ set-agents $\sset{j}$ in the same coalition with $x$.
  For each such set-agent, $x$ has positive utility $9$, so agent~$x$ has utility $-9\nn + 9 \nn = 0$ to his coalition.
  For every coalition $\Pi(\sset{j})$ with $\set{j} \notin \excov$, agent~$\agx$ also has utility 0 towards it.
  Therefore, agent~$\agx$ does not want to deviate.

  Every element-agent $\elag{i}$, for $i \in [3\nn]$, has negative utility of~$-3$ towards $x$ and positive utility of~$1$ to his predecessor $\elag{i-1}$ and successor $\elag{i+1}$ 
  (Here, $\elag{1}$ is the successor of $\elag{3\nn}$ and $\elag{3\nn}$ is the predecessor of $\elag{1}$).
  Since $\excov$ is an exact cover, there must exist a set $\set{j} \in \excov$ with $i \in \set{j}$ and so $\util{\elag{i}}(\sset{j}) = 1$. 
  Thus, $\util{\elag{i}}(\Pi(\elag{i})) = \util{\elag{i}}(\{x, \elag{i-1}, \elag{i+1}, \sset{j}\}) = -3 + 1 + 1 +1= 0$.
  For every coalition $\Pi(\sset{j})$ with $\set{j} \notin \excov$, each element-agent can have utility at most -2 towards it, because $\elag{i}$ has utility -3 towards $\dmag{\ell}{}$ and at most utility 1 to $\sset{j}$.
  Thus, $\elag{i}$ does not want to deviate.

  All set-agents $\sset{j}$, for $\set{j} \in \excov$, have positive utility 9 to $x$ and $1$ for each element-agent $\elag{i}$ with $i \in \set{j}$.
  Hence, $\excov$ is an exact cover, there are no two sets $\set{j},\set{\ell} \in \excov$ with $\set{j} \cap \set{\ell} \neq \emptyset$. 
  Therefore, the utility for every set-agent $\sset{j}$ towards his coalition is $\util{\sset{j}}(\Pi(\sset{j})) = \util{\sset{j}}(\{x,\elag{j_1}, \elag{j_2},\elag{j_3}\}) = 9 + 1 + 1 + 1 = 12$.
  Also, the set-agents $\sset{j}$ have no incentive to leave to a coalition with set-agents not in the exact cover, because this would only lead at most to utility 1.
  
  Clearly, every set-agent $\sset{j}$ with $\set{j} \notin \excov$ has utility one towards his coalition containing himself and an auxiliary-agent.
  They would not want to switch to another coalition with an auxiliary-agent because this does not improve the utility.
  Also, the coalition $\Pi(x)$ yields negative utility for $\sset{j}$ because $\set{j}$ intersects with at least two sets in $\excov$, each contributing a utility of -13.

  The auxiliary-agents do not want to join the coalitions $\Pi(\sset{j})$ with $\set{j} \notin \excov$, because it would not yield a higher utility.
  Moreover, the coalition $\Pi(x)$ yields negative utility towards $\dmag{\ell}{}$, since for each element-agent $\elag{i} \in \Pi(x)$, agent~$\dmag{\ell}{}$ has utility -3.

  We showed that no agent has incentive to leave his current coalition, therefore, the constructed partition $\Pi$ is \NS.
  \end{claimproof}

  \begin{claim}\label{clm:thm:symmad_ir_pa_backward}
    If $\Pi$ is an \IR\ partition with $\Pi(x) = \Pi(\elag{1})$, then the sets corresponding to the set-agents contained in $\Pi(x)$ form an exact cover of $I$.
  \end{claim}
  \begin{claimproof}{clm:thm:symmad_ir_pa_backward}
  Let $C$ be the coalition containing agents $x$ and $\elag{1}$.
  We aim to show that the sets corresponding to set-agents contained in $C$ form an exact cover of $I$.

  First, we show that there cannot be an auxiliary-agent $\dmag{\ell}{} \in C$. 
  Suppose, by contradiction, that there is an auxiliary-agent $\dmag{\ell}{} \in C$. 
  This would mean, $\elag{1}$ would have utility -3 towards $\agx$ and $\dmag{\ell}{}$ in $C$.
  However, $\elag{1}$ can obtain at most utility 5 from a coalition (three from a set-agent and two from $\elag{2}$ and $\elag{3\nn}$), so $C$ would be not \IR\ for $\elag{1}$, contradiction. 

  Secondly, we show that there cannot be two set-agents $\sset{j}, \sset{j'} \in C$ with $\set{j}\cap \set{j'}\neq \emptyset$.
  Suppose, towards a contradiction, $C$ contains two set-agents $\sset{j}, \sset{j'}$ with $\set{j}\cap \set{j'}\neq \emptyset$.
  Then, $\sset{j}$ has utility $-13$ towards~$\sset{j'}$ and vice versa.
  Since every set-agent can have at most 12 utility towards any coalition (9 towards $x$ and 1 for the three element-agents), coalition $C$ must yield a negative utility for $\sset{j}$ and $\sset{j'}$ and is not \IR\ for these two, a contradiction.

  Then, we show that all element-agents $\elag{i}$ for $i \in [3\nn]$ are contained in $C$. 
  Agent~$\elag{1}$ has a negative utility $-3$ towards agent~$x$.
  To make coalition~$C$ \IR\ for $\elag{1}$, we have to include at least one set-agent $\sset{j}$ with $1 \in \set{j}$.
  As shown before, we cannot add two sets which have common elements, so we can add at most one set-agent $\sset{j}$ with $1 \in \set{j}$.
  Therefore, we also need to add $\elag{2}$ and $\elag{3\nn}$ to make $C$ \IR\ for $\elag{1}$. 
  The same argument holds for $\elag{2}, \dots, \elag{3 \nn - 1}$ that they have to be together with their predecessor and successor.
  
  Therefore, for every element $i$, the coalition~$C$ must contain at least one set-agent $\sset{j}$ with $i \in \set{j}$ in order to make $C$ \IR\ for every element-agent~$\elag{i}$.
  This implies that the sets corresponding to the set-agents in~$C$ form a set cover.

  Let $\excov=\{\set{j}\mid \sset{j}\in C\}$.
  Because of our first observation, $\excov$ is indeed an exact cover.
  \end{claimproof}
  
 \begin{figure}[t!]
  \centering
  \begin{tikzpicture}[>=stealth',shorten <= 1pt, shorten >= 1pt]
  \def \xx {1.3}
  \def \xy {1}
  \foreach \x / \n in {1/1, 2/2, 3/3, 4/3n2, 5/3n1, 6/3n} {
    \node[unode] at (\x*\xx, -3*\xy) (s\n) {};
  }

  \node[unode] at (3.5*\xx, 0*\xy) (x) {};
  
  \foreach \x / \n in {1/1,2/2,3/3,4/3n2, 5/3n1, 6/3n} {
    \node[unode] at (\x*\xx, -1.5*\xy) (u\n) {};
  }

  \foreach \x / \n in {2/1, 5/2n} {
    \node[unode] at (\x*\xx, -4*\xy) (d\n) {};
  }

  \foreach \i / \p / \a / \r / \n in {%
    x/above left/2pt/0pt/{x},
    u1/below left/3pt/-12pt/{y = u_1},   u2/below left/0pt/-5pt/{u_2},   u3/below left/0pt/-5pt/{u_3},   u3n2/below left/0pt/-5pt/{u_{3\nn-2}},  u3n1/below left/0pt/-14pt/{u_{3\nn-1}},  u3n/below left/0pt/-5pt/{u_{3\nn}}, 
    s1/below left/-5pt/2pt/{s_1},   s2/below left/-5pt/2pt/{s_2},   s3/below left/0pt/-5pt/{s_3},   s3n2/below left/1pt/-8pt/{s_{3\nn-2}},  s3n1/below left/0pt/-14pt/{s_{3\nn-1}},  s3n/below right/-5pt/2pt/{s_{3\nn}},
    d1/below right/0pt/2pt/{d_1}, d2n/below left/0pt/2pt/{d_{2\nn}}%
  } {
     \node[\p = \a and \r of \i,fill=white,inner sep=0.1pt] {$\n$};
   }
  \foreach \y in {u, s} {
    \path (\y 3) -- node[midway,inner sep=1.7pt] (\y m) {$\dots$} (\y 3n2);
  }

  \coordinate (h1) at ($(s3)!0.5!(d1)$);
  \coordinate (h2) at ($(s3)!0.5!(d2n)$);  
 
  \begin{pgfonlayer}{bg}
    \foreach \s / \t / \aa / \w / \typ in {
    x/u1/0/{-3}/ec, x/u2/0/{-3}/ec, x/u3/0/{-3}/ec, x/u3n2/0/{-3}/ec, x/u3n1/0/{-3}/ec, x/u3n/0/{-3}/ec,
    u1/u3n/-15/{1}/fc, u1/u2/0/{1}/fc, u2/u3/0/{1}/fc, u3n2/u3n1/0/{1}/fc, u3n1/u3n/0/{1}/fc,
    u1/s1/0/{1}/fc, u1/s3/0/1/fc, u1/s3n/0/1/fc,
    s1/d1/0/{1}/fc, s3n/d2n/0/{1}/fc,
    d1/h1/0/{1}/fc, d2n/h2/0/{1}/fc,
    s1/s3/30/{-13}/ec, s3/s3n/24/{-13}/ec%
  }
  {
    \draw[-, \typ] (\s) edge[bend right = \aa] node[fill=white, inner sep=1pt, midway] {\small $\w$} (\t);
  }

  \draw[-, fc] (x) .. controls ($(u1)+(0,1.5)$) and  ($(u1)+(-1.5,0)$)   .. (s1) node[fill=white, inner sep=1pt, midway] {\small $9$};
  \draw[-, fc] (x) .. controls ($(u3n)+(0,1.5)$) and  ($(u3n)+(1.5,0)$)   .. (s3n) node[fill=white, inner sep=1pt, midway] {\small $9$};

  \draw[-, ec] (d1) .. controls ($(s1)-(0,1.5)$) and  ($(s1)-(1.5,0)$)   .. (u1) node[fill=white, inner sep=1pt, midway] {\small $-3$};
  \draw[-, ec] (d2n) .. controls ($(s3n)-(0,1.5)$) and  ($(s3n)+(1.5,0)$)   .. (u3n) node[fill=white, inner sep=1pt, midway] {\small $-3$};
  \end{pgfonlayer}
  \end{tikzpicture}
    \caption{Illustration of the relevant part of the reduction for \cref{thm:SYMMAD-IR-PA}, assuming that element~$1$ appears in sets~$\set{1}, \set{3}$, and $\set{3\nn}$. Note that not all utilities from the auxiliary-agents $\dmag{\ell}{}$ are drawn.}\label{fig:thm:SYMMAD-IR-PA}
  \end{figure}
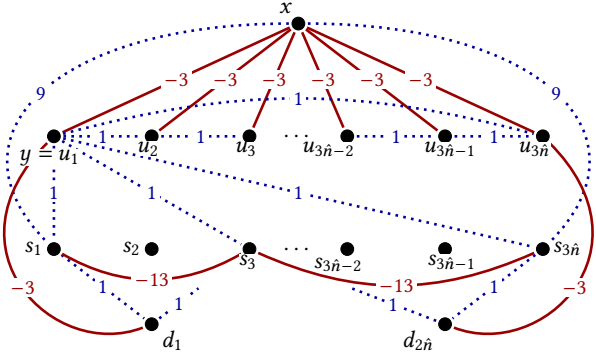

 The correctness of the construction follows immediately from Claims \ref{clm:thm:symmad_ir_pa_forward} and \ref{clm:thm:symmad_ir_pa_backward}.
}

Finally, we consider the control goal \GR: Making the grand coalition partition stable.
We already know from the previous section that for every $\S \in \{\IR, \IS,\NS, \CS\}$, \controlprob{\S}{\GR}{\AddAg}{\FA} is \NP-hard and \Wtwo-hard wrt.\ $\budget$ even for symmetric preferences.
We strengthen this result by showing the same for the \DelAg\ case. %

\begin{restatable}[\appsymb]{theorem}{thmDAGSYMADIRISNSGRADD}
  For each stability concept $\S \in \{\IR, \IS,\NS, \CS\}$ and each control action $\A \in \{\AddAg, \DelAg\}$, \controlprob{\S}{\GR}{\A}{\AD} is in \XP\ and \Wtwo-hard wrt.\ $\budget$ when the preference graph is a DAG.
  \label{thm:DAG-AD-IRISNS-GR-ADD}
\end{restatable}

\appendixproofwithstatement{thm:DAG-AD-IRISNS-GR-ADD}{\thmDAGSYMADIRISNSGRADD*}{
Throughout this proof, given an instance $I = (\els = [\nn], \sets, \scs)$ of \scp, for every $i \in [\nn]$, let \setdeg i denote the number of sets containing $i$.
Recall that, by \cref{obs:gra:eq,obs:dag:eq}, \IR, \IS, \NS, and \CS\ are all equivalent for the grand coalition partition.
Thus, we present the whole proof only for \IR.
We present the proof for the two different control actions separately.

\begin{restatable}{claim}{clmSYMADIRISNSGRDEL} \label{clm:DAG-AD-IRISNS-GR-ADD}
	\controlprob{\IR}{\GR}{\AddAg}{\AD} is \Wtwo-hard wrt.\ $\budget$ even when the preference graph is a DAG.
\end{restatable}

\begin{claimproof}{clm:DAG-AD-IRISNS-GR-ADD}
We reduce from \scp.
Let $I = (\els = [\nn], \sets,$ $\scs)$ be an instance of \scp.
W.l.o.g., we assume that every element~$i \in \els$ appears in at least one member of $\sets$.

We construct an instance of \controlprob{\IR}{\GR}{\AddAg}{\AD}.
The set of agents is $\agentsU \coloneqq \{u_i \mid i \in [\nn]\} \cup \{b\}$ and $\agentsW \coloneqq \{s_j \mid \set{j} \in \sets\}$.
We construct the utilities as follows, the unmentioned utilities are 0:
\begin{compactitem}[--]
	\item For every $i \in [\nn]$, we set $\util{u_i}(b) = -1$.
	\item For every $\set j \in \sets$, we set  $\util{u_i}(s_j)  = 1$ for every $i \in \set j$.
\end{compactitem}
We set $\budget \coloneqq \scs$.
Clearly $u_1, \dots, u_{\nn}, s_1, \dots, s_{m}, b$ is a topological order of the preference graph.

First assume that $I$ admits a cover $\excov$ such that $|\excov| \leq \scs$.
Let us construct $\sss \coloneqq \{s_j \mid \set j \in \excov\}$.
Clearly $|\sss| = |\excov | \leq \scs$.
It remains to show that the grand coalition partition is an \IR\ partition of $\agentsU \cup \sss$.
 
The agents in $\sss \cup \{b\}$ have no out-arcs, so they do not wish to deviate.
For every $i \in [\nn]$, because $\excov$ is a cover, there must be at least one $\set j \in \excov$ such that $i \in \set j$.
We obtain that   $\sum_{a \in \agentsU \cup \agentsW}\util{u_i}(a)  =\sum_{a \in \agentsU}\util{u_i}(a) + \sum_{a \in \sss}\util{u_i}(a) \ge -1 + 1 \ge 0$ and thus $u_i$ does not wish to deviate from the grand coalition.
Hence, the grand coalition partition is an \IR\ partition of $\agentsU \cup \sss$.\\

Now assume that there is a subset $\sss \subseteq \agentsW$ such that $|\sss| \leq \scs$ and 
the grand coalition partition is an \IR\ partition of $\agentsU \cup \sss$.
Let us define $\excov \coloneqq \{\set j \in \sets \mid s_j \in \sss\}$.
Assume, towards a contradiction, that there is an element $i \in [\nn]$ such that no set in $\excov$ contains $i$.
We obtain that  $\sum_{a \in \agentsU \cup \sss}\util{u_i}(a)  =\sum_{a \in \agentsU}\util{u_i}(a) + \sum_{a \in \sss}\util{u_i}(a) \ge -1 + 0 = -1$ and thus $u_i$ deviates from the grand coalition, a contradiction to individual rationality.
Thus, $\excov$ is a cover for $I$.
\end{claimproof}

\begin{restatable}{claim}{clmSYMADIRISNSGRDEL} \label{clm:DAG-AD-IRISNS-GR-DEL}
 \controlprob{\IR}{\GR}{\DelAg}{\AD} is \Wtwo-hard wrt.\ $\budget$ even when the preference graph is a DAG.
\end{restatable}

\begin{claimproof}{clm:DAG-AD-IRISNS-GR-DEL}
We reduce from \scp.
Let $I = (\els = [\nn], \sets, \scs)$ be an instance of \scp.
	
We construct an instance of \controlprob{\IR}{\GR}{\DelAg}{\AD}.
The set of agents is $\agentsU \coloneqq \{u_i \mid i \in [\nn]\} \cup \{a^z_i \mid i \in [\nn], z \in [\setdeg i - 1]\} \cup \{s_j \mid \set{j} \in \sets\}$.
We construct the utilities as follows, the unmentioned utilities are 0:
\begin{compactitem}[--]
	\item For every $i \in [\nn]$, we set $\util{u_i}(a^z_i) = 1$ for every $ z \in [\setdeg i - 1]$.
	\item For every $\set j \in \sets, i \in \set j$, we set $\util{u_i}(s_j)= -1$.
\end{compactitem}
We set $\budget \coloneqq \scs$. Clearly $u_1, \dots, u_{\nn}, s_1, \dots, s_m, a^1_1, \dots, a^{\setdeg 1}_1, \dots, a^1_m, \dots,$ $a^{\setdeg m}_m$ is a topological order of the preference graph.
	
First assume that $I$ admits a cover $\excov$ such that $|\excov| \leq \scs$.
Let us construct $\agentsU' \coloneqq \{s_j \mid \set j \in \excov\}$.
Clearly $|\agentsU'| = |\excov | \leq \scs$.
It remains to show that the grand coalition partition is an \IR\ partition of $\agentsU \setminus \agentsU'$.
	
The agents in $ \{a^z_i \mid i \in [\nn], z \in [\setdeg i - 1]\} \cup \{s_j \mid \set{j} \in \sets\} \setminus \agentsU'$ do not wish to deviate, because they have no non-zero utilities.
For every $i \in [\nn]$, recall that, since $\excov$ is a cover, there is at least one set in it that contains $i$.
Thus, we obtain that  $\sum_{x \in \agentsU \setminus \agentsU'}\util{u_i}(x)  =\sum_{x \in \agentsU}\util{u_i}(x) - \sum_{x \in \agentsU'}\util{u_i}(x) \geq (\setdeg i - 1 - \setdeg i) + 1 \geq 0$, and thus $u_i$ does not deviate.
Therefore, the grand coalition partition is an \IR\ partition of $\agentsU \setminus \agentsU'$.\\
	
Now assume that there is a subset $\agentsU' \subseteq \agentsU$ such that $|\agentsU'| \leq \scs$ and	the grand coalition partition is an \IR\ partition of $\agentsU \setminus \agentsU'$.
	
First observe that if an agent in  $ \{a^z_i \mid i \in [\nn], z \in [\setdeg i - 1]\}$ is in $\agentsU'$, we can remove this agent from $\agentsU'$ and the grand coalition partition of $\agentsU \setminus \agentsU'$ remains \IR; this holds because no agent has negative utility towards him and he has no negative utility towards anyone.
Thus we can assume that $ \{a^z_i \mid i \in [\nn], z \in [\setdeg i - 1]\} \cap \agentsU' = \emptyset$.
	
Next we show that if an agent $u_i$ is in $\agentsU'$ for some $i \in [\nn]$, then we can replace $u_i$ with $s_j$ for some $\set j \in \sets$ such that $i \in \set j, s_j \notin \agentsU'$.
If no such $s_j$ exists, we just remove $u_i$ from $\agentsU'$.
Call the updated~$\agentsU'$ with the name $\hat{\agentsU}$.
No agent in $\{u_{i'} \mid i' \in [\nn]\} \setminus (\agentsU' \cup \{u_i\})$ has negative utility towards $u_i$ or positive utility towards $s_j$, so if they do not wish to deviate from $\agentsU \setminus \agentsU'$, they also do not wish to deviate from $\agentsU \setminus \hat{\agentsU}$.
The agents in $\{s_j \mid \set{j} \in \sets\}$ have no non-zero utilities towards other agents, so they never deviate from any coalition.
It remains to show that $u_i$ does not wish to deviate from the new grand coalition partition of $\agentsU \setminus \hat\agentsU$.
Recall that no agent from $ \{a^z_i \mid i \in [\nn], z \in [\setdeg i - 1]\}$ is in $\hat\agentsU$ and there is a set $\set j$ such that $s_j \in \agentsU$ and $u_i$ has utility $-1$ towards $s_j$.
Thus we obtain that  $\sum_{x \in \agentsU \setminus \hat\agentsU}\util{u_i}(x)  =\sum_{x \in \agentsU}\util{u_i}(x) - \sum_{x \in \hat\agentsU}\util{u_i}(x) \geq (\setdeg i - 1 - \setdeg i) + 1 \geq 0$, and thus $u_i$ does not deviate from $\agentsU \setminus \hat{\agentsU}$.
	
We have now shown that we may assume $\hat{\agentsU} \subseteq \{s_j \mid \set{j} \in \sets\}$.
Let $\excov \coloneqq \{\set j \in \sets \mid s_j \in \hat{\agentsU}\}$.
Clearly $|\excov| = |\hat{\agentsU}| \leq \scs$.
It remains to show that $\excov$ covers $[\nn]$.
Assume, towards a contradiction, that there is an element $i \in [\nn]$ such that no set in $\excov$ contains $i$.
Then  $\sum_{x \in \agentsU \setminus \hat\agentsU}\util{u_i}(x)  =\sum_{x \in \agentsU}\util{u_i}(x) - \sum_{x \in \hat\agentsU}\util{u_i}(x) = (\setdeg i - 1 - \setdeg i) - 0 = -1$, and $u_i$ prefers being alone to the grand coalition, a contradiction to individual rationality.
Thus, $\excov$ is a cover for $I$.
\end{claimproof}

It remains to show that the problems are in \XP. A simple brute-force approach is sufficient to show this.

\begin{restatable}{claim}{clmDAGSYMADIRISNSGRADD}
  For each  control action $\A \in \{\AddAg, \DelAg\}$, \controlprob{\IR}{\GR}{\A}{\AD} is in \XP\  wrt.\ $\budget$.
  \label{clm:AD-IRISNS-GR-ADD}
\end{restatable}

\begin{claimproof}{clm:AD-IRISNS-GR-ADD}
Observe that we can verify in polynomial time whether the grand coalition partition is \IR.
Thus we can try every subset of $\agentsW$ (resp.\ $\agentsU$) of cardinality at most $\budget$.
The number of such subsets is in $O(|\agentsW|^{\budget})$ (resp.\ $O(|\agentsU|^{\budget})$), and thus this simple brute-force algorithm is in \XP\ wrt.~$\budget$.
\end{claimproof} 
This concludes the proof.
}

\begin{restatable}[\appsymb]{theorem}{thmSYMADIRISNSGRADD}
  For each stability concept $\S \in \{\IR, \IS,\NS\}$ and each control action $\A \in \{\AddAg, \DelAg\}$, \controlprob{\S}{\GR}{\A}{\AD} is \Wtwo-hard and in \XP\ wrt.\ $\budget$. The hardness holds even when the preference graph is symmetric.
  \label{thm:SYM-AD-IRISNS-GR-ADD}
\end{restatable}

\newcommand{\setdeg}[1]{\ensuremath{d(#1)}}
\appendixproofwithstatement{thm:SYM-AD-IRISNS-GR-ADD}{\thmSYMADIRISNSGRADD*}{
Throughout this proof, given an instance $I = (\els = [\nn], \sets, \scs)$ of \scp, for every $i \in [\nn]$, let \setdeg i denote the number of sets containing $i$.
Recall that, by \cref{obs:gra:eq}, \IR, \IS, and  \NS\ are equivalent for the grand coalition partition.
Thus we present the whole proof only for \IR.

For symmetric preferences and the control action \AddAg\ the hardness is a consequence of \cref{thm:FE-IR-GR-ADD}.
Thus it is sufficient to show the hardness for \DelAg.

\begin{restatable}{claim}{clmSYMADIRISNSGRDEL} \label{clm:SYM-AD-IRISNS-GR-DEL}
  For each stability concept $\S \in \{\IR, \IS,\NS\}$, \controlprob{\S}{\GR}{\DelAg}{\AD} is \Wtwo-hard wrt.\ $\budget$ even when the preference graph is symmetric.
\end{restatable}

\begin{claimproof}{clm:SYM-AD-IRISNS-GR-DEL}
We reduce from \scp.
Let $I = (\els = [\nn], \sets, \scs)$ be an instance of \scp.

We construct an instance of \controlprob{\S}{\GR}{\DelAg}{\AD} with symmetric utilities.
The set of agents is $\agentsU \coloneqq \{u_i \mid i \in [\nn]\} \cup \{a^z_i \mid i \in [\nn], z \in [\setdeg i - 1]\} \cup \{s_j \mid \set{j} \in \sets\} \cup \{b_j^w \mid  \set{j} \in \sets, w \in [|\set j| + |\sets|]\}$.
We construct the utilities as follows, the unmentioned utilities are~0:
\begin{compactitem}[--]
\item For every $i \in [\nn]$, we set $\util{u_i}(a^z_i) = \util{a^z_i}(u_i) = 1$ for every $ z \in [\setdeg i - 1]$.
\item For every $\set j \in \sets$, we set  $\util{s_j}(b^w_j) = \util{b^w_j}(s_j) = 1$ for every $ w \in [|\set j| + |\sets|]$.
\item For every $\set j \in \sets, i \in \set j$, we set $\util{u_i}(s_j) = \util{s_j}(u_i) = -1$.
\item For every $\set j, \set{j'} \in \sets$, we set $\util{s_j}(s_{j'}) = \util{s_{j'}}(s_j) = -1$.
\end{compactitem}
We set $\budget \coloneqq \scs$.

First assume that $I$ admits a cover $\excov$ such that $|\excov| \leq \scs$.
Let us construct $\agentsU' \coloneqq \{s_j \mid \set j \in \excov\}$.
Clearly $|\agentsU'| = |\excov | \leq \scs$.
It remains to show that the grand coalition partition is an \IR\ partition of $\agentsU \setminus \agentsU'$.

The agents in $ \{a^z_i \mid i \in [\nn], z \in [\setdeg i - 1]\} \cup  \{b_j^w \mid  \set{j} \in \sets, w \in [|\set j| + |\sets|]\}$ do not wish to deviate, because they have only non-negative utilities towards other agents.
For every $\set j \in \sets \setminus \excov$, we have that $\sum_{a \in \agentsU \setminus \agentsU'}\util{s_j}(a) \geq |\set j| + |\sets| - |\sets| - |\set j| \geq 0$, and thus~$s_j$ does not deviate.
For every $i \in [\nn]$, recall that, since $\excov$ is a cover, there is at least one set in it that contains $i$.
Thus we obtain that  $\sum_{a \in \agentsU \setminus \agentsU'}\util{u_i}(a)  =\sum_{a \in \agentsU}\util{u_i}(a) - \sum_{a \in \agentsU'}\util{u_i}(a) \geq (\setdeg i - 1 - \setdeg i) + 1 \geq 0$, and thus $u_i$ does not deviate.
Therefore the grand coalition partition is an \IR\ partition of $\agentsU \setminus \agentsU'$.\\

Now assume that there is a subset $\agentsU' \subseteq \agentsU$ such that $|\agentsU'| \leq \scs$ and 
the grand coalition partition is an \IR\ partition of $\agentsU \setminus \agentsU'$.

First observe that if an agent in  $ \{a^z_i \mid i \in [\nn], z \in [\setdeg i - 1]\} \cup  \{b_j^w \mid  \set{j} \in \sets, w \in [|\set j| + |\sets|]\}$ is in $\agentsU$, we can remove this agent from $\agentsU'$ and the grand coalition partition of $\agentsU \setminus \agentsU'$ remains \IR; this holds because no agent has negative utility towards him and he has no negative utility towards anyone.
Thus we can assume that $ \{a^z_i \mid i \in [\nn], z \in [\setdeg i - 1]\} \cup  \{b_j^w \mid  \set{j} \in \sets, w \in [|\set j| + |\sets|]\} \cap \agentsU' = \emptyset$.

Next we show that if an agent $u_i$ is in $\agentsU'$ for some $i \in [\nn]$, then we can replace $u_i$ with $s_j$ for some $\set j \in \sets$ such that $i \in \set j, s_j \notin \agentsU'$.
If no such $s_j$ exists, we just remove $u_i$ from $\agentsU'$.
Call the updated~$\agentsU'$ with the name $\hat{\agentsU}$.
No agent in $\{u_{i'} \mid i' \in [\nn]\} \setminus (\agentsU' \cup \{u_i\})$ has negative utility towards $u_i$ or positive utility towards $s_j$, so if they do not wish to deviate from $\agentsU \setminus \agentsU'$, they also do not wish to deviate from $\agentsU \setminus \hat{\agentsU}$.
Every agent in $ \{s_{j'} \mid \set{j'} \in \sets\} \setminus (\agentsU' \cup \{s_j\}) $ who has negative utility $u_i$ has the same negative utility for $s_j$, so if they do not wish to deviate from $\agentsU \setminus \agentsU'$, they also do not wish to deviate from $\agentsU \setminus \hat{\agentsU}$.
Moreover, the only agents who have negative utility towards $u_i$ are the agents $s_{j'}$ where $\set{j'} \in \sets, i \in \set j$; if we are in the case where we remove $u_i$ without adding any agents to $\agentsU$, then all of the agents who have negative utility towards $u_i$ are already in~$\agentsU'$.

It remains to show that $u_i$ does not wish to deviate from the new grand coalition partition of $\agentsU \setminus \hat\agentsU$.
Recall that no agent from $ \{a^z_i \mid i \in [\nn], z \in [\setdeg i - 1]\} \cup  \{b_j^w \mid  \set{j} \in \sets, w \in [|\set j| + |\sets|]\}$ is in $\agentsU'$ and there is a set $\set j$ such that $s_j \in \agentsU$ and $u_i$ has utility $-1$ towards $s_j$.
Thus we obtain that  $\sum_{a \in \agentsU \setminus \hat\agentsU}\util{u_i}(a)  =\sum_{a \in \agentsU}\util{u_i}(a) - \sum_{a \in \hat\agentsU}\util{u_i}(a) \geq (\setdeg i - 1 - \setdeg i) + 1 \geq 0$, and thus $u_i$ does not deviate from $\agentsU \setminus \hat{\agentsU}$.

We have now shown that we may assume $\hat\agentsU \subseteq \{s_j \mid \set{j} \in \sets\}$.
Let $\excov \coloneqq \{\set j \in \sets \mid s_j \in \hat\agentsU\}$.
Clearly $|\excov| = |\hat\agentsU| \leq \scs$.
It remains to show that $\excov$ covers $[\nn]$.
Assume, towards a contradiction, that there is an element $i \in [\nn]$ such that no set in $\excov$ contains $i$.
Then  $\sum_{a \in \agentsU \setminus \hat\agentsU}\util{u_i}(a)  =\sum_{a \in \agentsU}\util{u_i}(a) - \sum_{a \in \hat\agentsU}\util{u_i}(a) = (\setdeg i - 1 - \setdeg i) - 0 = -1$, and $u_i$ prefers being alone to the grand coalition, a contradiction to individual rationality.
Thus $\excov$ covers~$[\nn]$.
\end{claimproof}

To show \XP\ containment observe that \cref{clm:AD-IRISNS-GR-ADD} does not require the graph to be a DAG.
}

For core stability, \cref{thm:DAG-AD-IRISNS-GR-ADD} shows that \controlprob{\CS}{\GR}{\AddAg}{\AD} and \controlprob{\CS}{\GR}{\DelAg}{\AD} admit a polynomial-time algorithm when the preference graph is a DAG and the budget is a constant.
However, in contrast with \IR, \IS, and \NS, the problems for \CS\ are in general \coNP-complete even when $\budget = 0$.
The hardness holds even when the preferences are symmetric.
     
\begin{restatable}[\appsymb]{theorem}{thmADCSGRcoNP}
  For each control action~$\A\in \{\AddAg, \DelAg\}$,
  \controlprob{\CS}{\GR}{\A}{\AD} is \coNP-complete.
  It remains \coNP-hard even when $\budget = 0$ and the preference graph is symmetric.
  \label{thm:AD-CS-GS}
\end{restatable}

\newcommand{\verdeg}[1]{\ensuremath{\Delta(#1)}}

\appendixproofwithstatement{thm:AD-CS-GS}{\thmADCSGRcoNP*}{
We show \coNP-hardness for  \controlprob{\CS}{\GR}{\AddAg}{\AD}; containment is discussed in \cref{obs:complexityupperbounds}.
Since $\budget = 0$, the proof for  \controlprob{\CS}{\GR}{\DelAg}{\AD} is analogous.
We reduce from the classic \np-complete problem \cliqprob. 
\decprob{\cliqprob}{An undirected graph $G = (V, E)$, an integer $\scs$.}{Does $G$ admit a clique of size at least $\scs$?}

Let $I = (G = (V, E), \scs)$ be an instance of \cliqprob\ and we construct an \AD-instance with a symmetric preference graph $\utilG$.
We use \verdeg{v} to denote the degree of a vertex $v$.
The set of agents is $\agentsU = \{a_v \mid v \in V\} \cup \{b\}$, where $a_v, v \in V$ is a vertex-agent, and $b$ is a dummy-agent.

  We construct the utilities as follows:
  \begin{compactitem}[--]
	\item For each $v, v' \in V$ such that $v$ and $v'$ are adjacent in $G$, set $\util{a_v}(a_{v'}) = \util{a_{v'}}(a_v) = 1.$
  \item For each $v, v' \in V$ such that $v$ and $v'$ are not adjacent, set $\util{a_v}(a_{v'}) = \util{a_{v'}}(a_v) = - |V|.$
	\item For each $v \in V$, set $\util{a_v}(b) = \util{b}(a_v) = |V| \cdot (|V| - \verdeg{v} - 1) - \verdeg{v} + \scs - 2.$ We assume that every vertex has at least one vertex which is not adjacent to them, so $b$ has positive utility towards every agent.
  \end{compactitem}
The budget $\budget$ is zero.

We show that $I$ is a \yes\ of \cliqprob\ if and only if $\utilG$ is a \no\ of \controlprob{\CS}{\GR}{\AddAg}{\AD}.

Observe that for every $v \in V$
\begin{align}
\sum_{x \in \agentsU} \util{a_v}(x) &=  \sum_{v' \in V \setminus N(v)}  \util{a_v}(a_{v'})  + \sum_{v' \in N(v)}  \util{a_v}(a_{v'}) + \util{a_v}(b) \nonumber\\
&= (-|V| \cdot (|V| - \verdeg{v} - 1)) + \verdeg{v} \nonumber\\ &\quad+ (|V| \cdot (|V| - \verdeg{v} - 1) - \verdeg{v} + \scs - 2) \nonumber\\
&=  \scs - 2.\label{eq:AD-CS-GS}
\end{align}

First assume $I$ is a \yes\ of \cliqprob.
Let $K$ be a clique of cardinality at least $\scs$.
We show that $\coal \coloneqq \{a_v \mid v \in K\}$ blocks the grand coalition partition of $\agentsU$.
For every $v \in K$, the utility $a_v$ obtains under $\coal$ is $\sum_{a_{v'} \in \coal} \util{a_v}(a_{v'}) = \scs - 1 \stackrel{\eqref{eq:AD-CS-GS}}{>}\sum_{x \in \agentsU} \util{a_v}(x)$, because every pair of vertices in a clique is adjacent.
Thus $\coal$ blocks the grand coalition partition and $\utilG$ is a \no.\\

Now assume $\utilG$ is a \no\ for \controlprob{\CS}{\GR}{\AddAg}{\AD}, i.e., some coalition $\coal$ blocks the grand coalition partition of $\agentsU$.
Observe that, since $b$ obtains all of the agents he has positive utility towards and he has no negative utilities, he cannot be in $\coal$.
Thus $\coal$ must consist of vertex-agents only.
Let $K \coloneqq \{v \in V \mid a_v \in \coal\}$.

Let us first show that $K$ must be a clique.
Assume, towards a contradiction, that there is a pair of vertices $v, v' \in K$ that are not adjacent.
Then $\sum_{x \in \coal} \util{a_v}(x) = \util{a_v}(a_{v'}) + \sum_{\hat{v} \in K \setminus \{v'\}} \util{a_v}(a_{\hat{v}}) \leq -|V| + |V| - 1 < 0 \stackrel{\eqref{eq:AD-CS-GS}}{<} \sum_{x \in \agentsU} \util{a_v}(x)$, and thus $\coal$ cannot block the grand coalition partition.

Next  assume that $K$ is a clique of cardinality at most $\scs - 1$.
Then $\sum_{v' \in K} \util{a_v}(a_{v'}) = |K|  - 2\leq \scs - 2 \stackrel{\eqref{eq:AD-CS-GS}}{=} \sum_{x \in \agentsU} \util{a_v}(x) $ and thus $\coal$ does not block the grand coalition partition, a contradiction.
Thus $K$ must be a clique of size at least $\scs$.
This concludes the proof.
}

\section{Conclusion}\label{sec:conclusion}

Motivated by control in other computational social choice problems,
we introduce control in hedonic games. 
We study three control goals: ensuring an agent is not alone~(\NA), ensuring a pair of agents is together~(\PA), and ensuring the grand coalition is stable~(\GR), combined with two control actions: adding and deleting agents.
We present a complete complexity picture for these control goals and actions across four stability concepts and two preference representations--\FA\ and \AD.

Our work opens several potential directions for future research.
First, alternative control actions remain unexplored. 
In Stable Roommates and Marriage settings, a common control action is removing acceptability--making previously acceptable pairs unacceptable to each other.
Analogously, one could study removing friendship relations in \FA\ or adjusting utility values in \AD.
Second, inspired by destructive control in voting~\cite{Hemaspaandra07_Preclude}, one could study destructive control in hedonic games: ensuring a specific agent remains isolated or preventing a specific pair from being in the same coalition.
Third, other solution concepts merit investigation, such as Pareto optimal or strictly core stable partitions.
Finally, other simple compact preference representations warrant exploration, such as fractional hedonic games~\cite{aziz2019fractionalhg}, anonymous preferences~\cite{Bogomolnaia02_NS_Symmetric} or B- and W-preferences~\cite{cechlarova2001stability}.
We conjecture that the ideas of many of our hardness reductions could work also for fractional hedonic games, since \cref{obs:originalhard-follows} extends to any compact preference representation. 
However, anonymous, B- and W-preferences have a different structure and therefore likely require different ideas.

\clearpage

\bibliographystyle{ACM-Reference-Format} 
\bibliography{bib}


\begin{thebibliography}{31}


\ifx \showCODEN    \undefined \def \showCODEN     #1{\unskip}     \fi
\ifx \showDOI      \undefined \def \showDOI       #1{#1}\fi
\ifx \showISBNx    \undefined \def \showISBNx     #1{\unskip}     \fi
\ifx \showISBNxiii \undefined \def \showISBNxiii  #1{\unskip}     \fi
\ifx \showISSN     \undefined \def \showISSN      #1{\unskip}     \fi
\ifx \showLCCN     \undefined \def \showLCCN      #1{\unskip}     \fi
\ifx \shownote     \undefined \def \shownote      #1{#1}          \fi
\ifx \showarticletitle \undefined \def \showarticletitle #1{#1}   \fi
\ifx \showURL      \undefined \def \showURL       {\relax}        \fi
\providecommand\bibfield[2]{#2}
\providecommand\bibinfo[2]{#2}
\providecommand\natexlab[1]{#1}
\providecommand\showeprint[2][]{arXiv:#2}

\bibitem[\protect\citeauthoryear{Aziz, Brandl, Brandt, Harrenstein, Olsen, and
  Peters}{Aziz et~al\mbox{.}}{2019}]%
        {aziz2019fractionalhg}
\bibfield{author}{\bibinfo{person}{Haris Aziz}, \bibinfo{person}{Florian
  Brandl}, \bibinfo{person}{Felix Brandt}, \bibinfo{person}{Paul Harrenstein},
  \bibinfo{person}{Martin Olsen}, {and} \bibinfo{person}{Dominik Peters}.}
  \bibinfo{year}{2019}\natexlab{}.
\newblock \showarticletitle{Fractional Hedonic Games}.
\newblock \bibinfo{journal}{\emph{ACM Transactions on Economics and
  Computation}} \bibinfo{volume}{7}, \bibinfo{number}{2}
  (\bibinfo{year}{2019}).
\newblock
\showISSN{2167-8375}


\bibitem[\protect\citeauthoryear{Aziz and Savani}{Aziz and Savani}{2016}]%
        {AZ2016HG-Bookchapter}
\bibfield{author}{\bibinfo{person}{Haris Aziz} {and} \bibinfo{person}{Rahul
  Savani}.} \bibinfo{year}{2016}\natexlab{}.
\newblock \showarticletitle{Hedonic Games}.
\newblock In \bibinfo{booktitle}{\emph{Handbook of Computational Social
  Choice}}, \bibfield{editor}{\bibinfo{person}{Felix Brandt},
  \bibinfo{person}{Vincent Conitzer}, \bibinfo{person}{Ulle Endriss},
  \bibinfo{person}{J{\'{e}}r{\^{o}}me Lang}, {and} \bibinfo{person}{Ariel~D.
  Procaccia}} (Eds.). \bibinfo{publisher}{Cambridge University Press},
  \bibinfo{pages}{356--376}.
\newblock


\bibitem[\protect\citeauthoryear{Banerjee, Konishi, and S{\"o}nmez}{Banerjee
  et~al\mbox{.}}{2001}]%
        {Banerjee01_Core}
\bibfield{author}{\bibinfo{person}{Suryapratim Banerjee},
  \bibinfo{person}{Hideo Konishi}, {and} \bibinfo{person}{Tayfun S{\"o}nmez}.}
  \bibinfo{year}{2001}\natexlab{}.
\newblock \showarticletitle{Core in a Simple Coalition Formation Game}.
\newblock \bibinfo{journal}{\emph{Social Choice and Welfare}}
  \bibinfo{volume}{18}, \bibinfo{number}{1} (\bibinfo{year}{2001}),
  \bibinfo{pages}{135--153}.
\newblock


\bibitem[\protect\citeauthoryear{{Bartholdi~III}, Tovey, and
  Trick}{{Bartholdi~III} et~al\mbox{.}}{1992}]%
        {Bartholdi92_Control}
\bibfield{author}{\bibinfo{person}{John {Bartholdi~III}},
  \bibinfo{person}{Craig~A.\ Tovey}, {and} \bibinfo{person}{Michael~A.\
  Trick}.} \bibinfo{year}{1992}\natexlab{}.
\newblock \showarticletitle{How hard is it to control an election?}
\newblock \bibinfo{journal}{\emph{Mathematical and Computer Modelling}}
  \bibinfo{volume}{16}, \bibinfo{number}{8/9} (\bibinfo{year}{1992}),
  \bibinfo{pages}{27--40}.
\newblock


\bibitem[\protect\citeauthoryear{Boehmer, Bredereck, Heeger, and
  Niedermeier}{Boehmer et~al\mbox{.}}{2021}]%
        {Boehmer21_StableMarriage}
\bibfield{author}{\bibinfo{person}{Niclas Boehmer}, \bibinfo{person}{Robert
  Bredereck}, \bibinfo{person}{Klaus Heeger}, {and} \bibinfo{person}{Rolf
  Niedermeier}.} \bibinfo{year}{2021}\natexlab{}.
\newblock \showarticletitle{Bribery and Control in Stable Marriage}.
\newblock \bibinfo{journal}{\emph{Journal of Artificial Intelligence Research}}
   \bibinfo{volume}{71} (\bibinfo{year}{2021}), \bibinfo{pages}{993--1048}.
\newblock


\bibitem[\protect\citeauthoryear{Bogomolnaia and Jackson}{Bogomolnaia and
  Jackson}{2002}]%
        {Bogomolnaia02_NS_Symmetric}
\bibfield{author}{\bibinfo{person}{Anna Bogomolnaia} {and}
  \bibinfo{person}{Matthew~O. Jackson}.} \bibinfo{year}{2002}\natexlab{}.
\newblock \showarticletitle{The Stability of Hedonic Coalition Structures}.
\newblock \bibinfo{journal}{\emph{Games Economic Behavior}}
  \bibinfo{volume}{38}, \bibinfo{number}{2} (\bibinfo{year}{2002}),
  \bibinfo{pages}{201--230}.
\newblock


\bibitem[\protect\citeauthoryear{Brandt, Bullinger, and Tappe}{Brandt
  et~al\mbox{.}}{2024}]%
        {Brandt24_AI_NashStable}
\bibfield{author}{\bibinfo{person}{Felix Brandt}, \bibinfo{person}{Martin
  Bullinger}, {and} \bibinfo{person}{Leo Tappe}.}
  \bibinfo{year}{2024}\natexlab{}.
\newblock \showarticletitle{Stability based on single-agent deviations in
  additively separable hedonic games}.
\newblock \bibinfo{journal}{\emph{Artificial Intelligence}}
  \bibinfo{volume}{334} (\bibinfo{year}{2024}), \bibinfo{pages}{104160}.
\newblock


\bibitem[\protect\citeauthoryear{Bullinger, Elkind, and Rothe}{Bullinger
  et~al\mbox{.}}{2016}]%
        {BER2016}
\bibfield{author}{\bibinfo{person}{Martin Bullinger}, \bibinfo{person}{Edith
  Elkind}, {and} \bibinfo{person}{J{\"{o}}rg Rothe}.}
  \bibinfo{year}{2016}\natexlab{}.
\newblock \showarticletitle{Cooperative Game Theory}.
\newblock In \bibinfo{booktitle}{\emph{Economics and Computation, An
  Introduction to Algorithmic Game Theory, Computational Social Choice, and
  Fair Division}}, \bibfield{editor}{\bibinfo{person}{J{\"{o}}rg Rothe}} (Ed.).
  \bibinfo{publisher}{Springer}, \bibinfo{pages}{135--193}.
\newblock


\bibitem[\protect\citeauthoryear{Cechl{\'a}rov{\'a} and
  Romero-Medina}{Cechl{\'a}rov{\'a} and Romero-Medina}{2001}]%
        {cechlarova2001stability}
\bibfield{author}{\bibinfo{person}{Katar\'ina Cechl{\'a}rov{\'a}} {and}
  \bibinfo{person}{Antonio Romero-Medina}.} \bibinfo{year}{2001}\natexlab{}.
\newblock \showarticletitle{Stability in coalition formation games}.
\newblock \bibinfo{journal}{\emph{International Journal of Game Theory}}
  \bibinfo{volume}{29} (\bibinfo{year}{2001}), \bibinfo{pages}{487--494}.
\newblock


\bibitem[\protect\citeauthoryear{Chalkiadakis, Elkind, and
  Wooldridge}{Chalkiadakis et~al\mbox{.}}{2011}]%
        {CEW2011CooperativeGames}
\bibfield{author}{\bibinfo{person}{Georgios Chalkiadakis},
  \bibinfo{person}{Edith Elkind}, {and} \bibinfo{person}{Michael~J.
  Wooldridge}.} \bibinfo{year}{2011}\natexlab{}.
\newblock \bibinfo{booktitle}{\emph{Computational Aspects of Cooperative Game
  Theory}}.
\newblock \bibinfo{publisher}{Morgan {\&} Claypool Publishers}.
\newblock


\bibitem[\protect\citeauthoryear{Chen, Cs{\'a}ji, Roy, and Simola}{Chen
  et~al\mbox{.}}{2023}]%
        {CheCsaRoySim2023Verif}
\bibfield{author}{\bibinfo{person}{Jiehua Chen}, \bibinfo{person}{Gergely
  Cs{\'a}ji}, \bibinfo{person}{Sanjukta Roy}, {and} \bibinfo{person}{Sofia
  Simola}.} \bibinfo{year}{2023}\natexlab{}.
\newblock \showarticletitle{Hedonic Games With Friends, Enemies, and Neutrals:
  {R}esolving Open Questions and Fine-Grained Complexity (AAMAS 2023)}.
  \bibinfo{pages}{251--259}.
\newblock


\bibitem[\protect\citeauthoryear{Chen, Hatschka, and Simola}{Chen
  et~al\mbox{.}}{2025a}]%
        {CHS2025FPTCOMSOC}
\bibfield{author}{\bibinfo{person}{Jiehua Chen}, \bibinfo{person}{Christian
  Hatschka}, {and} \bibinfo{person}{Sofia Simola}.}
  \bibinfo{year}{2025}\natexlab{a}.
\newblock \showarticletitle{Computational {S}ocial {C}hoice: {P}arameterized
  Complexity and Challenges}.
\newblock \bibinfo{journal}{\emph{Computer Science Review}}
  (\bibinfo{year}{2025}).
\newblock
\newblock
\shownote{To appear.}


\bibitem[\protect\citeauthoryear{Chen, Kaczmarek, N\"{u}sken, Rothe, Schlotter,
  and Seeger}{Chen et~al\mbox{.}}{2025b}]%
        {che-kac-nue-rot-sch-see:c:control-in-computational-social-choice}
\bibfield{author}{\bibinfo{person}{Jiehua Chen}, \bibinfo{person}{Joanna
  Kaczmarek}, \bibinfo{person}{Paul N\"{u}sken}, \bibinfo{person}{J{\"o}rg
  Rothe}, \bibinfo{person}{Ildik{\'o} Schlotter}, {and} \bibinfo{person}{Tessa
  Seeger}.} \bibinfo{year}{2025}\natexlab{b}.
\newblock \showarticletitle{Control in Computational Social Choice}. In
  \bibinfo{booktitle}{\emph{Proceedings of the 34th International Joint
  Conference on Artificial Intelligence (IJCAI 2025)}}.
  \bibinfo{pages}{10391--10399}.
\newblock


\bibitem[\protect\citeauthoryear{Chen and Schlotter}{Chen and
  Schlotter}{2025}]%
        {Chen25_StableMatchControl}
\bibfield{author}{\bibinfo{person}{Jiehua Chen} {and}
  \bibinfo{person}{Ildik{\'o} Schlotter}.} \bibinfo{year}{2025}\natexlab{}.
\newblock \bibinfo{booktitle}{\emph{Control in Stable Marriage and Stable
  Roommates: Complexity and Algorithms}}.
\newblock \bibinfo{type}{{T}echnical {R}eport}.
  \bibinfo{institution}{arXiv:2502.01215}.
\newblock


\bibitem[\protect\citeauthoryear{Cygan, Fomin, Kowalik, Lokshtanov, Marx,
  Pilipczuk, Pilipczuk, and Saurabh}{Cygan et~al\mbox{.}}{2015}]%
        {CyFoKoLoMaPiPiSa2015}
\bibfield{author}{\bibinfo{person}{Marek Cygan}, \bibinfo{person}{Fedor~V.
  Fomin}, \bibinfo{person}{Lukasz Kowalik}, \bibinfo{person}{Daniel
  Lokshtanov}, \bibinfo{person}{D{\'a}niel Marx}, \bibinfo{person}{Marcin
  Pilipczuk}, \bibinfo{person}{Michal Pilipczuk}, {and} \bibinfo{person}{Saket
  Saurabh}.} \bibinfo{year}{2015}\natexlab{}.
\newblock \bibinfo{booktitle}{\emph{Parameterized Algorithms}}.
\newblock \bibinfo{publisher}{Springer}.
\newblock


\bibitem[\protect\citeauthoryear{Dimitrov, Borm, Hendrickx, and Sung}{Dimitrov
  et~al\mbox{.}}{2006}]%
        {dimitrov2006simple}
\bibfield{author}{\bibinfo{person}{Dinko Dimitrov}, \bibinfo{person}{Peter
  Borm}, \bibinfo{person}{Ruud Hendrickx}, {and} \bibinfo{person}{Shao~Chin
  Sung}.} \bibinfo{year}{2006}\natexlab{}.
\newblock \showarticletitle{Simple priorities and core stability in hedonic
  games}.
\newblock \bibinfo{journal}{\emph{Social Choice and Welfare}}
  \bibinfo{volume}{26}, \bibinfo{number}{2} (\bibinfo{year}{2006}),
  \bibinfo{pages}{421--433}.
\newblock


\bibitem[\protect\citeauthoryear{Downey and Fellows}{Downey and
  Fellows}{2013}]%
        {DF13}
\bibfield{author}{\bibinfo{person}{Rodney~G. Downey} {and}
  \bibinfo{person}{Michael~R. Fellows}.} \bibinfo{year}{2013}\natexlab{}.
\newblock \bibinfo{booktitle}{\emph{Fundamentals of Parameterized Complexity}}.
\newblock \bibinfo{publisher}{Springer}.
\newblock


\bibitem[\protect\citeauthoryear{Dr{\`{e}}ze and Greenberg}{Dr{\`{e}}ze and
  Greenberg}{1980}]%
        {Dreze80_Hedonic}
\bibfield{author}{\bibinfo{person}{Jacques~H. Dr{\`{e}}ze} {and}
  \bibinfo{person}{Joseph Greenberg}.} \bibinfo{year}{1980}\natexlab{}.
\newblock \showarticletitle{Hedonic Coalitions: {O}ptimality and Stability}.
\newblock \bibinfo{journal}{\emph{Econometrica}} \bibinfo{volume}{48},
  \bibinfo{number}{4} (\bibinfo{year}{1980}), \bibinfo{pages}{987--1003}.
\newblock


\bibitem[\protect\citeauthoryear{Faliszewski and Rothe}{Faliszewski and
  Rothe}{2016}]%
        {fal-rot:b:handbook-comsoc-control-and-bribery}
\bibfield{author}{\bibinfo{person}{Piotr Faliszewski} {and}
  \bibinfo{person}{J{\"o}rg Rothe}.} \bibinfo{year}{2016}\natexlab{}.
\newblock \showarticletitle{Control and Bribery in Voting}.
\newblock In \bibinfo{booktitle}{\emph{Handbook of Computational Social
  Choice}}, \bibfield{editor}{\bibinfo{person}{F.~Brandt},
  \bibinfo{person}{V.~Conitzer}, \bibinfo{person}{U.~Endriss},
  \bibinfo{person}{J.~Lang}, {and} \bibinfo{person}{A.~Procaccia}} (Eds.).
  \bibinfo{publisher}{Cambridge University Press}, Chapter~7,
  \bibinfo{pages}{146--168}.
\newblock


\bibitem[\protect\citeauthoryear{Feldman and Ruhl}{Feldman and Ruhl}{2006}]%
        {FeldmanRuhl_SCSS}
\bibfield{author}{\bibinfo{person}{Jon Feldman} {and} \bibinfo{person}{Matthias
  Ruhl}.} \bibinfo{year}{2006}\natexlab{}.
\newblock \showarticletitle{The Directed Steiner Network Problem is Tractable
  for a Constant Number of Terminals}.
\newblock \bibinfo{journal}{\emph{SIAM J. Comput.}} \bibinfo{volume}{36},
  \bibinfo{number}{2} (\bibinfo{year}{2006}), \bibinfo{pages}{543--561}.
\newblock


\bibitem[\protect\citeauthoryear{Floyd}{Floyd}{1962}]%
        {Floyd_Warshall}
\bibfield{author}{\bibinfo{person}{Robert~W. Floyd}.}
  \bibinfo{year}{1962}\natexlab{}.
\newblock \showarticletitle{Algorithm 97: Shortest path}.
\newblock \bibinfo{journal}{\emph{Commun. ACM}} \bibinfo{volume}{5},
  \bibinfo{number}{6} (\bibinfo{date}{June} \bibinfo{year}{1962}),
  \bibinfo{pages}{345}.
\newblock
\showISSN{0001-0782}


\bibitem[\protect\citeauthoryear{Galbraith}{Galbraith}{2012}]%
        {Galbraith2012MPKC}
\bibfield{author}{\bibinfo{person}{Steven~D. Galbraith}.}
  \bibinfo{year}{2012}\natexlab{}.
\newblock \bibinfo{booktitle}{\emph{Mathematics of Public Key Cryptography}}.
\newblock \bibinfo{publisher}{Cambridge University Press}.
\newblock
\showISBNx{978-1107013926}


\bibitem[\protect\citeauthoryear{Gonzalez}{Gonzalez}{1985}]%
        {Gonzalez1985}
\bibfield{author}{\bibinfo{person}{Teofilo~F. Gonzalez}.}
  \bibinfo{year}{1985}\natexlab{}.
\newblock \showarticletitle{Clustering to Minimize the Maximum Intercluster
  Distance}.
\newblock \bibinfo{journal}{\emph{Theoretical Computer Science}}
  \bibinfo{volume}{38} (\bibinfo{year}{1985}), \bibinfo{pages}{293--306}.
\newblock


\bibitem[\protect\citeauthoryear{Hemaspaandra, Hemaspaandra, and
  Rothe}{Hemaspaandra et~al\mbox{.}}{2007}]%
        {Hemaspaandra07_Preclude}
\bibfield{author}{\bibinfo{person}{Edith Hemaspaandra},
  \bibinfo{person}{Lane~A. Hemaspaandra}, {and} \bibinfo{person}{J{\"o}rg
  Rothe}.} \bibinfo{year}{2007}\natexlab{}.
\newblock \showarticletitle{Anyone but Him: {T}he Complexity of Precluding an
  Alternative}.
\newblock \bibinfo{journal}{\emph{Artificial Intelligence}}
  \bibinfo{volume}{171}, \bibinfo{number}{5-6} (\bibinfo{year}{2007}),
  \bibinfo{pages}{255--285}.
\newblock


\bibitem[\protect\citeauthoryear{Li, McCormick, and Simchi-Levi}{Li
  et~al\mbox{.}}{1992}]%
        {li1992point}
\bibfield{author}{\bibinfo{person}{Chung-Lun Li}, \bibinfo{person}{S~Thomas
  McCormick}, {and} \bibinfo{person}{David Simchi-Levi}.}
  \bibinfo{year}{1992}\natexlab{}.
\newblock \showarticletitle{The point-to-point delivery and connection
  problems: complexity and algorithms}.
\newblock \bibinfo{journal}{\emph{Discrete Applied Mathematics}}
  \bibinfo{volume}{36}, \bibinfo{number}{3} (\bibinfo{year}{1992}),
  \bibinfo{pages}{267--292}.
\newblock


\bibitem[\protect\citeauthoryear{Natu and Shu-Cherng}{Natu and
  Shu-Cherng}{1997}]%
        {natu1997point}
\bibfield{author}{\bibinfo{person}{Madan Natu} {and} \bibinfo{person}{Fang
  Shu-Cherng}.} \bibinfo{year}{1997}\natexlab{}.
\newblock \showarticletitle{The point-to-point connection problem—analysis
  and algorithms}.
\newblock \bibinfo{journal}{\emph{Discrete Applied Mathematics}}
  \bibinfo{volume}{78}, \bibinfo{number}{1-3} (\bibinfo{year}{1997}),
  \bibinfo{pages}{207--226}.
\newblock


\bibitem[\protect\citeauthoryear{Niedermeier}{Niedermeier}{2006}]%
        {Nie06}
\bibfield{author}{\bibinfo{person}{Rolf Niedermeier}.}
  \bibinfo{year}{2006}\natexlab{}.
\newblock \bibinfo{booktitle}{\emph{Invitation to Fixed-Parameter Algorithms}}.
\newblock \bibinfo{publisher}{Oxford University Press}.
\newblock


\bibitem[\protect\citeauthoryear{Peters}{Peters}{2017}]%
        {peters2017precise}
\bibfield{author}{\bibinfo{person}{Dominik Peters}.}
  \bibinfo{year}{2017}\natexlab{}.
\newblock \showarticletitle{Precise complexity of the core in dichotomous and
  additive hedonic games}. In \bibinfo{booktitle}{\emph{Proceedings of the
  International Conference on Algorithmic Decision Theory (ADT 2017)}}.
  Springer, \bibinfo{pages}{214--227}.
\newblock


\bibitem[\protect\citeauthoryear{Sung and Dimitrov}{Sung and Dimitrov}{2010}]%
        {Sung10_Additive}
\bibfield{author}{\bibinfo{person}{Shao{-}Chin Sung} {and}
  \bibinfo{person}{Dinko Dimitrov}.} \bibinfo{year}{2010}\natexlab{}.
\newblock \showarticletitle{Computational Complexity in Additive Hedonic
  Games}.
\newblock \bibinfo{journal}{\emph{European Journal of Operational Research}}
  \bibinfo{volume}{203}, \bibinfo{number}{3} (\bibinfo{year}{2010}),
  \bibinfo{pages}{635--639}.
\newblock


\bibitem[\protect\citeauthoryear{Woeginger}{Woeginger}{2013a}]%
        {Woeginger_CoreSurvey}
\bibfield{author}{\bibinfo{person}{Gerhard~J. Woeginger}.}
  \bibinfo{year}{2013}\natexlab{a}.
\newblock \showarticletitle{Core Stability in Hedonic Coalition Formation}. In
  \bibinfo{booktitle}{\emph{Proceedings of the 39th International Conference on
  Current Trends in Theory and Practice of Computer Science (SOFSEM 2013)}}
  \emph{(\bibinfo{series}{Lecture Notes in Computer Science},
  Vol.~\bibinfo{volume}{7741})}. \bibinfo{pages}{33--50}.
\newblock


\bibitem[\protect\citeauthoryear{Woeginger}{Woeginger}{2013b}]%
        {Woeginger13_CoreAdditiveHG}
\bibfield{author}{\bibinfo{person}{Gerhard~J. Woeginger}.}
  \bibinfo{year}{2013}\natexlab{b}.
\newblock \showarticletitle{A hardness result for core stability in additive
  hedonic games}.
\newblock \bibinfo{journal}{\emph{Mathematical Social Science}}
  \bibinfo{volume}{65}, \bibinfo{number}{2} (\bibinfo{year}{2013}),
  \bibinfo{pages}{101--104}.
\newblock


\end{thebibliography}

\ifcr
\else
\clearpage
\begin{table}[t!]
  \centering
  \Large \textbf{\appendixtitle}
\end{table}
\bigskip

\begin{appendix}
\appendixtext
\end{appendix}
\fi

\end{document}